\documentclass[11pt]{article}

\usepackage{amsmath,amssymb,amsthm,bm}
\usepackage{mathtools}
\usepackage{geometry}
\usepackage{booktabs}
\usepackage{graphicx}
\usepackage{natbib}
\usepackage{hyperref}
\usepackage{subcaption}
\usepackage{float}

\usepackage{amsmath}
\DeclareMathOperator{\Cov}{Cov} % then you can write \Cov(...)

\usepackage{tikz}
\usetikzlibrary{positioning,arrows.meta,bending}
\hypersetup{colorlinks=true, linkcolor=blue, citecolor=blue}
\usepackage{booktabs,threeparttable}
\usepackage{siunitx} 
\usepackage{multirow}
\usepackage{titlesec}
\usepackage{chngcntr}

\newcommand{\peer}{\bar c}                  % peer mean
\newcommand{\knorm}{k_{\text{norm}}}        % curvature in norm penalty
\newcommand{\kfermi}{\eta}                  % Fermi intensity
\newcommand{\dtilt}{\Delta}                 % binary fitness tilt (High vs Low)
\newcommand{\wutil}{w^{U}}                  % fitness (utility-based)
\newcommand{\wbin}{w^{B}}                   % fitness (binary-tilt)

\newtheorem{prop}{Proposition}

\title{Group Cooperation Diverges onto Durable Low versus High Paths: Public Goods Experiments in 134 Honduran Villages}

\date{\today}

\author{
  Marios Papamichalis\thanks{Human Nature Lab, Yale University, New Haven, CT 06511, marios.papamichalis@yale.edu} \\
  \and
  Nicholas A. Christakis\thanks{Human Nature Lab, Yale University, New Haven, CT 06511, nicholas.christakis@yale.edu} \\
  \and
  Feng Fu \thanks{Department of Mathematics, Dartmouth College, Hanover, NH 03755, USA, Department of Biomedical Data Science, Geisel School of Medicine at Dartmouth, Lebanon, NH 03756, USA, fufeng@gmail.com}
}

\begin{document}
\maketitle

% ------------------------------------------------------------------ %
%  ABSTRACT                                                          %
% ------------------------------------------------------------------ %

\begin{abstract}

%We performed large, lab-in-the-field experiment (2,591 participants across 134 Honduran villages; ten rounds) and tracked how contribution behavior unfolds in fixed, anonymous groups of size five. Contribution separates early into two durable paths, one low and one high, with rare convergence thereafter. High-path players can be identified with strong accuracy early on. Groups that begin with an early majority of above‑norm contributors (about 60\%) are very likely finish high. The empirical finding of a bifurcation, consistent with the theory, shows that early, high contributions by socially central people steer groups onto, and help keep them on, a high-cooperation path. 

We performed large lab-in-the-field experiment (2,591 participants across 134 Honduran villages; ten rounds) and tracked how contribution behavior unfolds in fixed, anonymous groups of size five. Contributions separates early into two durable paths, one low and one high, with rare convergence thereafter. High-path players can be identified with strong accuracy early on. Groups that begin with an early majority of above‑norm contributors (about 60\%) are very likely finish to high. The empirical finding of a bifurcation, consistent with the theory, shows that early, high contributions by socially central people steer groups onto, and help keep them on, a high-cooperation path.

\end{abstract}

% ------------------------------------------------------------------ %
%  INTRODUCTION                                                      %
% ------------------------------------------------------------------ %

\section{Introduction}

Why do some communities sustain costly cooperation while others succumb to chronic free-riding \cite{ostrom1990governing}? Local public goods (like clean water systems, school maintenance, road repairs, and vaccination drives) depend on repeated voluntary contributions that are individually costly and thus hard to explain with standard self-interest. Laboratory public goods games (PGGs) typically show a decline in average contributions over time. Such experiments (e.g. \cite{henrich2010markets}) also reveal striking and persistent heterogeneity in behavior: some individuals never contribute, some contribute consistently, and many settle into intermediate patterns. This variation often appears trait-like across contexts, in a “cooperative phenotype” \cite{peysakhovich2014humans}. Here, we examine whether persistent differences in cooperation across otherwise similar rural villages reflect individuals’ and villages’ characteristics.

We study repeated public-goods decisions by 2,591 villagers in 134 Honduran villages, each observed over ten rounds of a standardized PGG. Each participant’s contribution decisions are linked to detailed measures of their social ties (\emph{friendship} networks within the village) and their individual attributes (gender, religion and education). The stakes in the game are economically meaningful in this low-income setting: rural incomes are about \$2 per day ($\approx$ 50 Honduran Lempiras), and total session payouts ranged from minimum \$3.75 to maximum \$11.09  ($\approx$ 98 to 290 Lempiras), with per-round endowments up to 12 Lempiras.

Lab PGGs typically involve low stakes (e.g., \citealp{fehr2000cooperation}), and simply raising stakes does not, on average, increase cooperation \citep{kocher2008does}. In high-salience settings, however, larger stakes combined with informational or comparative feedback can boost contributions, see \citet{liu2024signaling} (China), \citet{cardenas2015between} (Colombia), and \citet{nhim2023water} (Cambodia). In this spirit, our high-salience field setting, substantial stakes relative to local incomes, plus repeated feedback, provide a demanding test of how peer exposure and individual traits shape contribution dynamics. Our experimental design and panel data structure allow us to identify the effect of peers and individual traits on contribution dynamics. In particular, we exploit the panel nature of the game with ten rounds of observations and use specifications with individual fixed effects and village-by-round fixed effects. By incorporating leave-one-out lagged peer contribution averages, we can separate a player’s responsiveness to peers from their own persistent cooperativeness, mitigating concerns about reflection or unobserved heterogeneity in group composition. This approach leverages within-person changes over time and across-round exposure to peer behavior to isolate the drivers of sustained cooperation.

We have three primary findings. First, across villages, contribution trajectories \emph{bifurcate} into two distinct and durable paths, high versus low cooperation, rather than converging to a single interior equilibrium. A small early advantage in a person’s peer environment (i.e., being surrounded by others who contribute in the initial rounds) is enough to push that individual onto a high-contribution path; lacking that, behavior tends to settle into persistently low contributions. At the village level, the probability of “finishing high” (i.e. sustaining a high-cooperation group by the final round) rises sharply once the early share of high contributors exceeds a threshold, with a data-driven critical mass around 60\% of players contributing above average in the first few rounds. This pattern is indicative of strategic complementarities or positive feedback: early cooperation begets more cooperation, locking the group into a high-contribution equilibrium, whereas an early deficit of cooperators leads to a low-contribution trap.

Second, network attributes, particularly the number of friendship ties measured prior to the game, but not adversarial ties, strongly predict entry into and persistence on the high-cooperation path, such that each additional friend is associated with higher average contributions and individuals who are more embedded (and less enmeshed in conflicts) are more likely to sustain cooperation. Personal attributes also matter: men contribute about 0.6 Lempiras more per round than women on average, and more schooling is positively associated with giving (modestly higher average contributions across schooling levels). Cultural identity plays a role as well; relative to Catholics (the majority religion in our sample), non-religious villagers and Protestants contribute about 0.4 and 0.2 additional Lempiras per round, respectively. Finally, age by itself has at most a limited effect on cooperation once we account for social ties. %suggesting that it is an individual’s social embeddedness and role, rather than just their age, that drives sustained contributions. %These patterns emerge in panel regression estimates of round-by-round giving, and they are echoed in extreme outcomes: having more friends or living in a village with a denser friendship network significantly lowers the odds of being a complete free-rider (contributing zero in all rounds), while having adversarial ties increases that risk. 

Third, our evidence points to social embeddedness, rather than raw material wealth, as a key driver of sustained cooperation. Neither baseline household wealth nor food security status robustly predicts contributions after controlling for social network position. In contrast, the measure of embeddedness remains a strong predictor, implying that individuals' position in the social fabric matters more than income per se for explaining cooperative behavior. Consistent with this interpretation, we find that men’s higher contributions are not explained by their earnings or wealth (which in our sample are not significantly higher than women’s, on average).

%The separation into high- and low-cooperation paths is visible under a range of reasonable “high” cutoffs and emerges early: using only the first to third rounds predicts final paths with Area Under the Curve (AUC) $\approx 0.829$, indicating high predictive accuracy. 

These patterns are not driven by any single threshold choice or by terminal-round outcomes. 
Dynamics are path-dependent: downward moves from high to low (drop $H\!\to\!L$) are more common than rebounds from low to high (rise $L\!\to\!H$). Friendship ties buffer against $H\!\to\!L$ drops, and reporting no religion is associated with $L\!\to\!H$ rises. Round-to-round updating places greater weight on others’ prior contributions than on one’s own past choice; higher education weakens simple repetition and increases adjustment toward the group, and men place more weight on peers. These qualitative patterns are robust across specifications, and a variance decomposition indicates that most residual variation lies at the individual level, with little remaining structure at the village level after accounting for observed covariates. %, and switching concentrates among more-educated men, whereas low-state players tend to persist. 

Group members choose simultaneously each round and cannot observe others’ choices within the round, so within‑round imitation is not possible. Once we absorb player fixed effects and village$\times$round shocks, contemporaneous peer co‑movement largely vanishes; using a strong \emph{internal deeper‑lag} instrument, together with an over‑identified leave‑one‑village shift–share and a cross‑fitted optimal‑IV construction, we estimate a large \emph{lagged} peer effect. 

Using instrumental variables designs that leverage as‑if random group composition and a deeper‑lag internal instrument while netting out player and village$\times$round shocks, we find that the instantaneous (same‑round) peer effect is near zero, whereas the \emph{lagged} peer effect is approximately one‑for‑one. In short, exogenous variation in last round’s peer environment shifts current giving almost one‑for‑one, consistent with norm updating rather than within‑round imitation (Section~\ref{sec:iv_results}, Table~\ref{tab:iv_updated}). Taken together, the findings suggest actionable levers: \emph{seeding} early cooperation via socially central individuals can steer communities toward sustained collective action.

{\bf Related Literature}: Our study contributes to and bridges several strands of literature. Repeated public-goods experiments in the lab reliably find that average contributions decline over time, yet this aggregate trend masks persistent individual heterogeneity in cooperative “types.”\footnote{See, among others, \citet{keser1999strategic,andreoni2001fair,chaudhuri2010does}, who document declining mean contributions in repeated PGGs, and \citet{fischbacher2001people}, who first showed that many subjects are \emph{conditional cooperators} (willing to match others’ contributions) while others are free-riders or unconditional cooperators. Subsequent work classifies participants into types whose proportions drive group outcomes \citep{fischbacher2006heterogeneous,kurzban2005experiments,bruhin2019many,rand2011dynamic}.} The coexistence of free-riders and conditional cooperators can sustain cooperation at an intermediate level in some groups, rather than leading to the complete unraveling predicted by standard theory. Building on this insight, evolutionary game models and network theories have demonstrated how strategic diversity and positive feedbacks can generate multiple stable equilibria in cooperation. A range of theoretical studies show that under plausible conditions, populations can bifurcate into high- and low-cooperation steady states via mechanisms such as threshold effects or collective-risk dilemmas, nonlinear payoff structures that reward clustering of cooperators, partner selection or network dynamics that allow cooperators to interact, and other eco-evolutionary feedback loops.\footnote{See, e.g., \citet{hauert2008ecological,fu2009partner,wang2009emergence} on threshold public-goods scenarios and collective-risk games leading to cooperation thresholds and tipping points; \citet{zhang2013tale,wu2014social} on nonlinear public goods payoff functions enabling coexistence of cooperators and defectors; \citet{doebeli2004evolutionary} on partner choice creating stable cooperator–defector pairings; and \citet{wang2010evolutionary,wang2020steering} on evolutionary dynamics in structured populations that allow persistent strategy diversity.} Our empirical finding of a bifurcation into high vs. low contribution paths in village communities is consistent with these theoretical predictions of multiple self-reinforcing equilibria in cooperative behavior.

Field and network experiments show that social ties shape cooperation. In public-goods
settings, friendship links increase joint provision \citep{haan2006friendship}, and the
\emph{quality} of ties, who you are connected to, can matter more than the sheer number
of links \citep{shirado2013quality}. These patterns align with our result that more friends
predict entry into and persistence on the high-cooperation path. Experimental evidence also documents history-dependence in contribution dynamics,
where early information regimes and feedback induce persistent differences in later rounds
\citep{fellner2021information}, consistent with the early sorting and durable regimes we observe.

Meta-analyses and field evidence suggest systematic gender differences in cooperative
preferences and behavior \citep{furtner2021gender,seguino1996gender}, echoing our
finding that men contribute more on average and are more responsive to peers. Identity primes and religious affiliation shape economic choices, including prosociality \citep{benjamin2016religious}; in our data, non-religious participants—and, more modestly, Protestants relative to Catholics—contribute more on average, and reporting no religion predicts upward movement from low to high. Higher socioeconomic status does not uniformly reduce prosocial behavior in field settings \citep{andreoni2021higher}; in developing-country contexts, real-stakes cooperation often hinges on social capital rather than income per se \citep{cardenas2000real}. These patterns dovetail with our result that education and embeddedness, not baseline wealth, best predict sustained cooperation.

Further empirical work in field settings likewise underscores the role of social structure and heterogeneity in sustaining collective action. Cooperation does not arise in a vacuum: local leadership, enforcement mechanisms, and network connections have been shown to facilitate cooperative behavior in real communities. For example, in Ugandan producer cooperatives, groups with democratically elected monitors (with sanctioning authority) achieved higher contributions than groups with randomly assigned monitors \citep{baldassarri2011centralized,grossman2014impact}. In Ethiopian forest user groups, contributions and rule compliance were sustained where communities invested in costly monitoring and where a higher share of members behaved as conditional cooperators \citep{rustagi2010conditional}. Complementing this, a field experiment shows that when leaders are endowed with punishment power, cooperation improves when sanctions are applied fairly and efficiently, whereas antisocial leader punishment undermines outcomes \citep{kosfeld2015leader}. Similarly, even in the absence of formal enforcement, cooperation has been observed to cluster within social networks, for instance, food-sharing and camp cooperation decisions among Hadza hunter–gatherers tend to be correlated within friendship ties, suggesting that more cooperative individuals bond together \citep{apicella2012social}. Our results echo these patterns: we find that villagers with more friends (and fewer adversaries) are significantly more likely to contribute and keep contributing, indicating that social embeddedness and peer influence are critical for overcoming free-rider incentives in a community context.

Standard theory predicts full free-riding in voluntary public good provision \citep{bernheim1986voluntary,brennan1985private}, yet laboratory experiments consistently find that subjects initially contribute a substantial share of their endowment \citep{isaac1988group,andreoni1995cooperation}. Absent intervention, however, contributions typically decline over repeated rounds, prompting debate about the source of early cooperation. \citet{andreoni1995cooperation} estimates that roughly half of initial contributions are intentional, arising from prosocial kindness or warm-glow preferences rather than confusion, whereas \citet{palfrey1997anomalous} conclude that altruism plays little role once decision errors are accounted for. Indeed, contributors seem to derive private utility from giving: government provision crowds out private donations only partially, consistent with impure altruism \citep{andreoni1993experimental}. Behavior also responds to framing and context. Simply framing the decision as contributing to a public good (rather than an economically equivalent taking frame) significantly boosts contributions \citep{andreoni1995warm}, and making contributions public instead of confidential increases giving due to social pressure \citep{daughety2010public}. To rationalize such patterns, economists have developed models of social preferences that incorporate fairness, reciprocity, and inequality aversion \citep{fehr1999theory,charness2002understanding}. Still, without mechanisms to reinforce cooperation, voluntary contributions remain fragile. Even non-pecuniary sanctions like disapproval or shaming can sustain higher contributions by establishing norms of compliance \citep{masclet2003monetary}. In fact, groups will endogenously establish institutions to promote cooperation: given the opportunity, participants often vote to implement sanctioning systems or binding agreements, which then mitigate free-riding \citep{kosfeld2009institution}. Beyond such formal mechanisms, reputational incentives also facilitate cooperation. For example, sellers who commit to donate a portion of proceeds in online marketplaces obtain higher prices, indicating that buyers reward prosocial behavior \citep{elfenbein2010greater}. Together, these findings illustrate that appropriately designed incentives, monetary or social, can substantially improve the provision of public goods in experimental settings.

Heterogeneity among individuals and groups leads to markedly different cooperation outcomes in public goods games. Numerous studies document diversity in people's willingness to cooperate and their sensitivity to others' behavior. \citet{fischbacher2010social} identify a prevalent type of conditional cooperator, individuals who contribute more when they believe others will contribute, alongside a minority of free-riders. Such heterogeneity implies that peer effects can produce virtuous or vicious cycles: groups with a higher share of conditional cooperators (or strong reciprocators) tend to sustain cooperation through positive reciprocity, whereas those dominated by selfish types are prone to unraveling as defection begets defection. Social identity and network structure further mediate these dynamics. In field settings, more homogeneous communities or tightly knit groups often achieve better public goods provision, whereas diverse or fragmented groups struggle to cooperate \citep{beach2017gridlock}. Likewise, who interacts with whom can shape norm formation: if individuals can choose partners, they often self-sort into cooperative clusters, effectively isolating non-cooperators \citep{kinateder2017public}. Certain individuals in key network positions or leadership roles may disproportionately influence outcomes, and effective communication can help coordinate expectations. In particular, delegating a mediator to structure group communication has been shown to improve coordination in public good provision by resolving strategic uncertainty \citep{halonen2021coordinating}. Information conditions are also critical: when the value of the public good or others' contributions is uncertain, even well-intentioned players may hesitate to contribute, creating an informational social dilemma \citep{cox2021common}. Finally, the timing of moves introduces additional strategic complexities. In dynamic contribution games, people often delay or withhold contributions in hopes that others will pay the cost, leading to inefficient under-provision or costly bargaining delays \citep{bowen2019collective}. Experiments indicate that many subjects struggle with contingent reasoning about others' future actions, which can trap groups in suboptimal contribution paths \citep{calford2024contingent}. Taken together, this evidence underscores that seemingly small differences in group composition, incentives, or information can cause groups to diverge onto very different cooperation trajectories. Some groups manage to establish self-reinforcing norms of high cooperation, through reciprocity, peer influence, and effective sanctioning, while others languish in persistent free-riding equilibria. This bifurcation into high- and low-contribution regimes is consistent with the two durable paths of group cooperation observed in our Honduras setting.

Our study also builds on a growing body of research mapping entire social networks in developing-country villages and experimenting with network-based interventions (\cite{perkins2015social}). For example, seeding central individuals markedly increased adoption of health innovation \citep{kim2015social}; prosocial behaviors cascade through ties \citep{fowler2010cooperative,jordan2013contagion,airoldi2024induction}; antagonistic ties transmit information across clusters \citep{ghasemian2024structure}; and network position correlates with well-being and perceptual biases \citep{lee2024depression,rodriguez2024forgetting}. Complementary evidence shows that changing interaction structure can shift cooperation, via artificial agents or partner choice or that information provision can re-wire social ties \citep{shirado2020collective,rand2011dynamic,papamichalis2025educational}.

With 2{,}591 participants across 134 villages, our experiment is among the largest \emph{lab-in-the-field} cooperation experiments to date. For scale, landmark work in Uganda ran in 50 producer co-ops with 1{,}543 farmers and tested centralized sanctioning and leader legitimacy \citep{baldassarri2011centralized,grossman2014impact,baldassarri2015cooperative}. Post-conflict field studies typically span dozens of communities, e.g., Burundi (35 villages; $\sim$300 participants) and Sierra Leone (200 villages; $\sim$2{,}300 participants) where contributions to local public goods were measured within broader interventions \citep{voors2012violent,cilliers2016reconciling}. A recent Rwanda initiative conducted PGGs in 147 villages ($\sim$1{,}764 participants), emphasizing communication spillovers \citep{coutts2024age}. By contrast, very large \emph{online} PGGs can aggregate orders of magnitude more decisions (e.g., $\sim$135{,}000 players; $\sim$1.5M decisions), but they trade off face-to-face logistics and village representativeness \citep{otten2022human}. Our contribution is therefore to couple large numbers of participants with in-person, community-based play, enabling heterogeneity analyses across many villages while preserving field realism.

It is worth noting that our setting represents a particularly high-salience environment for cooperation, due to the combination of low incomes and relatively large stakes. One might suspect that higher monetary stakes would by themselves induce more cooperation out of fear of losing a valuable payoff. However, experimental evidence indicates that simply increasing stake size does not automatically lead to higher contributions (nor higher punishment of free-riders) in public-goods games \citep{kocher2008does}. Instead, the effectiveness of incentives often hinges on contextual factors, such as feedback, peer comparisons, or institutional rules, and these factors tend to have an amplified impact when the stakes are large relative to participants’ incomes. In our experiment, despite the high stakes, we find that it is the social and behavioral context (who your peers are, how you perceive and respond to them, and how much control you have over resources) that chiefly determines cooperation, rather than the absolute monetary value at play. This aligns with the view that enhancing the social incentives and education of participants can be more effective in sustaining public-good contributions than simply raising the financial stakes for contributing.

%In summary, our study makes several contributions. First, we empirically demonstrate a bifurcation in repeated cooperation in the field and show that eventual high- vs.\ low-cooperation paths can be predicted early in the interaction. Second, we quantify how social embeddedness, personal identities, and education shape both average contribution levels and the probability of sustaining high cooperation, thereby linking micro-level traits to macro-level outcomes. Third, we provide evidence of a group-level tipping point (a critical mass of initial cooperators around 60\%) that secures a high-cooperation equilibrium, helping to explain why some communities achieve lasting cooperation while others do not. Fourth, we document path dependence in transitions: downward moves from high to low cooperation are more common than rebounds; switching concentrates among more-educated men, low-state individuals are more persistent, friendship networks buffer against $L\!\to\!H$ drops, and non-religious identity is associated with $L\!\to\!H$ rises. Fifth, we show that round-to-round updating is primarily peer-responsive rather than self-referential: individuals place more weight on others’ recent contributions than on their own past choice; education lowers inertia and increases adjustment toward the group; and men place greater weight on peers’ behavior. Our results are robust across specifications and diagnostics, and a variance decomposition indicates that most residual heterogeneity lies at the individual level rather than the village level.\\

The paper proceeds as follows. Section~2 describes the setting and sample; the mapping of village friendship and adversarial networks; the public–goods protocol and stakes; and the peer-exposure constructs that motivate our empirical tests. Section~3 documents sampling, session logistics, variable construction, and summary statistics. Section~4 formalizes the stage game, derives best-reply behavior, and introduces the evolutionary calibration and identification strategy. Section~5 presents the core evidence: nonparametric tipping, HMM persistence, village critical mass, calibrated Fermi–Moran process, heterogeneity in contributing, trajectory clustering and causal effects through instrumental variables. In addition to comprehensive statistical analyses, we  calibrate a mechanistic evolutionary game model with our data to explain the bifurcation dynamics between high and low contribution behaviors.  Section~6 interprets the findings, notes limitations, and outlines policy levers. 

%switching and learning, , structural back-out, and the , plus early-warning prediction

\section{Setting and Experimental Design}
\label{sec:exp_design}

%This section describes the setting, sample, network measurement, game protocol, stakes, and the peer-exposure constructs that underpin our empirical strategy. We close by outlining simple predictions that motivate the tests that follow.

We implemented a lab-in-the-field study in rural western Honduras (department of Copán), recruiting \emph{2{,}591} adult villagers from \emph{134} isolated villages. Within each village, we ran one or more sessions in a common public space and collected survey data. All procedures were approved by relevant institutional review boards; participation was voluntary; and all payments were made in cash immediately after the session.

Before play, we mapped within-village social networks using standard name-generator questions (see \cite{airoldi2024induction}). Each respondent identified (i) a set of \emph{friends} (e.g. people they spend time with, ask advice of, or would borrow/lend small amounts of money to) and (ii) \emph{adversarial} ties (people they avoid, have conflicts with, or “dislike”). These lists produced, at the individual level, friend degree and adversarial degree, and, at the village level, summary measures (e.g., average friend degree). We measured \emph{education} as years of schooling completed. Catholics constitute the majority of villages, with observed minorities of Protestants and non-religious respondents. We also collected other standard individual-level measures, including age, gender, marital status, education, indigenous status, and food insecurity (as a measure of player wealth). Village-level measures involved remoteness (access routes), village size, friendship density, and adversarial density. %Among these covariates, we find little robust association with contributions for food insecurity, age, marital status, network densities, and village remoteness once social embeddedness and identity are accounted for; by contrast, education and friendship ties show systematic links in the empirical sections that follow.

Sessions consisted of \emph{ten} identical rounds of a linear public-goods game (PGG) played in fixed groups of \emph{five} participants drawn from the same village. In each round, every player received an endowment of \emph{12 Lempiras} and chose how much to contribute to a common pool. The sum of contributions was \emph{doubled} and then evenly redistributed among group members. Choices were recorded privately; at the end of the session, participants were paid the sum of their round-by-round earnings. Participants were not told the total number of rounds ex ante to avoid horizon effects and within each round choices were simultaneous and private; feedback on group outcomes arrived only after the round closed.

%Our setting sits at this high-salience end; below we show that social embeddedness and education, rather than raw wealth, track sustained cooperation. 

Stakes were economically meaningful. Typical rural cash income in this region is about \emph{L\,50} (\$2) per day, while average session payouts ranged from approximately 98 to 290 Lempiras (\$3.75 to \$11.09).\footnote{Amounts refer to the study period and local prices in Copán.} With per-round endowments up to 12 Lempiras, the decision to contribute or free-ride carries nontrivial opportunity costs.

Simultaneous, anonymous play precludes within‑round imitation; peers’ behavior is only observed between rounds. We neither varied nor revealed participants’ wealth during play, consistent with evidence that making wealth visible in networked public-goods settings depresses cooperation and welfare \citep{nishi2015inequality} while villages. While villagers knew they were playing with other villagers, they did not know who was in their own group, we record, for each individual and round, the chosen contribution and contemporaneous group behavior. Our empirical tests rely on \emph{leave-one-out} peer measures that exclude the focal individual from the group’s average, and on lagged peers (previous-round group behavior) to proxy the evolving “local norm.” The panel structure (10 rounds per person) allows us to separate responsiveness to peers from individual inertia using fixed effects and lag structure.

We link each person’s round-by-round contributions to (i) their friendship degree, (ii) identity (gender, religion), and (iii) education. This design lets us ask whether individuals with more friends (and fewer adversarial ties) are more likely to sustain high contributions, and whether identity and education predict entry into and persistence on a high-contribution path, over and above any association with wealth.

%Two empirical regularities motivate our specification choices. First, early-round behavior is informative: the share of high contributors in rounds 1–3 is strongly predictive of end-of-session outcomes at both the individual and village levels. Second, peer exposure is meaningfully heterogeneous across individuals due to group composition and to differences in network position and education. We therefore define outcomes at both the intensive margin (Lempiras contributed per round) and the extensive margin (propensity to free-ride).

Three pragmatic considerations guided the design: a) Group size ($N{=}5$) and repetition ($T{=}10$): Small groups balance tractability and salience: each player internalizes a visible share of the public return, while $T{=}10$ rounds are sufficient to observe learning, persistence, and path formation without inducing fatigue or attrition. b) Fixed groups and private decisions: Keeping peers fixed across rounds isolates within-group dynamics and path dependence; private, simultaneous choices reduce experimenter-demand effects and within-round coordination. c) Measurement aligned to mechanisms: Mapping \emph{friendship} and \emph{adversarial} ties before play lets us test whether embeddedness and frictions predict both entry into and persistence on high-contribution paths; measuring \emph{education} (rather than only wealth) emphasizes informational and learning channels most relevant for repeated updates.

Two sources of variation support our empirical strategy: a) Early peer environments: Initial rounds generate meaningful dispersion in group behavior (shares of high contributors, leave-one-out means). Conditional on \emph{individual} and \emph{village-round} fixed effects, lagged leave-one-out peer averages provide a time-varying exposure measure that helps isolate responsiveness to peers from own inertia. b) Heterogeneous embeddedness and education: Individuals differ in \emph{individual} and village characteristics. These pre-play characteristics create stable cross-sectional gradients that predict exposure to cooperative peers over the session and allow us to study heterogeneity in responsiveness and persistence.

Our setting is designed to mirror everyday local public goods (water system upkeep, school cleanups, road repairs, vaccination drives etc). If social embeddedness and education matter more than wealth, say for repeated giving, then: (i) small early differences in peer environments should sort individuals into distinct contribution “paths,” (ii) friend degree should be positively associated with entering and remaining on a high-contribution path, (iii) identity covariates (gender, religion) and education should have predictive power even after conditioning on wealth proxies, and (iv) wealth per se should have limited explanatory power once we account for (i)–(iii). The empirical sections that follow formalize and test these predictions using panel specifications with individual and village-round fixed effects and leave-one-out lagged peer means.

\section{Data}
\label{sec:data}

\subsection{Honduras Repeated Public Good Games Data}

%We study rural communities in the department of Copán, Honduras, a mountainous region on the Guatemalan border comprising small villages where incomes average roughly \$2 per day (L\,50). 

Our cooperation experiment was conducted between June~2022 and June~2023 in a representative subset of \emph{134} villages drawn from a long-running cohort of \emph{176} mapped villages in the area. Within each of the 134 study villages, we first randomly invited 40 adult residents and, conditional on attendance and consent on the session day, formed groups with 15–20 participants per village. \footnote{All procedures received IRB approval; participants provided written informed consent. Payments and logistics were reviewed for local appropriateness with community leaders.}

Sessions were held in central, accessible locations within villages. Upon arrival and consent, participants were randomly assigned into anonymous groups of $N=5$ and seated so that communication was not possible. All choices were made on Android tablets after a short guided training by research assistants.

Each session consisted of $T=10$ identical rounds of a linear PGG. In each round, every participant received an endowment of 12 Lempiras and chose how much to contribute to a common pool (0–12 L). The pooled contributions were \emph{doubled} and split equally among the five group members. After each round, participants observed the group outcome and their earnings for that round before proceeding to the next decision. Participants were not told the total number of rounds ex ante and could not see or identify their groupmates; the interface displayed only five avatars (including self). After round 10, subjects received their cumulative earnings from the game plus a show-up fee of 50 L (roughly a local day wage).

Random assignment to groups within villages achieved balance on observable traits (gender, age, friends, and religion); see Appendix \ref{tab:contsum}, \ref{tab:bin} and \ref{tab:cat}.

\subsection{Data Heterogeneity}

Across 2{,}591 players and ten rounds, there are 25{,}910 potential player–round observations. Models that use lagged outcomes exclude round~1 by construction; the linear mixed‑effects analyses use 22{,}239 decisions, and instrument variable (IV) specifications use 19{,}818 (contemporaneous peers) and 25{,}388 (lagged peers) observations due to lag structure and occasional incomplete histories. Missingness is very low and not systematically related to baseline gender, age, schooling, or friend counts (see \ref{app:data}).

We link round-by-round contributions to rich individual and village covariates collected in the broader cohort:

\begin{itemize}
  \item \textbf{Demographics and identity:} age, gender, marital status, education, indigenous status, religion (Catholic, Protestant, none).
  \item \textbf{Connections/Ties:} sociocentric mapping of friendship and adversarial ties within the village; from these we construct individual friend counts and village-level friendship densities and network size.
  \item \textbf{Resources:} individual food insecurity.
  \item \textbf{Geography/access:} an index of village isolation (access routes).
\end{itemize}

Table~\ref{tab:het_vars} lists the analysis variables used to explain heterogeneity in cooperation.

\begin{table}[H]
\centering
\begin{tabular}{llp{7cm}}
\toprule
\textbf{Variable} & \textbf{Type} & \textbf{Description / Coding} \\
\midrule
Contributing amount & Dependent & Lempiras contributed each round (0–12).\\
Age & Continuous & Years.\\
Gender & Binary & 1 = male; 0 = female.\\
Friends & Count & Named friends in village network.\\
Adversaries & Count & Named adversaries in village network.\\
Food insecurity & Binary & 1 = insufficient (waves 1–4); 0 = sufficient.\\
Marital status & Binary & 1 = married/civil union; 0 = otherwise.\\
Education & Ordinal & 0 = none … 13 = $>$secondary. \\
Indigenous status & Binary & 1 = Indigenous; 0 = non-indigenous.\\
Religion & Categorical & 2 = Catholic; 1 = Protestant; 0 = none.\\
Access routes & Continuous & 1–5 scale; higher values indicate greater isolation.\\
Friendship density & Continuous & Share of realized friendship ties in village graph.\\
Adversarial density & Continuous & Share of realized adversarial ties.\\
Village size & Count & Number of nodes in village network.\\
\bottomrule
\end{tabular}
\caption{Variables used to explain heterogeneity in contribution behaviour.}
\label{tab:het_vars}
\end{table}

Across $25{,}910$ player--round observations, contribution amounts (0--12 Lempiras) averaged $6.355$ (SD $3.452$). Regarding sample composition, participants were on average $36.79$ years old (SD $15.76$; range $14$--$89$). The typical player named $7.03$ friends (SD $4.51$; range $0$--$35$) and $0.78$ adversaries (SD $1.28$; range $0$--$10$). Men constitute $41.0\%$ ($n{=}1{,}063$) and women $59.0\%$ ($n{=}1{,}528$). Just under half report food insecurity ($43.9\%$, $n{=}1{,}088$), and about two–thirds are married or in a civil union ($66.2\%$, $n{=}1{,}714$). Indigenous identity is reported by $12.8\%$ ($n{=}318$). Educational attainment is concentrated at lower levels: level~0 $25.4\%$ ($n{=}635$), level~1 $7.3\%$ ($n{=}182$), level~2 $13.7\%$ ($n{=}343$), level~3 $16.7\%$ ($n{=}416$), level~4 $8.4\%$ ($n{=}209$), level~5 $4.7\%$ ($n{=}118$), and level~6 $18.2\%$ ($n{=}455$), with higher categories rarer (level~8 $3.8\%$, $n{=}94$; level~11 $1.3\%$, $n{=}33$; level~13 $0.4\%$, $n{=}11$). By religion, Catholics constitute the majority ($57.1\%$, $n{=}1{,}424$), followed by Protestants ($33.7\%$, $n{=}842$) and those with no religion ($9.2\%$, $n{=}230$). Villages had a mean network size of $193.747$ nodes (SD $118.230$; range $54$--$668$), with average friendship density $0.019$ (SD $0.010$; range $0.003$--$0.059$) and adversarial density $0.002$ (SD $0.002$; range $0.000$--$0.010$). The access--routes index averaged $1.902$ (SD $0.714$; range $1$--$5$).

\subsection{Initial Data Analysis}

Figure~\ref{fig:mean_variance} summarizes the evolution of contributing over ten rounds. Mean contributions fall from just over 7 L to roughly 6 L, while cross-sectional variance rises by nearly 25\%. 

Kernel densities (lower-left) reveal a visible split into low ($\leq$6 L) and high ($>$6 L) modes after round~3, foreshadowing the two persistent cooperation paths we estimate below. Despite the widening dispersion, round-to-round correlations remain high ($\approx0.55$–$0.75$), indicating strong individual persistence. Figure~\ref{fig:density} recasts the distributions as counts by contribution band. The 0–3 L group nearly doubles, the 8–10 L group shrinks steadily, and the intermediate (4–7 L) and full (11–12 L) bands remain comparatively flat, visualizing migration toward the low end.

Figure~\ref{fig:Initial_High_Low} tracks types. The numbers of high ($\geq$ 6 L) and low ($<$ 6 L) contributors converge by round~6. Pure free-riding (0 L) rises monotonically, while full contributions (12 L) remain roughly constant. Figure ~\ref{fig:contribution_by_round} depicts counts by round. More specifically, free riders increase over the session (from $165$ in round~1 to a peak of $230$ in round~9, finishing at $224$ in round~10), while full contributors remain relatively stable (roughly $218$--$237$ per round). Low contributors rise from $1{,}064$ (round~1) to $1{,}341$ (round~10), whereas high contributors decline from $1{,}527$ to $1{,}250$ over the same period. The number of low contributors overtakes high contributors by round~8, reflecting the gradual shift in mass toward lower contribution levels across rounds.

%Code Initial Analysis

% data NEW_WORK.R

\begin{figure}[H]
  \centering

  % --- top row ---
  \begin{subfigure}{0.49\textwidth}
    \centering\includegraphics[width=\linewidth]{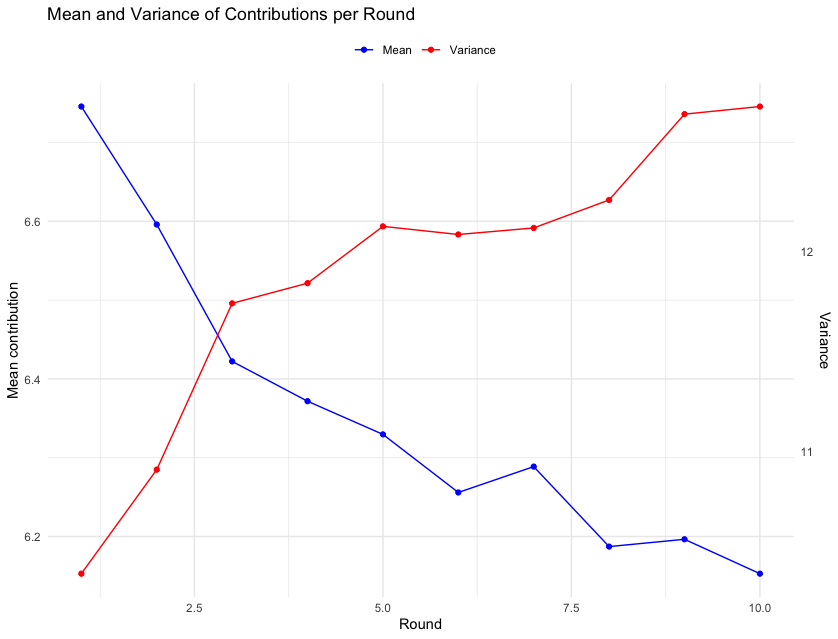}
         \subcaption{Mean and Variance per Round}  
    \label{fig:mean_variance}
  \end{subfigure}\hfill
  \begin{subfigure}{0.49\textwidth} \centering\includegraphics[width=\linewidth]{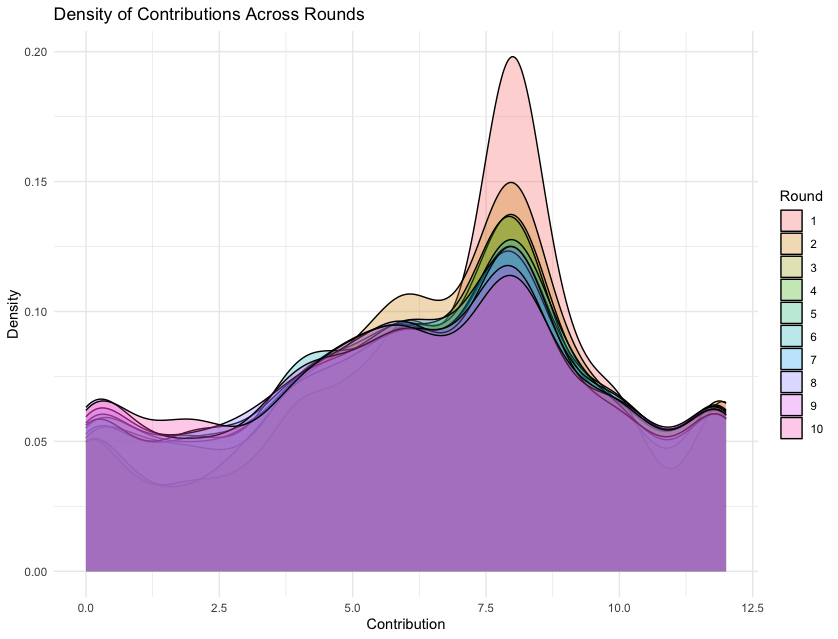}
       \subcaption{Density per Round}  
    \label{fig:density}
  \end{subfigure}

  \medskip

  % --- bottom row ---
  \begin{subfigure}{0.49\textwidth}
 \centering\includegraphics[width=\linewidth]{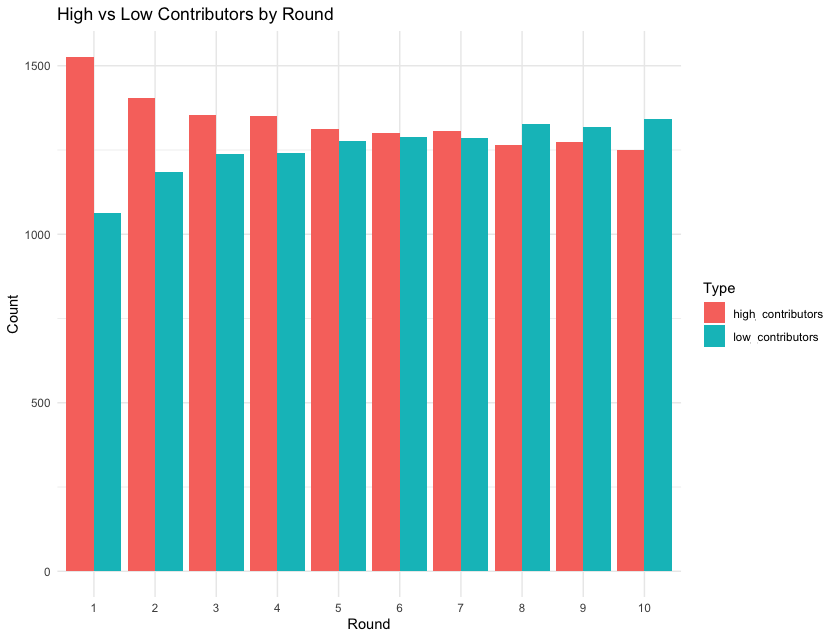}
      \subcaption{High and Low per Round}  
    \label{fig:Initial_High_Low}
      \end{subfigure}\hfill     
  \begin{subfigure}{0.49\textwidth} \centering\includegraphics[width=\linewidth]{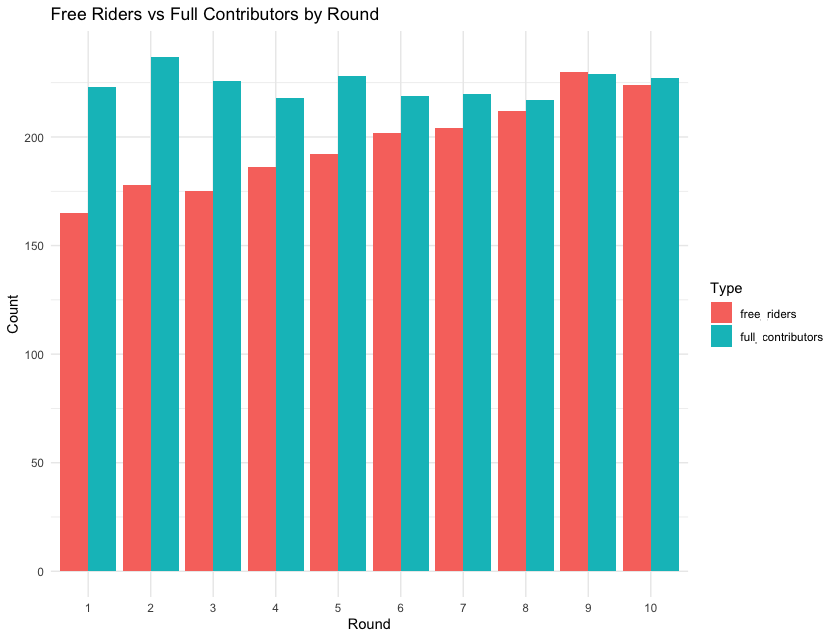}
     \subcaption{Free Riders versus Full Contributors}   \label{fig:contribution_by_round}
  \end{subfigure}

  \caption{\textbf{Initial data analysis.} (a) Mean contributions decline from just over 7\,L to about 6\,L as cross-sectional variance rises across rounds. (b) Round-wise contribution densities become visibly bimodal after round~3, separating into low ($\leq$6\,L) and high ($>$6\,L) modes. (c) Headcounts of High ($\geq$6\,L) and Low ($<$6\,L) contributors converge by round~6. (d) Free riders (0\,L) increase monotonically, whereas full contributors (12\,L) remain comparatively stable over the session.}
  \label{fig:summary_panels_0}
\end{figure}

% -------- Paragraphs for Tables~\ref{tab:bin} and \ref{tab:cat} --------

\section{Behavioral Modeling through Adaptive Dynamics and Evolutionary Game  Theory}
\label{sec:theory}
%========================================================
This section formalizes the repeated public goods game used in the experiment, derives best-reply behavior, and adaptive dynamics, and sets out a parsimonious evolutionary game dynamic that motivates the empirical tests that follow. All objects are defined when first introduced; proofs are sketched in text and provided in the Appendix.

%--------------------------------------------------------
\subsection{Stage Game and Instantaneous Payoff}
\label{sec:game}
%--------------------------------------------------------
\paragraph{Environment.}
Each session consists of $T=10$ identical rounds played in fixed groups of $N=5$ players. In round $t\in\{1,\dots,T\}$:
\begin{enumerate}
  \item Player $i$ chooses a contribution $c_{it}\in[0,12]$ (lempiras).
  \item The group receives a linear public return $b\sum_{j=1}^{N}c_{jt}$, where the marginal benefit $b>0$ is common knowledge.
  \item The return is split equally among the $N$ players; each contributed unit costs its owner $\kappa>0$ ($\kappa = 1$ by default).
\end{enumerate}

\noindent\textbf{Assumption A1 (socially efficient, privately costly).} We impose $b>\kappa$ (social surplus from a unit contribution is positive) and $b/N<\kappa$ (the private marginal gain is below the private marginal cost), so contributing is socially beneficial but individually costly.

\paragraph{Preferences.}
Player $i$’s instantaneous utility in round $t$ has three components:
\begin{itemize}
\item \emph{Material payoff:} 
$\displaystyle
\frac{b}{N}\bigl[c_{it}+(N-1)\bar c_{-i,t}\bigr]-\kappa\,c_{it}$, 
where $\bar c_{-i,t}:=\frac{1}{N-1}\sum_{j\ne i}c_{jt}$ denotes the contemporaneous peer mean.

\item \emph{Altruism:} $d_i\,(c_{it})^{\alpha}$ with individual weight $d_i>0$ and curvature $\alpha\in(0,1)$ (diminishing returns).

\item \emph{Aversion to being at the norm:} $-\,h_i\exp\!\bigl[-\knorm(c_{it}-\bar c_{-i,t-1})^{2}\bigr]$, 
and $h_i\in[0,1]$ measures the salience of anti-conforming to the lagged peer norm 
$\bar c_{-i,t-1}:=\frac{1}{N-1}\sum_{j\ne i}c_{j,t-1}$, and $\knorm>0$ controls curvature.\footnote{Any smooth, bounded \emph{cost} function $P(c_{it}-\bar c_{-i,t-1})$ that attains its maximum at zero deviation (so the utility term $-P(\cdot)$ is minimal there) yields qualitatively similar results.}\footnote{Any smooth, bounded penalty that is maximized at zero deviation yields qualitatively identical results.}
\end{itemize}

\noindent Combining terms yields the instantaneous utility
\begin{equation}
\label{eq:utility}
\pi_i^{t}
=\frac{b}{N}\Bigl[c_{it}+(N-1)\bar c_{-i,t}\Bigr]
-\kappa\,c_{it}
+d_i\bigl(c_{it}\bigr)^{\alpha}
-h_i\exp\!\Bigl[-\knorm\bigl(c_{it}-\bar c_{-i,t-1}\bigr)^{2}\Bigr].
\end{equation}

\paragraph{Baseline curvature.} In both theory and estimation we adopt
\[
\alpha=\tfrac12,
\]
% \[
% \boxed{\alpha=\tfrac12},
% \]
which (i) parsimoniously captures concavity and (ii) delivers closed-form expressions. In SI, we show that all qualitative results survive for $\alpha\in\{0.3,0.7\}$, as a robustness check on $\alpha$.

%--------------------------------------------------------
\subsection{Best-Reply Behavior}
\label{sec:bestreply}
%--------------------------------------------------------
Treating $\bar c_{-i,t-1}$ as given, player $i$ solves $D(c)=\partial\pi_i^{t}/\partial c_{i,t}=0$ over $c\in[0,12]$:
\[
D(c)
=\Bigl(\tfrac{b}{N}-\kappa\Bigr)
+d_i\,\alpha\,c^{\alpha-1}
+2\knorm\,h_i\,(c-\bar c_{-i,t-1})
       \exp\!\bigl[-\knorm(c-\bar c_{-i,t-1})^{2}\bigr]=0.
\]
Because $\alpha<1$, the altruism term is strictly decreasing and concave up in $c$; the aversion to being at the norm term is monotone in $h_i$ and pulls the optimum away $\bar c_{-i,t-1}$. For the parameter range we consider, the best reply is numerically single-valued, continuous, and bounded on $[0,12]$.

\subsection{Adaptive Dynamics}
\label{sec:adaptive}
%--------------------------------------------------------
Assume a mutant $i$ playing $c_m$ ($c_{i,t} = c_m$) and a population playing strategy $c$ (namely $\bar c_{-i,t-1} = c$) for contributing, and then the selection gradient that governs the adaptation of the best possible strategy $c$ is given by
\begin{equation}
\dot{c} = D(c)=\frac{\partial\pi_i^{t}}{\partial c_{i,t}}\left|_{c_{i,t} = c}\right. = \Bigl(\tfrac{b}{N}-\kappa\Bigr)
+d_i\,\alpha\,c^{\alpha-1}
\label{adaptivedyn}
\end{equation}

Solving $D(c) = 0$ we obtain the singular strategy $c^*$.  The stability of $c^*$ is determined by $D'(c^*)$. $c^*$ is convergence stable if $D'(c^*)<0$.

%--------------------------------------------------------
\subsection{Convergence Stability and Evolutionary Branching}
\label{sec:stability}
%--------------------------------------------------------
\begin{prop}[Convergence stability]
\label{prop:css}
Under Assumption~A1 and $\alpha\in(0,1)$, for any $d_i>0$ and $h_i\ge0$, the adaptive dynamics~\ref{adaptivedyn} admit at most one singular strategy:
\[
c^{*}
=\Bigl[\bigl(\kappa-\tfrac{b}{N}\bigr)\big/\bigl(d_i\alpha\bigr)\Bigr]^{\!1/(\alpha-1)}\in(0,12],
\]
which is convergence stable:
\end{prop}
\noindent\emph{Sketch.} With $0 < \alpha<1$, $D(c)> 0$ for $c \to 0$ and $D(c)<0$ as $c \to \infty$; $D$ crosses zero exactly once, yielding one possible  singular strategy $0< c^* \le 12$. Local convergence follows from standard adaptive-dynamics arguments: $c^*$ is an attractor since $D'(c^*) <0$.

\noindent The convergence stable strategy $c^*$ has the following closed form for $\alpha=\tfrac12$.\;
\[
c^{*}
=\left[\frac{0.5\,d_i}{\kappa-b/N}\right]^{\!2}\in(0,12]\quad\text{whenever }b<N\kappa.
\]

\begin{prop}[ESS vs.\ branching]
\label{prop:ess}
The singular strategy $c^{*}$ is convergence stable since $D'(c^{*})<0$. Then $c^{*}$ is an evolutionarily stable strategy (ESS) iff
\[
d_i\,\alpha(\alpha-1)\,(c^{*})^{\alpha-2}+2\knorm\,h_i<0.
\]
If the inequality reverses, $c^{*}$ remains convergence-stable but is not ESS, so evolutionary branching can arise.
\end{prop}
The population initially converges to the singular strategy $c^*$ by following the selection gradient. However, when the ESS condition fails, that is, when
\[
d_i\,\alpha(\alpha - 1)(c^{*})^{\alpha-2} + 2\knorm h_i > 0,
\]
$c^*$ becomes a local fitness minimum, so ``mutants'' close to $c^*$ on both sides can invade, leading to the onset of evolutionary branching.

\subsection{Evolutionary Dynamics}
\label{sec:moran}
%--------------------------------------------------------
To capture long-run selection on behavioral rules, we embed \eqref{eq:utility} in a Moran birth–death process (weak selection) on a large population:
\begin{enumerate}
  \item \emph{Reproduction:} an individual is drawn with probability proportional to $w_i=1+\delta\,\pi_i^{t}$, $\delta\in(0,1)$.
  \item \emph{Replacement:} the offspring replaces a uniformly chosen individual elsewhere.
\end{enumerate}
Let $\hat c$ denote the empirical mean of round-1 contributions. Define a binary state
\[
\theta_t \;=\;\mathbf{1}\{c_{it}>\hat c\}\in\{0,1\}\quad\text{(High vs.\ Low)}.
\]
Then $\{\theta_t\}$ is a Markov chain with transition matrix $\mathbf T(\delta)=\bigl[p_{xy}\bigr]_{x,y\in\{0,1\}}$. For small $\delta$ we linearize $p_{xy}$ analytically and verify the approximation via agent-based simulation.
%--------------------------------------------------------
\subsection{Identification and Empirical Mapping}
\label{sec:ident}
%--------------------------------------------------------
Only the product $\phi_i = 2\knorm h_i$ enters the ESS second-order (ESS/branching) condition, whereas the first-order condition for the singular strategy $c^*$ depends only on $d_i$. Separate identification of $(d_i,\phi_i)$ therefore leverages auxiliary measures and the panel:
\begin{itemize}
  \item $h_i$, we conceptually interpret it as aversion to being at the norm; in our empirical mapping we calibrate $k h_i$ from dispersion.
  \item $\knorm$ is calibrated from the cross-round dispersion of $(c_{it}-\bar c_{-i,t-1})^{2}$.
  \item $d_i$ is estimated by maximum likelihood given $(h_i,\knorm)$, with variance-inflation checks.
\end{itemize}
For empirical alignment with the observed two-basin structure, we also analyze a binary High/Low state $\theta_t$ and calibrate a two-parameter Fermi–Moran process. Let $\theta_i\in\{H,L\}$ and define
\[
w_i(\theta_i)=
\begin{cases}
1+d, & \theta_i=H,\\
1, & \theta_i=L,
\end{cases}
\qquad
\Pr(\text{reproduce }i)=\dfrac{\exp(k\,w_i)}{\sum_{\ell}\exp(k\,w_\ell)}.
\]
We then choose $(d,k)$ to minimize $\mathrm{RSS}=\lVert T_{\mathrm{sim}}(d,k)-T_{\mathrm{emp}}\rVert_F^2$, where $T_{\mathrm{emp}}$ is the empirical round-1$\to$round-10 transition matrix. This calibration serves as a compact stochastic bridge from the micro model to the observed High/Low persistence and changes.% and mobility.

\subsection{Instrumental–variables strategy}
\label{sec:iv_strategy}

We estimate peer effects using two complementary two-stage least squares (2SLS) designs that differ in their fixed effects and instruments.

\paragraph{(A) Contemporaneous peers with LOO composition IVs (round \& village FE).}
For the same–round peer mean $\bar c_{-i,gt}$ we use \emph{leave–one–out} (LOO) composition instruments constructed from predetermined traits $W^{(l)}$:
\[
Z_{ig}^{(l)} \;=\; \frac{1}{N-1}\sum_{j\neq i} W_j^{(l)}.
\]
Random assignment into groups within villages guarantees relevance, while anonymity (avatars) and the predetermined nature of $W$ support exclusion. Because these LOO instruments are (largely) time–invariant at the player level when group composition is fixed, we \emph{do not} include individual fixed effects here; instead we remove \emph{round} and \emph{village} fixed effects by two–way demeaning,
\begin{equation}
\tilde z_{igt} \;=\; z_{igt} - \bar z_{\cdot\cdot t} - \bar z_{v\cdot\cdot} + \bar z_{\cdot\cdot\cdot},
\label{eq:tw_demean}
\end{equation}
and estimate
\begin{equation}
\tilde c_{igt} \;=\; \beta\,\tilde{\bar c}_{-i,gt} \;+\; \tilde X'_{igt}\gamma \;+\; \tilde\varepsilon_{igt},
\qquad
\tilde{\bar c}_{-i,gt}\ \text{endogenous}.
\label{eq:second_stage}
\end{equation}
% with first stage
% \begin{equation}
% \tilde c_{igt} \;=\; \beta\,\tilde{\bar c}_{-i,gt} \;+\; \tilde X'_{igt}\gamma \;+\; \tilde\varepsilon_{igt},
% \qquad
% \tilde{\bar c}_{-i,gt} \;\;\text{endogenous}.
% \label{eq:second_stage}
% \end{equation}
In the preferred lean set we use $\{Z_{ig}^{(\text{male})}, Z_{ig}^{(\text{no\,religion})}, Z_{ig}^{(\text{indigenous})}\}$ (dropping \textit{Protestant}). We report weak–IV $F$ (Kleibergen–Paap), Wu–Hausman, and Sargan tests.

\paragraph{(B) Lagged peers with time–varying shift–share IV (individual FE + village$\times$round FE).}
To capture learning, we target the \emph{lagged} peer mean $\bar c_{-i,g,t-1}$ and use a single time–varying shift–share instrument. Let $s^{(l)}_{g,-i}$ be the LOO group share of trait $l$, and let $\mu^{(l)}_{t-1}$ be the round–$(t\!-\!1)$ cross–village mean of that trait. We implement a leave–one–village (LOV) variant that excludes village $v(i)$ from the shifter:
\begin{equation}
Z^{LOV}_{igt} \;=\; \sum_{l} s^{(l)}_{g,-i}\,\mu^{(-v),(l)}_{t-1},
\label{eq:Z_LOV}
\end{equation}
This delivers \emph{time variation} in the instrument even when $s^{(l)}_{g,-i}$ is fixed for a player. In estimation we absorb \emph{individual} and \emph{village$\times$round} fixed effects (i.e., two–way within relative to these FE) and estimate
\begin{equation}
\tilde c_{igt} \;=\; \beta\,\widetilde{\bar c}_{-i,g,t-1} \;+\; \tilde X'_{igt}\gamma \;+\; \tilde\varepsilon_{igt},
\qquad
\widetilde{\bar c}_{-i,g,t-1}\ \text{instrumented by}\ \tilde Z^{LOV}_{igt}.
\label{eq:second_stage_lov}
\end{equation}
Standard errors are clustered by group and are robust to heteroskedasticity (small–sample–adjusted HC1). Because the cluster–robust first stage is moderate in this design, we accompany 2SLS with weak–ID–robust Anderson–Rubin confidence sets.

\medskip
\noindent Design (A) asks whether the same–round peer mean causally moves $c_{it}$ once we purge common shocks (round, village); the IV estimates are near zero, consistent with simultaneity driving OLS co‑movement. Design (B) asks whether players \emph{update} toward last round’s group norm; the LOV shift–share IV yields a positive, economically meaningful response, with inference based on cluster–robust KP $F$ and AR intervals.

\section{Results}
\label{sec:results}

% ===== Insert after Table~\ref{tab:het_vars} =====
\subsection{Main empirical markers of the bifurcation}
\label{sec:markers}

Behavior separates into two basins along three complementary markers. 
(i) The empirical drift crosses zero once at $c^{*}\!\approx 5.75$~L (95\% band $[5.50,6.01]$): contributions below drift up, and those above drift down (Fig.~\ref{fig:summary_panels}\,a). 
(ii) A two–state HMM identifies a Low state (mean $\approx 4.09$~L) and a High state (mean $\approx 9.05$~L) with very persistent transitions and rarer escapes from Low than slips from High (Fig.~\ref{fig:summary_panels}\,b). 
(iii) At the village level, the probability of ending High rises steeply once the early share above $c^{*}$ exceeds $\approx\!60\%$ (95\% CI $[0.55,0.64]$; Fig.~\ref{fig:summary_panels}\,c).
These markers align with the singular strategy and branching logic in Section~\ref{sec:theory}.

\subsubsection{Empirical tipping point \(c^{*}\).}

We recover a single, data–driven breakpoint at \(c^{*}\!\approx 5.75\)~L (95\% band \([5.50,6.01]\)) at which the expected round–to–round change in giving switches sign (Figure~\ref{fig:tipping}). Contributions below \(c^{*}\) tend to drift upward, while contributions above \(c^{*}\) drift downward. This pattern is the empirical hallmark of two behavioral “basins”: a low–giving basin that pulls up toward \(c^{*}\), and a high–giving basin that relaxes back toward it. \(c^{*}\) can be understood as the singular attractor found in our adaptive dynamics. Substantively, \(c^{*}\) offers a transparent diagnostic in the long term: groups spending early rounds mostly above the threshold are primed to sustain higher giving, whereas those starting mostly below are at risk of settling into a low–contribution regime. This divergence corresponds to evolutionary branching of two contribution paths, identified by our evolutionary game theory analysis above.

For each player \(i\) and round \(t=1,\dots,9\) we form the increment \(\Delta c_i^{\,t}:=c_i^{\,t+1}-c_i^{\,t}\) and estimate the conditional mean function
\[
m(c)\;=\;\mathbb{E}\!\left[\Delta c\,\middle|\,c\right]
\]
nonparametrically using a generalized additive model with a cubic regression spline in \(c\) (Gaussian family, identity link; smoothing chosen by REML). Pointwise 95\% confidence bands are obtained via cluster bootstrap at the individual level. The tipping point \(c^{*}\) is defined as the unique zero of the fitted drift, \(\hat m(c^{*})=0\); its 95\% interval is computed from the bootstrap distribution of the root (equivalently, as the set of \(c\) values where the confidence band intersects zero). Results are robust to local–polynomial kernels and bandwidth choices.\footnote{We exclude round 10 when forming \(\Delta c\), trim extreme \(c\) values at conventional percentiles to limit leverage, and obtain nearly identical \(c^{*}\) under local linear regression with plug-in bandwidth.}

% data NEW_WORK.R

\begin{figure}[H]
  \centering

  % --- top row ---
  \begin{subfigure}{0.49\textwidth}
    \centering\includegraphics[width=\linewidth]{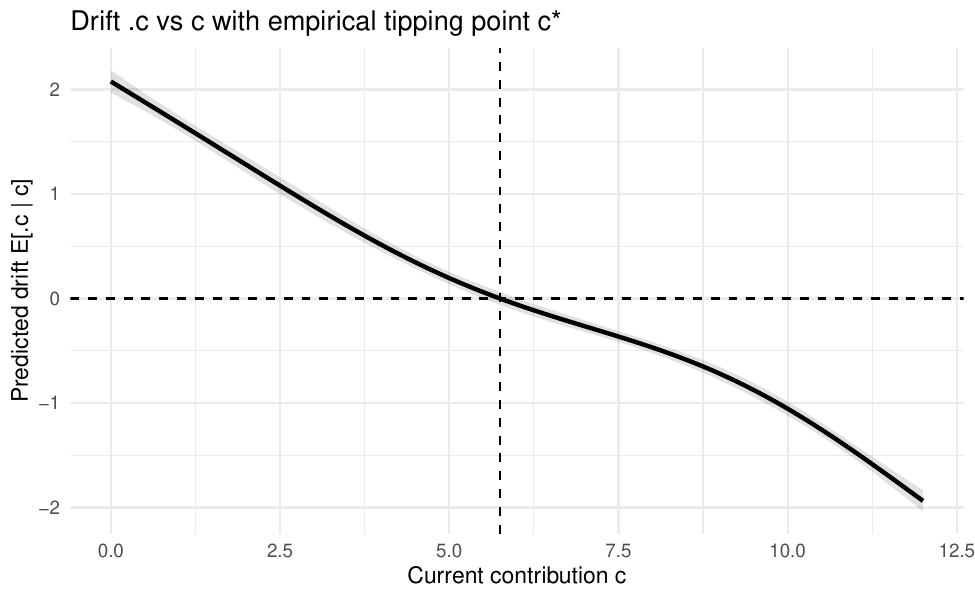}
    \subcaption{Empirical drift $\mathbb{E}[\Delta c \mid c]$ and 95\% band. The zero crossing at $c^{*}\!\approx 5.75$~L separates upward (left) from downward (right) drift.}
    \label{fig:tipping}
  \end{subfigure}\hfill
  \begin{subfigure}{0.49\textwidth} \centering\includegraphics[width=\linewidth]{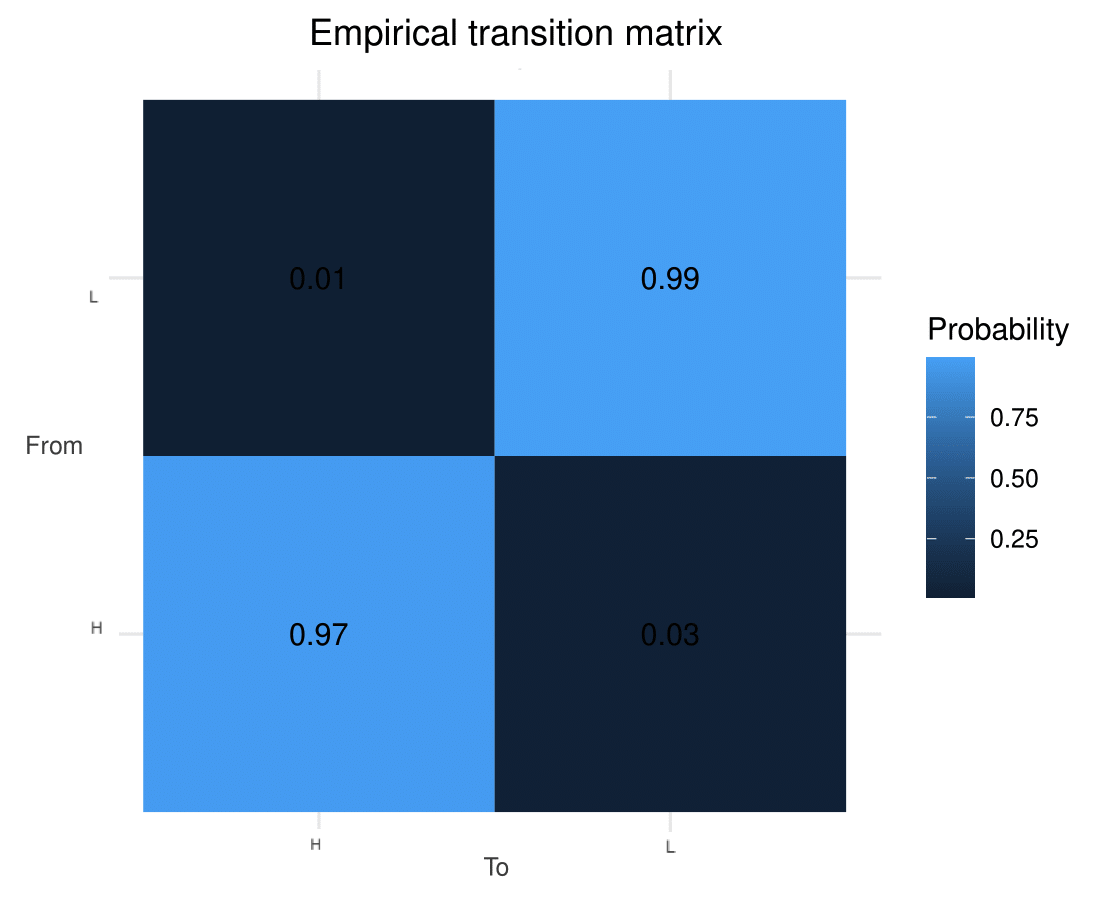}
    \subcaption{Two-state HMM transitions (rows = current, columns = next). 
      Sticky diagonals; rarer escapes from Low than slips from High.}
    \label{fig:hmm}
  \end{subfigure}

  \medskip

  % --- bottom row ---
  \begin{subfigure}{0.49\textwidth}
 \centering\includegraphics[width=\linewidth]{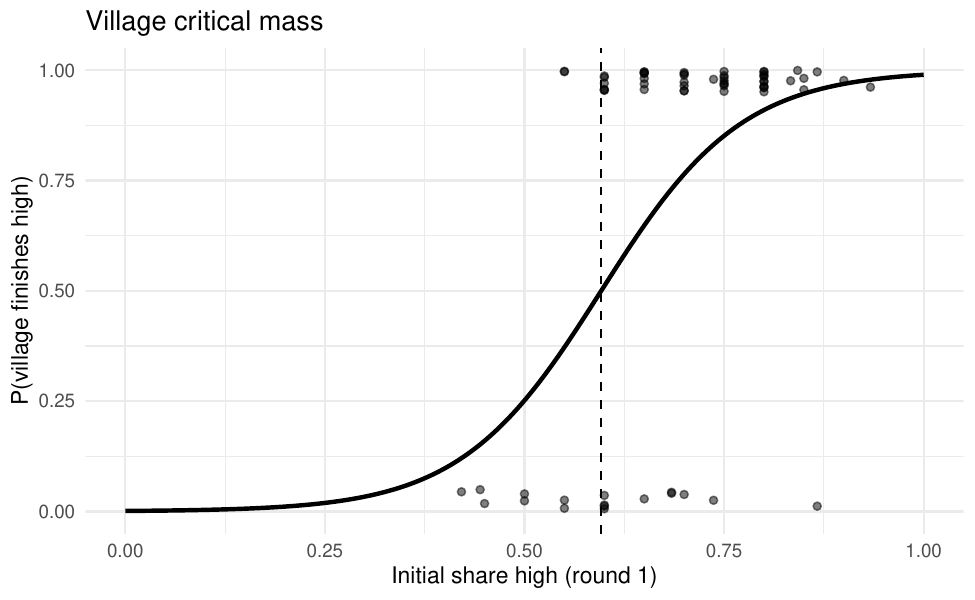}
    \subcaption{Village critical mass: probability of finishing high vs.\ initial share $>c^{*}$. The 50–50 threshold is $\approx 0.60$.}
    \label{fig:criticalmass}
      \end{subfigure}\hfill     
  \begin{subfigure}{0.49\textwidth} \centering\includegraphics[width=\linewidth]{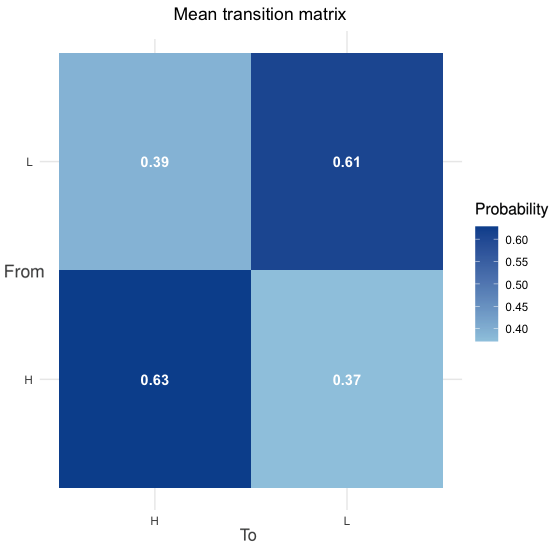}
     \subcaption{Empirical two–state transition matrix of calibrated two–parameter Fermi–rule Moran process over the session.}   \label{fig:calibration_markov_hmm:markov}
  \end{subfigure}

  \caption{\textbf{Empirical markers of the bifurcation.} (a) Drift $\mathbb{E}[\Delta c\mid c]$ crosses zero at $c^{*}\!\approx 5.75$\,L, separating upward from downward adjustment. (b) Two-state HMM transitions show sticky diagonals with rarer $L\!\to\!H$ than $H\!\to\!L$ moves. (c) Village critical mass: the probability of finishing High rises sharply with the initial share above $c^{*}$; the 50–50 threshold is $\sim 0.60$. (d) Calibrated two-parameter Fermi–rule Moran process: empirical two-state transition matrix over the session reproducing persistence and asymmetry.}
  \label{fig:summary_panels}
\end{figure}

\subsubsection{State persistence and asymmetry (HMM).}

A two–state Hidden Markov Model identifies a Low state (mean \(\approx 4.09\)~L) and a High state (mean \(\approx 9.05\)~L), with strongly persistent transitions: \(P(H\!\to\!H)=0.974\) and \(P(L\!\to\!L)=0.991\) (Figure~\ref{fig:hmm}). Off–diagonal moves are rare, and notably, escapes from Low (\(L\!\to\!H\)) are less frequent than slips from High (\(H\!\to\!L\)). These estimates formalize path dependence: once a player enters a state, subsequent play tends to remain there, with an asymmetry that makes Low particularly “sticky.” This finding validates the two–basin view implied by the drift analysis and rules out simple mean reversion as the sole driver and emphasizes prevention over cure: avoiding early \(H\!\to\!L\) slippage (e.g., by buttressing early norms) is easier than engineering later \(L\!\to\!H\) recoveries.

\subsubsection{Village critical mass.}

At the village level, the probability of \emph{ending} in the high–cooperation basin rises steeply with the initial share of players above $(c^{*})$ (round~1), with a data–driven “critical mass” at \(\hat s^{\text{crit}}\!\approx 0.60\) (95\% CI \([0.55,0.64]\)); the slope is large and precise (\(z\!\approx\!4.98\), \(p<10^{-6}\); Figure~\ref{fig:criticalmass}). Villages near or above this 60\% mark very likely finish high; those below it rarely do. This village–level tipping relationship connects individual dynamics to collective outcomes: early majorities of above–threshold contributors lock in favorable feedback, while early deficits are hard to overcome. The result provides a concrete targeting rule for interventions, seed early contributions among socially central actors until the initial above, $(c^{*})$ share clears $\approx 60\%$, after which cooperative dynamics are more likely to be self–sustaining.

% --------------------------------------------------------------
\subsubsection{Evolutionary Dynamics of Bifurcation }
\label{sec:calibration}

% Moran Heterogeneous

We provide a compact, descriptive calibration that links the observed two–state
dynamics to a standard binary birth–death imitation process. Full derivations, simulation
details, and robustness checks are provided in SI; the calibration is illustrative (not causal)
and is not used to estimate treatment effects.

To connect the High–Low split to a minimal evolutionary representation, we calibrate a
two-parameter Fermi rule Moran process to the empirical two–state transition matrix over the
sessions. The two parameters have simple interpretations: $d$ is a \emph{selection tilt} toward High
(positive) or Low (negative), and $k$ is \emph{imitation intensity} (higher $k$ reduces noise). In each
microscopic update, one individual reproduces with probability proportional to
$\exp(k\,w_i)$, where $w_i=1+d\,\mathbf 1\{\theta_i=H\}$; the offspring replaces a uniformly chosen
peer. This binary birth–death process is a parsimonious stochastic analogue to best–reply dynamics
and offers a compact micro‑foundation for two basins.

We target the two–state transition matrix inferred from the HMM (latent state; Fig.~\ref{fig:calibration_markov_hmm:markov}),
which filters measurement noise; calibrating to the integrated Markov–logit row means yields very
similar patterns (Appendix \ref{app:bifurcation}). We choose $(d,k)$ to minimize the Frobenius norm
$RSS=\lVert T_{\text{sim}}(d,k)-T_{\text{emp}}\rVert_F^{2}$, using a coarse grid
($d\in[-2,3]$, $k\in[0,1.5]$, step $0.25$) followed by L‑BFGS‑B refinement and averaging simulated
transition matrices over many runs.

The optimum is
\[
\hat d=-0.51,\qquad \hat k=0.50,\qquad RSS=0.057.
\]
The calibrated process reproduces persistence of High contributors well and downward drift reasonably:
\[
p_{HH}^{\text{emp}}=0.69\ \text{vs}\ \hat p_{HH}=0.64,\qquad
p_{HL}^{\text{emp}}=0.31\ \text{vs}\ \hat p_{HL}=0.36.
\]
It over‑predicts upward mobility from Low ($p_{LH}^{\text{emp}}=0.18$, $\hat p_{LH}=0.35$) and thus
under‑predicts Low stickiness ($p_{LL}^{\text{emp}}=0.82$, $\hat p_{LL}=0.65$). The negative
$\hat d<0$ indicates a net selection tilt favoring Low once imitation strength $\hat k\approx0.5$ is
accounted for, consistent with the clustering and hazard evidence that $H\!\to\!L$ drops outnumber
$L\!\to\!H$ escapes. The remaining gap on $LL$ points to frictions absent from the minimal process
(e.g., thresholds, reputation costs, or heterogeneity in $d_i$) anchoring Low more firmly in practice.

\subsection{Heterogeneity}

% ==============================================================
\subsubsection{Heterogeneity in Contributions}
\label{sec:hetero}

% THIS IS THE CODE IN Section2_NEW.R

\begin{figure}[H]
  \centering
  \includegraphics[width=\textwidth]{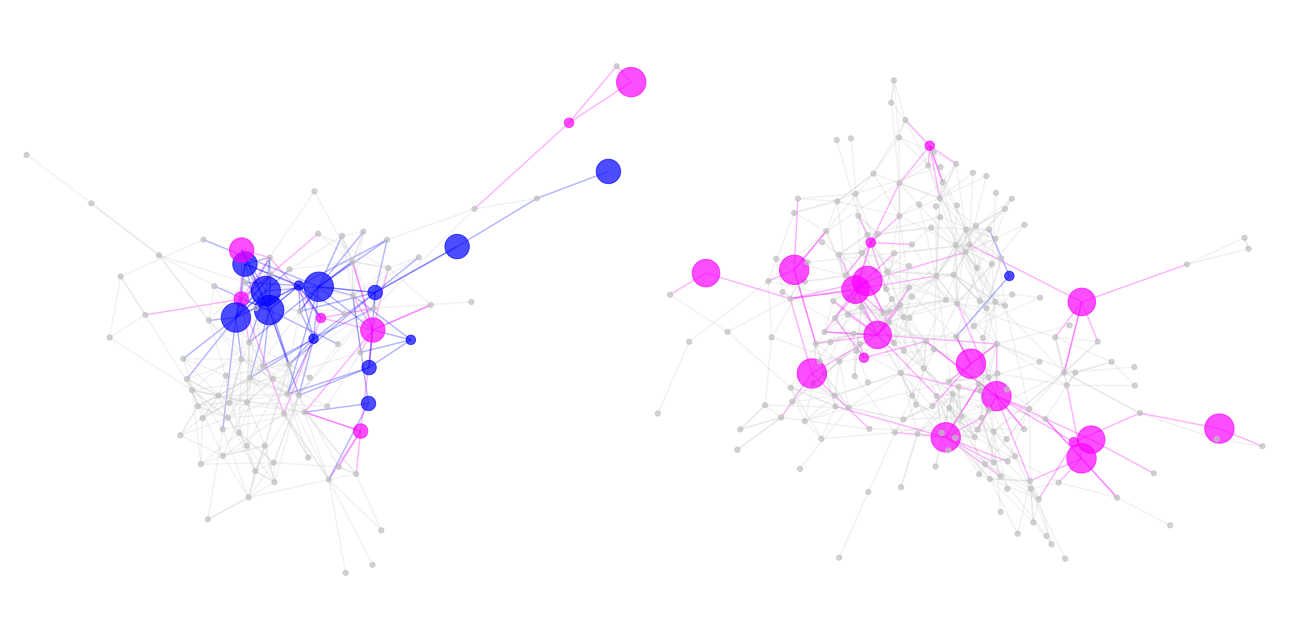}
  \caption{\textbf{Friendship networks in two illustrative villages.} Gray nodes and ties show full village sociocentric networks, colored nodes are participants in the PGG (\textcolor{blue}{men}, \textcolor{magenta}{women}); edges denote friendship ties among these subjects, and node size is proportional to total contribution. Village~25 (left) is male‑dominated and has a lower mean contribution (5.9~L) than Village~4 (right, 7.1~L). In Village~25, the most central men supply a disproportionate share of the public good, consistent with regression evidence that male status and social embeddedness are positively associated with giving.}
\label{fig:villages_network}
\end{figure}

% --- Definitions of lags and round counter
\paragraph{Definitions.}
Let $v(i)$ denote the village and $g(i)$ the (fixed) game group of player $i$.
For $t\ge2$, define the own and peer lags and the round counter as
\begin{align}
c^{\text{own}}_{ig,t-1} &:= c_{ig,t-1},\\
c^{\text{peer}}_{ig,t-1} &:= \bar c_{-i,g,t-1}
=\frac{1}{N-1}\sum_{j\neq i} c_{jg,t-1},\\
r_t &:= \text{round number} \in\{1,\dots,T\}.
\end{align}
Throughout this subsection, $c^{\text{peer}}_{ig,t-1}$ is measured in \emph{Lempiras} (0--12).

% --- Linear mixed-effects (autoregressive) model
\paragraph{Model.}
The linear mixed-effects model for Lempiras contributed per round is

\begin{align}
\label{eq:lme_main}
\hspace*{-1.5cm}
c_{igt} \;=\;&\;
\underbrace{\beta_0
+ \beta_1\,c^{\text{own}}_{ig,t-1}
+ \beta_2\,c^{\text{peer}}_{ig,t-1}
+ \beta_3\,r_t
+ X_{ig}'\gamma}_{\text{fixed part}} \nonumber \\[6pt]
&\;+\;
\underbrace{\delta_{g(i)} \;+\; \alpha_{v(i)} \;+\; \kappa_{v(i),t} \;+\; a_i}_{\text{random intercepts}}
\;+\;
\underbrace{b_{i1}\,r_t \;+\; b_{i2}\,c^{\text{own}}_{ig,t-1} \;+\; b_{i3}\,c^{\text{peer}}_{ig,t-1}}_{\text{by-player random slopes}}
\;+\;
\varepsilon_{igt}.
\end{align}
% \begin{equation}
% \label{eq:lme_main}
% \hspace*{-1.5cm}
% \boxed{%
% c_{igt}
% = \underbrace{\beta_0
% + \beta_1\,c^{\text{own}}_{ig,t-1}
% + \beta_2\,c^{\text{peer}}_{ig,t-1}
% + \beta_3\,r_t
% + X_{ig}'\gamma}_{\text{fixed part}}
% \;+\;
% \underbrace{\delta_{g(i)} \;+\; \alpha_{v(i)} \;+\; \kappa_{v(i),t} \;+\; a_i}_{\text{random intercepts}}
% \;+\;
% \underbrace{b_{i1}\,r_t \;+\; b_{i2}\,c^{\text{own}}_{ig,t-1} \;+\; b_{i3}\,c^{\text{peer}}_{ig,t-1}}_{\text{by-player random slopes}}
% \;+\;
% \varepsilon_{igt},}
% \end{equation}
where $X_{ig}$ collects observed covariates (e.g., gender, age, education, and network measures).

% --- Random effects and errors
\paragraph{Random effects and errors.}
Random intercepts for groups, villages, village$\times$round cells, and players:
\begin{align}
\delta_{g} &\sim \mathcal{N}(0,\ \sigma_{\delta}^{2}), \qquad g=1,\dots,G,\\
\alpha_{v} &\sim \mathcal{N}(0,\ \sigma_{\alpha}^{2}), \qquad v=1,\dots,V,\\
\kappa_{v,t} &\sim \mathcal{N}(0,\ \sigma_{\kappa}^{2}) \quad \text{for observed village$\times$round cells},\\
a_{i} &\sim \mathcal{N}(0,\ \sigma_{a}^{2}), \qquad i=1,\dots,I,
\end{align}
and by-player random slopes (assumed \emph{uncorrelated}):
\begin{equation}
\label{eq:rand_slopes}
\begin{bmatrix} b_{i1}\\[2pt] b_{i2}\\[2pt] b_{i3} \end{bmatrix}
\;\sim\;
\mathcal{N}\!\left(
\begin{bmatrix} 0\\[2pt] 0\\[2pt] 0 \end{bmatrix},
\ \mathrm{diag}\!\big(\tau_{r}^{2},\ \tau_{\text{own}}^{2},\ \tau_{\text{peer}}^{2}\big)
\right),
\qquad
\text{with } \Cov\!\big((b_{i1},b_{i2},b_{i3}),\,a_i\big)=\mathbf{0}.
\end{equation}
The idiosyncratic error is
\begin{equation}
\varepsilon_{igt} \sim \mathcal{N}(0,\ \sigma_{\varepsilon}^{2}),
\qquad
\text{independent of all random effects.}
\end{equation}
We note a few considerations in our analysis below:\\
\noindent
(i) Equation~\eqref{eq:lme_main} is an ARX specification: own dynamics enter via $c^{\text{own}}_{ig,t-1}$ and social dynamics via the lagged peer mean $c^{\text{peer}}_{ig,t-1}$.\\
(ii) Random intercepts allow group-, village-, village$\times$round-, and player-specific baselines; by-player random slopes let the persistence ($c^{\text{own}}_{ig,t-1}$), peer responsiveness ($c^{\text{peer}}_{ig,t-1}$), and round trend ($r_t$) vary across players.\\
(iii) ``Uncorrelated random slopes'' is implemented by the diagonal covariance in \eqref{eq:rand_slopes}; this can be relaxed if slope correlations are of interest.

%\paragraph{Scaling of the peer regressor in this section.}
%Throughout \S5.2 we use the lagged leave--one--out peer mean $c^{\text{peer}}_{ig,t-1}$ in \emph{lempiras} (0--12). By contrast, \S5.4--\S5.5 report the same peer signal \emph{scaled by the 12L endowment}. Odds ratios in those hazards/Markov models are therefore per one--endowment (12L) change; for interpretation per 1L, raise the reported OR to the power $1/12$.

\paragraph{Main results.}We estimate a linear mixed-effects model\footnote{Random intercepts for villages and players, and by-player
\emph{uncorrelated} random slopes for the round counter $r_t$, own lag $c^{own}_{i,t-1}$, and peer lag $c^{group}_{i,t-1}$}
with the dependent variable equal to Lempiras contributed per round.
Using \(N=22{,}239\) observations from \(2{,}471\) players in \(134\) villages (Table~\ref{tab:rs_learn_full_A})),
the baseline predicted contribution is \(4.48\)~L per round (intercept).

Own persistence is strong: the coefficient on the individual lag $c^{\text{own}}_{ig,t-1}$ is \(+0.126\) (SE \(=0.0091\), \(t=13.72\), \(p<2\times10^{-16}\)).
The lagged peer mean $c^{\text{peer}}_{ig,t-1}$ is small and statistically indistinguishable from zero. Because group shocks and simultaneity can attenuate the lagged‑peer slope in RE models, we treat them as descriptive; peer‑updating is tested with FE/IV in Section~\ref{sec:iv_results}.

Among covariates, men contribute \(+0.569\)~L relative to women (\(p=2.25\times10^{-6}\)). Figure~\ref{fig:villages_network} illustrates a compatible pattern: Village~25 is male‑dominated and its most central men account for a disproportionate share of contributions, whereas Village~4, with a higher share of women, exhibits a higher mean. These are \emph{associations}, not causal effects: part of the gap may be mediated by unobserved resources, status, or network position only partly captured by included covariates and random effects. Intuitively, the association may reflect gendered norms around public giving and competitiveness, looser intra‑household budget constraints for men in some settings, or higher tolerance for strategic risk in early rounds.

Each additional named friend is associated with slightly higher giving with \(+0.03\)~L (\(p=0.039\)). This pattern is consistent with a natural \emph{social capital} mechanism: having more friends may heighten reciprocity pressures, increase reputational visibility, and expose individuals to stronger pro-social norms. Although the estimated effect is small for any single tie, it can accumulate for highly connected players. Because the specification is in levels, however, the coefficient likely blends both “preference” and “exposure” channels, and may also proxy for centrality or social status.

Education is positively associated with contributions (\(+0.0539\)~L, \(p=0.039\)). More educated participants give a little more, consistent with mechanisms like (i) better comprehension of dynamic incentives and conditional cooperation, (ii) higher patience or lower present-bias, or (iii) correlation with income or opportunity cost. Again, this is descriptive: education may proxy for multiple latent attributes (numeracy, earnings, status) that we cannot separately identify here. 

Those reporting no religion contribute slightly more on average, \(+0.368\)~L (\(p=0.080\), marginal). This could reflect selection into “no religion” that correlates with urban exposure, education, or different local social networks, rather than doctrinal effects per se. Given the marginal $p$-value and potential confounding, we treat this as suggestive background rather than an interpretable structural difference.

Age, adversarial ties, food insecurity, marital status, Protestant or Indigenous identity, and village characteristics, such as network density, network size, and road access, do not predict whether a villager finishes High vs. Low; all coefficients are not statistically significant at conventional levels. We view this linear mixed model as descriptive (at the village level). Our causal/dynamic peer-response estimates come from the IV specification with individual and village\(\times\)round fixed effects and both lags (Sec.~\ref{sec:iv_results}, Table~\ref{tab:learn_fixef}). Robustness checks, which include financial autonomy (only 400 samples) is applied on 17{,}847 observations with 1983 players in 134 villages, Table \ref{tab:rs_learn_full_B}.

These descriptive ARX results show strong own persistence and modest average peer slopes in random‑effects models; below we turn to a two‑regime view that captures the evident bifurcation in trajectories and links peer exposure to regime changes.

% ---------------------------------------------------------------------------
% Trajectory clustering — empirical evidence for the CSS bifurcation
% ---------------------------------------------------------------------------

\subsubsection{Heterogeneity in Bifurcation through Hierarchical Clustering}
\label{sec:traj-cluster}

% THIS IS THE CODE IN Clustering m1_state and mod_mem

Let $c^\star$ denote the round‑1 mean contribution (baseline unless noted) and define the High/Low state
$s_{igt}=\mathbf{1}\{c_{igt}\ge c^\star\}\in\{0,1\}$.
To study within‑session dynamics we estimate
\begin{equation}
\label{eq:dyn_state_logit}
\mathrm{logit}\,\Pr\!\big(s_{igt}=1\mid s_{ig,t-1},\,m_{g,t-1},\,X_i,\,\alpha_{v(i)},\,\delta_{g(i)},\,a_i\big)
=\beta_0+\rho\,s_{ig,t-1}+\lambda\,m_{g,t-1}+\theta\,t + X_i^\prime\gamma
+\alpha_{v(i)}+\delta_{g(i)}+a_i,
\end{equation}
where $m_{g,t-1}:=\bar c_{-i,g,t-1}/12$ is the lagged, leave‑one‑out peer mean scaled by the 12L endowment (so a one‑unit change equals a \emph{full endowment} shift), $t$ is the round counter, $X_i$ are baseline covariates, and $(\alpha_{v(i)},\delta_{g(i)},a_i)$ are random intercepts for village, group (nested in village), and player, respectively. We follow \citet{wooldridge2005simple} to adjust for dynamic initial conditions via
$a_i=\psi_0+\psi_1\,s_{ig1}+\psi_2\,\bar Z_i+u_i$. Odds ratios reported for $m_{g,t-1}$ are per \emph{one endowment} (12L) change in the lagged peer mean.
For a per‑1L interpretation, raise the reported OR to the power $1/12$.

\paragraph{Main results.}We cluster complete ten‑round histories using Ward–D\textsuperscript{2} on the $z$‑scored player multiplied by round matrix. We restrict to players with all ten observed rounds and $z$‑score contributions by round before clustering to avoid scale differences across rounds. Average silhouette width peaks at two clusters ($\bar S=0.362$) and falls thereafter ($\bar S=0.258,0.218,0.176,0.152$ for $k=3,\dots,6$; Table~\ref{tab:hc_sil_by_k}).
At two clusters ($K=2$) we obtain two cohesive groups (Table~\ref{tab:hc_sizes_sil}): a larger \emph{High} cluster ($n=1{,}826$, $\bar S=0.34$)
and a smaller \emph{Low} cluster ($n=765$, $\bar S=0.42$). Alignment with round‑10 end states is strong (Table~\ref{tab:hc_confusion}):
the High cluster contains 1{,}041 of 1{,}055 HH finishers (98.7\%) and most movers (330/472 HL; 179/195 LH), whereas the Low cluster is predominantly LL (593/848=69.9\%) and rarely HH (14/848=1.7\%).
Cluster mean paths remain well separated across rounds (Figure~\ref{fig:traj_by_cluster}).
    
\begin{table}[H]
\centering
\caption{Average silhouette by number of clusters $K$ (Ward–D\textsuperscript{2})}
\label{tab:hc_sil_by_k}
\begin{tabular}{lccccc}
\toprule
$K$ & 2 & 3 & 4 & 5 & 6 \\
\midrule
Mean silhouette $\bar S$ & 0.362 &  0.258 & 0.218 & 0.176 & 0.152 \\
\bottomrule
\end{tabular}
\end{table}

\begin{table}[H]
\centering
\caption{Cluster sizes and mean silhouette widths at two clusters ($K$=2)}
\label{tab:hc_sizes_sil}
\begin{tabular}{lcc}
\toprule
Cluster & Size & Mean silhouette \\
\midrule
C1 (High) & 1{,}826 & 0.34 \\
C2 (Low)        &   765   & 0.42 \\
\bottomrule
\end{tabular}
\end{table}

\begin{table}[H]
\centering
\caption{Ward clusters vs.\ round‑10 end states (High/Low relative to round‑1 mean)}
\label{tab:hc_confusion}
\begin{tabular}{lrrrr}
\toprule
 & \multicolumn{4}{c}{Final end state} \\
\cmidrule(l){2-5}
HC cluster & \textbf{HH} & \textbf{HL} & \textbf{LH} & \textbf{LL} \\
\midrule
C1 (High) & 1{,}041 & 330 & 179 & 276 \\
C2 (Low)        & 14 & 142 & 16 & 593 \\
\bottomrule
\end{tabular}
\end{table}

\begin{figure}[H]
  \centering
  \includegraphics[width=\textwidth]{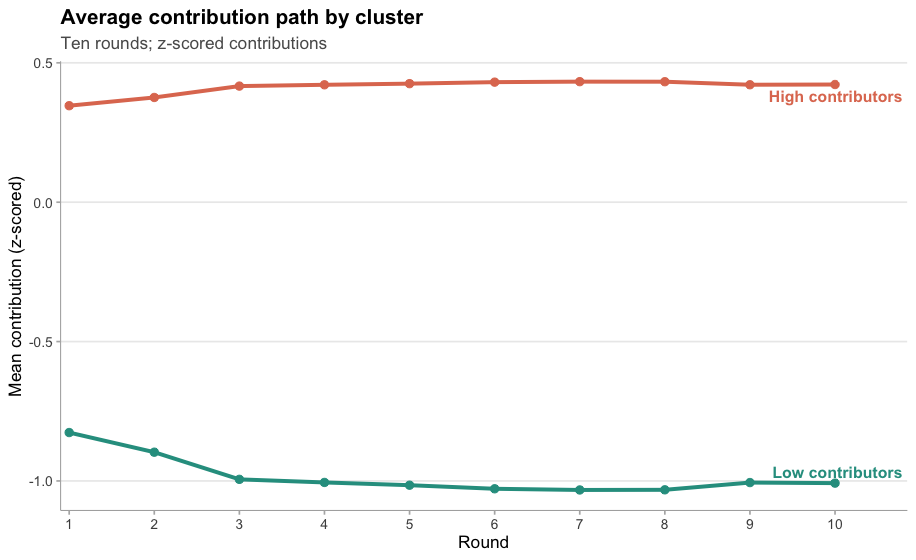}
  \caption{Mean $z$‑scored contribution paths by Ward cluster (two clusters $K=2$). The High path starts above the round‑wise mean and remains high; the Low path starts below and declines.}
  \label{fig:traj_by_cluster}
\end{figure}

Table~\ref{tab:m1_state_or} reports odds ratios (OR) from \eqref{eq:dyn_state_logit} on $N\!=\!22{,}239$ decisions (2,471 players; 134 villages). The lagged \emph{peer} signal (scaled so that 1 equals a full endowment) is the dominant driver
(OR per endowment $\approx 5.97$, $p\!\ll\!0.001$).
The own state is mildly persistent (OR $=1.22$, $p<0.001$), there is a small negative time drift
(OR $=0.98$, $p=0.012$), and initial High strongly predicts subsequent High (OR $\approx 87.66$, $p\!\ll\!0.001$). Players embedded, on average, in lower peer environments are less likely to be High (Avg.\ peer mean$_{t-1}$: OR $=0.82$, $p<0.001$).

Among baseline covariates, men are more likely to be High (OR $=1.33$, $p=0.029$); number of friends is marginally positive (OR $1.12$, $p=0.089$); other covariates are imprecise.

%Village‑level heterogeneity is negligible (random‑intercept variance at the boundary).

As a complement to the dynamic model, we regress LCMM class membership (High vs.\ Low) on baseline covariates with a village random intercept (Table~\ref{tab:mem_lcmm_or}).
Men (OR $=1.43$, $p<0.001$) and players with more friends (per SD; OR $=1.10$, $p=0.037$) are more likely to belong to the High class.
Reporting no religion (vs.\ Catholic) is positively associated (OR $=1.36$, $p=0.043$), while the ethnicity marker is negatively associated (OR $=0.70$, $p=0.010$).
The Protestant indicator (vs.\ Catholic) is positive but marginal (OR $=1.18$, $p=0.095$); other covariates are not statistically different from zero.
Village‑level variance is small (0.104).

Full‑trajectory clustering and the round‑by‑round state model consistently reveal a two‑regime contribution process.
Persistence in the cooperative (High) regime is driven primarily by lagged peer behavior and initial conditions, with mild own‑state stickiness and a small negative time trend.
Male gender and social connectedness (friends) are the most robust correlates of occupying the High regime; religion markers (versus Catholic) are positive or marginally so, while the ethnicity indicator is negative.
Other baseline covariates, including education and food insecurity, show no systematic association once dynamics are controlled.

Moreover, the two trajectory clusters mirror the HMM’s sticky High/Low states and the asymmetric mobility documented above: cluster‑specific transitions exhibit higher $H\!\to\!L$ than $L\!\to\!H$ when averaged over rounds. Our Fermi–Moran calibration summarizes this pattern as a negative selection tilt ($\hat d<0$) at moderate imitation intensity ($\hat k\approx0.5$), i.e., the Low basin is harder to escape than the High basin is to slip from. This bridges the clustering evidence with the hidden‑state transitions in Fig.~\ref{fig:hmm} and the village critical‑mass gradient, and provides a compact evolutionary summary of the persistence and asymmetry in the data (Sec.~\ref{sec:calibration}).\footnote{See Appendix for calibration details and robustness.}

Taken with the HMM transitions and the village critical‑mass gradient, the two trajectory clusters provide an empirical basis for a compact evolutionary summary: a Fermi–Moran imitation process with a negative selection tilt toward Low and moderate imitation intensity reproduces the persistence and asymmetry in the data.

\subsection{Causal peer effects: 2SLS estimates}

\label{sec:iv_results}

\paragraph{Identification and scaling.}
Choices are simultaneous within a round, so within–round imitation is mechanically ruled out: players do not observe others’ anonymous contributions until the round ends. We partial out player, village$\times$round shocks and instrument the (demeaned) peer mean. Let tildes denote variables demeaned by player and by village$\times$round (two–way within). For the \emph{lagged} peer effect we estimate
\begin{align}
\text{(Second stage)}\quad \tilde{c}_{igt} &= \beta\,\tilde{\bar c}_{-i,g,t-1} + \varepsilon_{igt}, \\
\text{(First stage)}\quad \tilde{\bar c}_{-i,g,t-1} &= \Pi\,Z_{igt} + u_{igt},
\end{align}
with standard errors clustered by group. The peer regressor is measured in lempiras (0--12), so $\beta$ is interpretable as L per 1 L change in the (demeaned) peer mean.

\paragraph{Instruments.}
We strengthen the first stage by exploiting an \emph{internal, deeper‑lag} instrument,
\[
Z_{t-2}\equiv \bar c_{-i,g,t-2},
\]
i.e., the leave–one–out peer mean two rounds earlier, which is predetermined with respect to period-$t$ shocks once fixed effects absorb player types and village$\times$round shocks. We also consider an over‑identified specification that adds the leave–one–village time‑varying shift–share instrument $Z_{\text{LOV}}$ (composition weights $\times$ lagged outside‑village shocks), as in our discussion of shift–shares above.\footnote{See identification discussion in the main text.}

\paragraph{Results.}
Table~\ref{tab:iv_updated} reports 2SLS estimates with player and village$\times$round FE. Using $Z_{t-2}$ alone (just‑identified), the first stage is extremely strong (cluster‑robust $F\approx 6{,}800$) and the second‑stage estimate is tightly centered around one:
\[
\hat\beta=1.008\;(\text{SE }0.0166).
\]
Adding $Z_{\text{LOV}}$ as a second instrument leaves the estimate unchanged within sampling error
($\hat\beta=1.007$, SE $0.0165$), the joint first stage remains very strong (cluster‑robust $F\approx 3{,}434$), and over‑identification is comfortable (Sargan $p=0.159$).
As a third check, a cross‑fitted “optimal IV’’ that predicts $\tilde{\bar c}_{-i,g,t-1}$ from $\{Z_{t-2},Z_{t-3},Z_{\text{LOV}}\}$ by lasso/ridge and uses the out‑of‑fold prediction as the instrument yields $\hat\beta=1.004$ (SE $0.015$) with an even larger first‑stage $F\approx 11{,}800$.

\begin{table}[t]
\centering
\caption{IV estimates with player and village$\times$round FE (updated). Cluster‑robust SEs (group).}
\label{tab:iv_updated}
\begin{tabular}{lccc}
\toprule
& (1) $Z_{t-2}$ & (2) $Z_{t-2}+Z_{\text{LOV}}$ & (3) CF--IV \\
\midrule
Peer regressor & $\tilde{\bar c}_{-i,g,t-1}$ & $\tilde{\bar c}_{-i,g,t-1}$ & $\tilde{\bar c}_{-i,g,t-1}$ \\
$\hat\beta$ & $1.008^{***}$ & $1.007^{***}$ & $1.004^{***}$ \\
SE (cluster = group) & $(0.0166)$ & $(0.0165)$ & $(0.0152)$ \\
Weak‑IV $F$ & $6{,}800$ (1 df) & $3{,}434$ (2 df, joint) & $11{,}774$ (1 df) \\
Sargan $p$ & n.a. & $0.159$ & n.a. \\
Obs. & $\sim 24{,}9$k & $\sim 24{,}9$k & $\sim 24{,}9$k \\
Fixed effects & \multicolumn{3}{c}{Player + Village$\times$Round (two‑way within)} \\
\bottomrule
\end{tabular}

\begin{flushleft}
\footnotesize Notes: Estimates are on two‑way demeaned variables. $Z_{t-2}$ is the leave‑one‑out peer mean two rounds earlier; $Z_{\text{LOV}}$ is a leave‑one‑village shift–share shifter constructed from LOO group shares and lagged cross‑village shocks. All $F$ statistics are cluster‑robust. CF--IV uses cross‑fitted lasso/ridge predictions of $\tilde{\bar c}_{-i,g,t-1}$ from $\{Z_{t-2},Z_{t-3},Z_{\text{LOV}}\}$. \\[-4pt]
\end{flushleft}
\end{table}

\paragraph{Diagnostics and validation.}
(i) A permutation test that shuffles $Z_{t-2}$ within village$\times$round cells yields a permutation $p\approx 0$ for the first‑stage $F$, indicating that the observed relevance is not a by‑chance artifact of within‑cell structure. (ii) A placebo at the earliest round with instrument variation (demeaned within village) is near zero. (iii) Over‑identification is passed in column (2). Together, these checks support the maintained exclusion restrictions; we therefore focus on the lagged‑peer effect with the internal deeper‑lag instrument.

\paragraph{Interpretation.}
Once player and village$\times$round shocks are absorbed, the contemporaneous peer correlation remains small. In contrast, exogenous variation in \emph{lagged} peer exposure has a large and precisely estimated effect on current giving. The updated designs substantially alleviate weak‑ID concerns present when using a time‑varying LOV instrument alone (KP $F\approx 6$ and wide AR intervals in the draft) by leveraging a high‑signal internal lag. The magnitude $\hat\beta\approx 1$ is consistent with strong one‑for‑one norm updating across rounds.\footnote{See the draft discussion of the earlier LOV‑only specification and weak‑ID diagnostics.}

\paragraph{Exogeneity \& timing.} Within-round imitation is mechanically ruled out (simultaneous play), and village$\times$round fixed effects absorb common shocks. Our internal deeper-lag instrument $Z_{t-2}$ is predetermined with respect to period-$t$ shocks after player and village$\times$round FE; the very strong first stage, a placebo at the earliest round with instrument variation, and over-identification tests when adding a leave-one-village shift–share $Z_{\text{LOV}}$ support the exclusion. Estimates are stable to cross-fitted optimal-IV, and weak-ID–robust Anderson–Rubin intervals contain the point estimates.

\paragraph{Link to bifurcation.} 
The nearly one‑for‑one lagged peer effect implies strong norm‑tracking across periods: when the group norm in $t-1$ rises by 1 L, individuals almost match that change in $t$. This dynamic re‑enforces the tipping and persistence facts: early advantages lock in via updating, consistent with the HMM’s high diagonal and the Fermi–Moran calibration’s moderate imitation intensity $\hat k
\approx 0.5$ and negative selection tilt $\hat d < 0$ that together reproduce persistence and downward drift. In short, the IV evidence provides the causal micro‑link that sustains the two‑basin macro pattern.

% --------------------------------------------------------------

% ------------------------------------------------------------------ %
%  DISCUSSION                                                        %
% ------------------------------------------------------------------ %
\section{Discussion}
\label{sec:discussion}

Three empirical facts organize our results. First, cooperation bifurcates early into two durable paths. Descriptively, (see Section~\ref{sec:data}), 
average giving declines from just over 7~L to roughly 6~L across ten rounds while cross–sectional variance increases and the contribution distribution becomes visibly bimodal after early rounds; clustering and a two–state HMM then confirm that most players settle into a High trajectory and a sizeable minority persist in a Low trajectory, with occasional escapes from Low rarer than drops from High. Second, entry into and persistence on the High path is strongly associated with social embeddedness and education. Each additional friend is linked to slightly higher giving and fewer high–to–low drops; education reduces simple repetition and increases alignment with the group; and men contribute more on average. Relative to Catholics, non–religious and Protestant participants tend to give more. After accounting for these factors, standard resource proxies show weak association with contributions. Third, round–to–round updating is primarily peer–responsive: people react more to what their peers did in the previous round than to their own last choice.

These within‑session patterns in Section~\ref{sec:data}, declining means, rising dispersion, and visible bimodality after round~3, align with the four empirical markers of early tipping point, asymmetric state persistence, a village‑level critical mass near 60\% and the empirical bifurcation results with Fermi-Moran process.

At the village level, a clear critical mass emerges: when roughly 60\% of early contributors are above the emerging norm, villages very likely finish high. This ties individual learning to collective outcomes and explains why small early advantages can lock in persistent differences across communities. 

%Early decisions are highly informative at the individual level as well, using only rounds~1–3 yields strong out–of–sample discrimination of who will end high. 

%These patterns are reinforced by complementary analyses of switching and learning during the session. Hazards are asymmetric ($P(L\!\to\!H){\approx}0.185$ vs.\ $P(H\!\to\!L){\approx}0.195$), drops occur earlier than rises, and friendship networks buffer high–to–low movement. 

These patterns are reinforced by complementary analyses. Updating in levels with player and village$\times$round fixed effects shows that individuals place substantially more weight on peers’ last–round contributions than on their own lagged choice; schooling lowers inertia and slightly increases peer responsiveness, and men are more socially responsive on average. 

Causally, instrumental‑variables estimates clarify that the contemporaneous peer coefficient collapses toward zero once player and village$\times$round shocks are absorbed, whereas the \emph{lagged} peer effect is large, approximately one‑for‑one, with very strong first stages, over‑identification passed, and cross‑fitted optimal‑IV confirming the magnitude. 

And a calibrated birth–death (Moran) imitation process provides a compact behavioral map: it reproduces persistence among high contributors and the observed rate of downward drift, while under–predicting stickiness in the Low state, pointing to frictions (e.g., imitation thresholds, reputational penalties, or heterogeneity in conformity) that bind low contributors more tightly in practice. 

Viewed against prior work, our evidence reinforces that early informational environments generate path dependence \citep{fellner2021information}; network embeddedness predicts sustained cooperation \citep{haan2006friendship,shirado2013quality}; gender differences in responsiveness and average giving are common \citep{furtner2021gender,seguino1996gender};  religious identity can shape prosocial choices \citep{benjamin2016religious}; and socioeconomic status on its own is not a reliable predictor of lower prosociality in the field \citep{andreoni2021higher,cardenas2000real}.

Our design prioritizes realism and within–session dynamics—fixed groups of five drawn from the same villages and ten repeated decisions with economically meaningful stakes. We did not vary group size, multiplier, or horizon, and stake size alone is unlikely to overturn the qualitative patterns we document \citep{kocher2008does}. Nonetheless, the convergence of descriptive dynamics, a village–level critical mass, and strong lagged peer responsiveness underscores a coherent behavioral picture of how cooperation persists or unravels in real communities. At the same time, results may differ with cross-village mixing, formal sanctions, or different multipliers; associations for non-wealth covariates (e.g., gender, religion) should be interpreted as descriptive; and the evolutionary calibration is likewise descriptive rather than structural. It summarizes persistence and asymmetry but omits heterogeneity in reputation or thresholds that likely underpins the extra stickiness of the Low state.

%early–warning prediction,  asymmetric hazards, 

\vspace{0.5em}

On the empirical side, designs that exogenously vary early peer exposure and interaction structure (e.g., randomized seeding within mapped networks) would sharpen causal interpretation of learning slopes and basin entry. On the theoretical side, enriching the evolutionary representation with heterogeneity and explicit reputation/threshold effects could reconcile the residual stickiness of Low while preserving tractability. More generally, the combination of mapped networks, within-session dynamics, and evolutionary calibration provides a portable template for diagnosing and shifting cooperation basins in other collective-action settings. These results likely generalize to repeated local public‑goods settings in small, familiar groups with meaningful stakes; effects may differ under cross‑village mixing or formal enforcement institutions. 

Our findings have some implications. Because early peer environments are strongly predictive and a village-level critical mass is apparent, light-touch policies that (i) seed early contributions through socially central connectors, (ii) increase the visibility of cooperative acts across rounds, and (iii) bolster learning environments that help people align updates with pro-social local norms are natural candidates for moving groups into, and keeping them in, the High basin. 

%The evidence also points to the value of strengthening education, which is consistently associated with higher giving and less inertia. 

\bibliographystyle{apalike}
\bibliography{references}

\section*{Acknowledgments}	
	
We thank our field team in Honduras; the Honduras Ministry of Health; and our implementation partners, the InterAmerican Development Bank, World Vision Honduras, Child Fund Honduras, and Dimagi. Rennie Negron supervised the overall execution of the field data collection and Liza Nicoll managed the data. Mark McKnight and Wyatt Israel developed the Trellis software to map networks and collect the data. We are grateful to Tom Keegan for management oversight. 

\section*{Funding}

This research was supported by the Bill and Melinda Gates Foundation, with additional support from the NOMIS Foundation, the Pershing Square Foundation, Paul Graham, and NIH R01AG062668. Support was also partly provided by NIH R01MH134715, NSF awards CAREER IIS-1149662 and IIS-1409177, and by ONR awards YIP N00014-14-1-0485 and N00014-17-1-2131.  

\section*{Author contributions}

Conceptualization: MP, FF and NAC; Methodology: MP, FF and NAC; Statistical analysis performance: MP; Writing: MP, FF, and NAC; Funding acquisition: NAC.

\section*{Competing interests}

The authors declare that they have no competing interests.

\section*{Data and Code Availability}

All data and code used to produce the results in this paper will be deposited in an AEA-compliant trusted repository prior to publication. The replication package will include de-identified data, programs to construct the analysis datasets from the raw data, and programs to reproduce all tables and figures in the paper. 

\section*{Ethics}
The underlying field trial was registered with ClinicalTrials.gov, number NCT02694679. The Yale IRB and the Honduran Ministry of Health approved all data collection procedures (Protocol \# 1506016012) and all participants provided informed consent.

% Number eqs/figs/tables within appendix sections: (A.1), Fig. A.1, Table A.1
\numberwithin{equation}{section}
\counterwithin{figure}{section}
\counterwithin{table}{section}

% Subsections as A1, A2 (no dot); sections remain A, B, C...
\makeatletter
\renewcommand{\thesubsection}{\thesection\arabic{subsection}}
\makeatother

% Print "Appendix A" (not just "A") in the section heading label
\titleformat{\section}
  {\normalfont\Large\bfseries}
  {\thesection}{1em}{}

% If you use \autoref from hyperref, make it say "Appendix"
\providecommand*{\sectionautorefname}{}

% --- Reusable "Online Appendix" cover macro ---
\newcommand{\makeonlineappendixtitle}[3]{%
  \begin{titlepage}
    \centering
    {\Large\bfseries ONLINE APPENDIX FOR\par}
    \vspace{1.2ex}
    {\LARGE\itshape #1\par} % main paper title
    \vspace{1ex}
    {\large #2\par}         % authors
    \vfill
    {\normalsize #3\par}    % date or version
  \end{titlepage}%
}

\clearpage
\appendix              % switch sections from 1,2,3 to A,B,C,...

% --- Appendix-only numbering style: subsections as A1, B1 (no dot) ---
\makeatletter
\renewcommand{\thesubsection}{\thesection\arabic{subsection}} % A1, A2, ...
\renewcommand{\p@subsection}{\thesection}                     % keeps refs tidy
\makeatother

% --- Optional: make section headings print "Appendix A: <title>" ---
\titleformat{\section}{\normalfont\Large\bfseries}{Appendix \thesection:}{1em}{}

% Optional: add a ToC line for the whole appendix block
\addcontentsline{toc}{part}{Online Appendix}

% ---------- ONLINE APPENDIX COVER PAGE ----------
\makeonlineappendixtitle
  {Group Cooperation Diverges onto Durable Low versus High Paths: Public Goods Experiments in 134 Honduran Villages}
  {Marios Papamichalis \quad Nicholas A.\ Christakis \quad Feng Fu}

% immediately after \appendix and just before Appendix A starts:
\renewcommand{\theprop}{\arabic{prop}} % show 1, 2, ...
\setcounter{prop}{0}                   % reset so the next prop is 1

\section{Theory \& Proofs}\label{app:setting}
\subsection{Notation}

%--------------------------------------------------------
%--------------------------------------------------------

\renewcommand{\arraystretch}{1.15}
\begin{table}[H]\centering\small
\begin{tabular}{@{}ll@{}}
\toprule
\textbf{Symbol} & \textbf{Meaning} \\
\midrule
\multicolumn{2}{l}{\emph{Indices and counts}}\\
$ i, g, v, t $          & Player, group, village, round indices. \\
$ T=10 $                & Number of rounds per session. \\
$ N=5 $                 & Group size. \\
\midrule
\multicolumn{2}{l}{\emph{Stage game and preferences}}\\
$ c_{igt}\in[0,12] $    & Contribution (Lempiras) of player $i$ in group $g$, round $t$. \\
$ \peer_{-i,gt} := \frac{1}{N-1}\sum_{j\neq i} c_{jgt} $ 
                        & Leave-one-out peer mean in round $t$ (demeaned version $\tilde{\peer}_{-i,gt}$). \\
$ b>0,\ \kappa>0 $      & Public return marginal benefit; private marginal cost (A1: $b>\kappa$ and $b/N<\kappa$). \\
$ \pi_i^{t} $           & Instantaneous utility/payoff in round $t$. \\
$ d_i>0,\ \alpha\in(0,1)$ & Altruism weight and curvature (we set $\alpha=\tfrac{1}{2}$ in estimation). \\
$ h_i\in[0,1],\ \knorm>0 $ & Norm salience and curvature of norm penalty in the utility. \\
$ \phi_i := 2\,\knorm h_i $ & Effective norm–pull (identified product). \\
$  c^{*} $     & Singular strategy (see Prop.~\ref{prop:css}). \\
$ c^{*} $               & Empirical tipping point from drift $\mathbb{E}[\Delta c\mid c]$. \\
$ \hat c $              & Round-1 empirical mean; threshold for $H/L$ classification. \\
\midrule
\multicolumn{2}{l}{\emph{Panel/IV notation}}\\
$ X_{igt} $             & Control vector. \\
$ \tilde{z}_{igt} $     & Two–way demeaning (round and village): $z-\bar z_{\cdot\cdot t}-\bar z_{v\cdot\cdot}+\bar z_{\cdot\cdot\cdot}$. \\
$ \beta,\ \gamma $      & Second-stage peer and control coefficients. \\
$ \varepsilon_{igt},\ u_{igt} $ & Second-stage and first-stage errors. \\
$ Z^{(l)}_{ig} $        & LOO composition IV: mean of trait $l$ among $i$’s groupmates. \\
$ s^{(l)}_{g,-i} $      & LOO share of trait $l$ in $i$’s group. \\
$ \mu^{(l)}_{t-1} $     & Cross-village round-$(t\!-\!1)$ mean (shifter) for trait $l$. \\
$ Z^{SS}_{igt} $        & Shift–share IV: $ \sum_l s^{(l)}_{g,-i}\,\mu^{(l)}_{t-1} $. \\
\midrule
\multicolumn{2}{l}{\emph{States and learning}}\\
$ s_{it}\in\{L,H\} $    & Player-time state (Low/High). \\
\midrule
\multicolumn{2}{l}{\emph{Evolutionary calibration (binary Fermi model)}}\\
$ \wutil_i=1+\delta\,\pi_i^{t} $ & Utility-based Moran fitness (weak selection $\delta\in(0,1)$). \\
$ \wbin_i(\,s_i\,)=1+\dtilt\cdot\mathbf{1}\{s_i=H\} $
                        & Binary fitness with tilt $\dtilt$ favoring $H$. \\
$ \kfermi>0 $           & Fermi imitation intensity. \\
$ T_{\mathrm{emp}},\ T_{\mathrm{sim}}(\dtilt,\kfermi) $
                        & Empirical and simulated $H/L$ transition matrices. \\
$ \mathrm{RSS}=\|T_{\mathrm{sim}}-T_{\mathrm{emp}}\|_F^2 $ 
                        & Calibration loss (Frobenius norm). \\
\bottomrule
\end{tabular}
\end{table}

For notational consistency, we write and use \( c_{igt} \), which is equivalent to \( c_{i,g,t} \). Commas in subscripts are included only when the indices \( i \), \( g \), or \( t \) take values of \(-1\) or \(-2\).

\vspace{1ex}
\noindent\textbf{Instantaneous utility (stage game, round $t$).}
\[
\pi_i^t
=\underbrace{\frac{b}{N}\big(c_{i,t}+(N-1)\bar c_{-i,t}\big)-\kappa\,c_{i,t}}_{\text{material payoff}}
\;+\;\underbrace{d_i\,(c_{i,t})^{\alpha}}_{\text{altruism}}
\;-\;\underbrace{h_i\,\exp\!\big[-\knorm\,(c_{i,t}-\bar c_{-i,t-1})^2\big]}_{\text{aversion to being at the norm}}\
\]
, where $\bar c_{-i,t-1}$ is the previous-round peer norm.

\noindent\textbf{Best-reply first-order condition.}
\[
D(c)\;=\;\left(\frac{b}{N}-\kappa\right)\;+\;d_i\,\alpha\,c^{\alpha-1}
\;+\;2\knorm\,h_i\,(c-\bar c_{-i,t-1})\,
\exp\!\big[-\knorm(c-\bar c_{-i,t-1})^2\big]\;=\;0.
\]

\begin{prop}[Convergence stability]\label{prop:css}
Assume A1 ( $b>\kappa$ and $b/N<\kappa$ ) and $\alpha\in(0,1)$. Consider the adaptive-dynamics selection gradient for a monomorphic resident $c$,
\[
D(c)\;=\;\left.\frac{\partial \pi_i^t(c_m;\,c)}{\partial c_m}\right|_{c_m=c}
=\Big(\tfrac{b}{N}-\kappa\Big)\;+\;d_i\,\alpha\,c^{\alpha-1},
\]
where the aversion to being at the norm term vanishes at $c_m=c$ because $\partial/\partial c_m\big[-h_i\exp\{-\knorm(c_m-c)^2\}\big] = +2\knorm h_i(c_m-c)\exp\{-\knorm (c_m-c)^2\}=0$ at $c_m=c$.
Then $D$ has a unique zero $c^*\in(0,12]$ given by
\[
c^{*}
=\Big[\big(\kappa-\tfrac{b}{N}\big)/(d_i\alpha)\Big]^{\!1/(\alpha-1)},
\]
and $c^*$ is a convergence-stable singular strategy.
\end{prop}

\begin{proof}
Because $\alpha\in(0,1)$, we have $c^{\alpha-1}$ strictly decreasing and $\lim_{c\downarrow0} c^{\alpha-1}=+\infty$, $\lim_{c\uparrow\infty} c^{\alpha-1}=0$. With $b/N<\kappa$, the constant term satisfies $(b/N-\kappa)<0$. Hence
\[
\lim_{c\downarrow0}D(c)=+\infty,\qquad \lim_{c\uparrow\infty}D(c)=\tfrac{b}{N}-\kappa<0,
\]
and by continuity there exists at least one zero. Uniqueness follows from strict monotonicity:
\[
D'(c)=d_i\,\alpha(\alpha-1)\,c^{\alpha-2}<0\qquad\text{for all }c>0.
\]
Solving $D(c)=0$ yields the stated formula for $c^*$. Convergence stability in adaptive dynamics requires $D'(c^*)<0$ (the CSS condition), which holds because $D'$ is negative everywhere. Finally, the endowment cap $12$ enforces $c^*\in(0,12]$.
\end{proof}

\begin{prop}[ESS vs.\ branching]\label{prop:ess}
At the singular strategy $c^*$ of Proposition~\ref{prop:css}, the curvature of the invasion fitness with respect to a mutant's own trait is
\[
\left.\frac{\partial^2 \pi_i^t(c_m;\,c)}{\partial c_m^2}\right|_{c_m=c=c^*}
\;=\; d_i\,\alpha(\alpha-1)\,(c^*)^{\alpha-2}\;+\;2\knorm \,h_i.
\]
Hence $c^*$ is an ESS iff this is negative; if it is positive (while $D'(c^*)<0$), then $c^*$ is convergence-stable but invasible, i.e., an evolutionary branching point.
\end{prop}

\begin{proof}
Write the instantaneous utility as
\[
\pi_i^t(c_m;\,c)=\Big(\tfrac{b}{N}-\kappa\Big)c_m\;+\;\tfrac{b(N-1)}{N}\,c
\;+\;d_i\,c_m^{\alpha}\;-\;h_i\exp\!\big\{-\knorm (c_m-c)^2\big\}.
\]
The first derivative with respect to $c_m$ is
\[
\frac{\partial \pi}{\partial c_m}
=\Big(\tfrac{b}{N}-\kappa\Big)\;+\;d_i\,\alpha\,c_m^{\alpha-1}
\;+\;2\knorm \,h_i(c_m-c)\,\exp\!\big\{-\knorm (c_m-c)^2\big\}.
\]
At $c_m=c$, the last term vanishes, giving the selection gradient in Prop.~\ref{prop:css}. The second derivative with respect to $c_m$ is
\[
\frac{\partial^2 \pi}{\partial c_m^2}
=d_i\,\alpha(\alpha-1)\,c_m^{\alpha-2}
\;+\;2\knorm \,h_i\cdot\Big[\tfrac{\mathrm{d}}{\mathrm{d}c_m}\big((c_m-c)\,e^{-\knorm (c_m-c)^2}\big)\Big].
\]
The bracket equals $1$ at $c_m=c$ (since $\frac{\mathrm{d}}{\mathrm{d}x}[x e^{\knorm x^2}]\big|_{x=0}=1$), hence
\[
\left.\frac{\partial^2 \pi}{\partial c_m^2}\right|_{c_m=c}
= d_i\,\alpha(\alpha-1)\,c^{\alpha-2}+2\knorm \,h_i.
\]
Evaluating at $c=c^*$ yields the stated expression. In canonical adaptive dynamics, a CSS point is an ESS if this curvature is negative (local fitness maximum); if it is positive while $D'(c^*)<0$, the point is convergence-stable but invasible, leading to evolutionary branching.
\end{proof}

\noindent\textbf{Moran birth–death dynamics (weak selection).}
At each microscopic update:
\[
\Pr(\text{reproduce } i)\;\propto\;w_i\;=\;1+\delta\,\pi_i^t, \qquad
\text{offspring replaces a uniformly chosen peer.}
\]

\noindent\textbf{Fermi-rule calibration (binary fitness).}
Let $\theta_i \in \{H,L\}$ denote High/Low contribution states. Define
\[
w_i(\theta_i)=
\begin{cases}
1+d, & \theta_i=H \ \text{(High)},\\[2pt]
1,   & \theta_i=L \ \text{(Low)}.
\end{cases}
\]
Selection (multinomial Fermi):
\[
\Pr(\text{reproduce } i)=\frac{\exp(k\,w_i)}{\sum_{\ell}\exp(k\,w_\ell)}.
\]
Alternatively, pairwise Fermi:
\[
\Pr(i\text{ replaces }j)=\frac{1}{1+\exp\!\big(-k\,[w_i-w_j]\big)}.
\]
Calibrate $(d,k)$ by
\[
\mathrm{RSS}=\big\|T_{\mathrm{sim}}(d,k)-T_{\mathrm{emp}}\big\|_F^2.
\]

\section{Additional data details - Tables \& Figures}\label{app:data}

This appendix documents the sample composition and key distributional facts that underlie the main analyses. We report continuous/count variables in Table~\ref{tab:contsum}, binary indicators in Table~\ref{tab:bin}, and multi–category variables in Table~\ref{tab:cat}. Figure~\ref{fig:Per_Round} summarizes how the mass of contributions moves across ranges over the ten rounds, and Table~\ref{tab:contribution_by_round} provides corresponding headcounts by round and type. All the percentages are balanced and indicative to the real percentages of the population in Honduran Villages (see \cite{Oles2024Maternal, airoldi2024induction,papamichalis2025educational} for comparison).

The panel contains $25{,}910$ player–round decisions from $2{,}591$ participants (fixed groups of five over $T\!=\!10$ rounds). Average contributions are $6.355$ Lempiras (SD $3.452$), with individual covariates covering demographics (age, gender, marital status, indigenous status, schooling), identity (religion), social ties (friends, adversaries), and village characteristics (network densities, network size, access routes); Tables~\ref{tab:contsum}--\ref{tab:cat} report sample composition and the scale of key covariates used throughout, while Figure~\ref{fig:Per_Round} visualizes within-session redistribution of mass across contribution bands. Figure~\ref{fig:Per_Round} shows a gradual reallocation of mass: the $0$--$3$~L range grows over time, the $8$--$10$~L range contracts, and the $4$--$7$~L and $11$--$12$~L bands remain relatively flat. Counts by round (Table~\ref{tab:contribution_by_round}) mirror this pattern: free riders rise from $165$ (round 1) to $224$ (round 10), while full contributors remain comparatively stable (about $218$--$237$ per round). The number of low contributors increases from $1{,}064$ to $1{,}341$, and high contributors decrease from $1{,}527$ to $1{,}250$ by round 10. Missingness is low and not systematically related to baseline gender, age, schooling, or friends.

\begin{table}[H]
\caption{Summary statistics (continuous/count variables).}
\label{tab:contsum}
\centering
\begin{tabular}{llrrrrr}
\toprule
variable & type & n & mean & sd & min & max\\
\midrule
Contributing & continuous & 25910 & 6.355 & 3.452 & 0.000 & 12.000\\
Age & continuous & 25880 & 36.792 & 15.761 & 14.000 & 89.000\\
Friends & count & 25900 & 7.027 & 4.514 & 0.000 & 35.000\\
Adversaries & count & 25900 & 0.776 & 1.278 & 0.000 & 10.000\\
Access routes & continuous & 25910 & 1.902 & 0.714 & 1.000 & 5.000\\
Friendship density & continuous & 25910 & 0.019 & 0.010 & 0.003 & 0.059\\
Adversarial density & continuous & 25910 & 0.002 & 0.002 & 0.000 & 0.010\\
Network size & count & 25910 & 193.747 & 118.230 & 54.000 & 668.000\\
\bottomrule
\end{tabular}
\end{table}

% Requires: \usepackage{booktabs}  % you already have this
% Optional: \usepackage{float} for [H], but side-by-side works better with [!htbp]

\begin{table}[!htbp]
\centering
\begin{minipage}[t]{0.48\textwidth}
\centering
\caption{Binary variables: counts and proportions by level.}
\label{tab:bin}
\footnotesize
\begin{tabular}{llrr}
\toprule
variable & value & n & pct\\
\midrule
Gender & 0 & 15280 & 0.590\\
Gender & 1 & 10630 & 0.410\\
Food Insecurity & 0 & 13900 & 0.561\\
Food Insecurity & 1 & 10880 & 0.439\\
Marital Status & 1 & 17140 & 0.662\\
Marital Status & 0 & 8740 & 0.338\\
Indigenous & 0 & 21740 & 0.872\\
Indigenous & 1 & 3180 & 0.128\\
\bottomrule
\end{tabular}
\end{minipage}\hfill
\begin{minipage}[t]{0.48\textwidth}
\centering
\caption{Categorical variables: counts and proportions by category.}
\label{tab:cat}
\footnotesize
\begin{tabular}{llrr}
\toprule
variable & value & n & pct\\
\midrule
Education & 0 & 6350 & 0.254\\
Education & 1 & 1820 & 0.073\\
Education & 2 & 3430 & 0.137\\
Education & 3 & 4160 & 0.167\\
Education & 4 & 2090 & 0.084\\
Education & 5 & 1180 & 0.047\\
Education & 6 & 4550 & 0.182\\
Education & 8 & 940 & 0.038\\
Education & 11 & 330 & 0.013\\
Education & 13 & 110 & 0.004\\
Religion & 0 (None) & 2300 & 0.092\\
Religion & 1 (Protestant) & 8420 & 0.337\\
Religion & 2 (Catholic) & 14240 & 0.571\\
\bottomrule
\end{tabular}
\end{minipage}
\end{table}

\begin{figure}[H]
  \centering
  \includegraphics[width=0.9\textwidth]{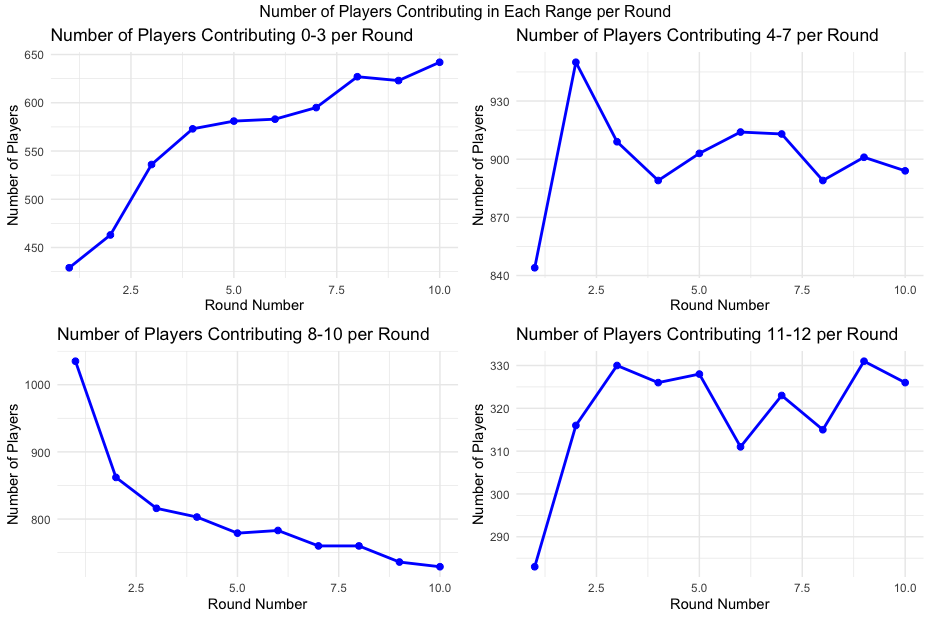}
  \caption{\textbf{Head-counts by contribution range.} Growth in 0–3 L, contraction in 8–10 L; 4–7 L and 11–12 L bands comparatively flat.}
  \label{fig:Per_Round}
\end{figure}

\begin{table}[H]
\centering
\begin{tabular}{ccccc}
\hline
Round & Free Riders & Full Contributors & Low Contributors & High Contributors \\
\hline
1  & 165 & 223 & 1064 & 1527 \\
2  & 178 & 237 & 1186 & 1405 \\
3  & 175 & 226 & 1237 & 1354 \\
4  & 186 & 218 & 1241 & 1350 \\
5  & 192 & 228 & 1278 & 1313 \\
6  & 202 & 219 & 1290 & 1301 \\
7  & 204 & 220 & 1285 & 1306 \\
8  & 212 & 217 & 1327 & 1264 \\
9  & 230 & 229 & 1317 & 1274 \\
10 & 224 & 227 & 1341 & 1250 \\
\hline
\end{tabular}
\caption{Counts by contribution type and round.}
\label{tab:contribution_by_round}
\end{table}

\section{Bifurcation diagnostics}\label{app:bifurcation}

\subsection{Early-warning prediction from rounds 1–3}
\label{sec:earlywarning}

We ask whether a few early decisions are sufficient to forecast who will end in the High path by round~10. Using only a player’s own choices from rounds~1–3, we fit an out-of-sample logistic model with three interpretable predictors: the early \emph{mean} contribution (level), the early \emph{standard deviation} (stability), and the early \emph{linear slope} (trend). This simple specification attains strong discrimination (AUC $\approx 0.829$ \ref{fig:Per_Round_}). Effect sizes are large for level and modest for stability: each additional lempira in the early mean multiplies the odds of finishing High by about $1.85$ (95\% CI $[1.73,\,1.98]$), whereas greater early variability reduces those odds (OR $\approx 0.76$, $p<0.01$). The early slope is only marginally informative beyond level (OR $\approx 1.09$, $p=0.059$). At a threshold of $0.517$, sensitivity is $0.84$ and specificity is $0.70$ (see table \ref{tab:earlywarning}).\\

\begin{figure}[H]
  \centering
  \includegraphics[width=0.9\textwidth]{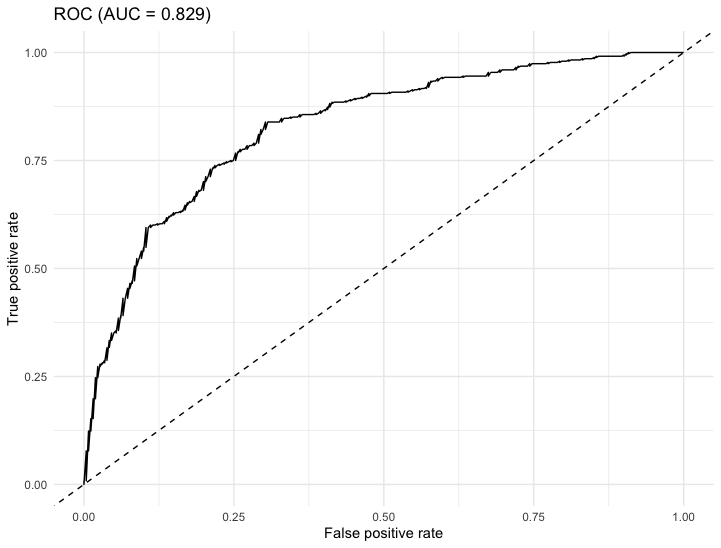}
  \caption{The separation into high- and low-cooperation paths is visible under a range of reasonable “high” cutoffs and emerges early: using only the first to third rounds predicts final paths with Area Under the Curve (AUC) $\approx 0.829$, indicating high predictive accuracy.}
  \label{fig:Per_Round_}
\end{figure}

These results show that three rounds suffice to identify likely long-run cooperators with useful accuracy. Substantively, the dominance of the early \emph{level} corroborates the tipping interpretation: small initial advantages in contribution behavior are highly persistent. The negative association of early variability indicates that unstable starters are less likely to lock into the High path, even when their average is comparable. Practically, this provides an operational early-warning rule for within-session targeting (e.g., directing recognition, feedback, or network seeding toward borderline groups after round~3) while keeping the model simple and transparent.

\begin{table}[H]
  \centering
  \caption{Early-warning GLM (odds ratios). Early mean dominates; early variability reduces odds; early slope borderline positive.}
  \label{tab:earlywarning}
  \begin{tabular}{lccc}
    \toprule
    Predictor & OR & 95\% CI & \(p\) \\
    \midrule
    Early mean (per L) & 1.85 & [1.73, 1.98] & $<10^{-6}$ \\
    Early SD           & 0.83 & [0.76, 0.91] & $<0.01$ \\
    Early slope        & 1.09 & [0.997, 1.20] & 0.059 \\
    \addlinespace
    Sensitivity @ 0.517 & 0.84 & --- & --- \\
    Specificity @ 0.517 & 0.70 & --- & --- \\
    \bottomrule
  \end{tabular}
\end{table}

\subsection{Integrated Markov logit}

% ==============================================================
\subsubsection*{Dynamics of high/low transitions: discrete hazards, duration, and robustness}

% High Low  mod_rise, mod_drop 

We dichotomize round-$t$ contributions at the round‑1 mean ($c^\star$) and define $s_t\in\{0(L),1(H)\}$. The panel yields $N_{\text{drop}}=11{,}529$ at–risk observations for drops $H\!\to\!L$ and $N_{\text{rise}}=10{,}710$ for rises $L\!\to\!H$ across 134 villages. Raw transition counts are in Table~\ref{tab:counts}. The unconditional hazards are $P(H\!\to\!L)=0.195$ and $P(L\!\to\!H)=0.185$.

% ---------------------- Mathematical specification ----------------------
\paragraph{State definition and peer signal.}
Let $c_{it}$ be player $i$'s contribution in round $t$. The high/low state is
\begin{align}
  c^\star &= \frac{1}{N_1}\sum_{i: t=1} c_{i1}, \qquad
  s_{it} = \mathbb{I}\{c_{it} \ge c^\star\} \in \{0(L),\,1(H)\}. \label{eq:state}
\end{align}
Let $g(i,t)$ denote $i$'s group in round $t$ with size $G_{g(i,t),t}$. The leave--one--out peer mean (scaled by the endowment $E=12$) is
\begin{align}
  \bar c^{(-i)}_{g(i,t),t} &= \frac{1}{G_{g(i,t),t}-1}\!\!\sum_{j\in g(i,t),\,j\neq i} c_{jt},
  \qquad m_{it}=\frac{\bar c^{(-i)}_{g(i,t),t}}{E}, \qquad x_{i,t-1}=m_{i,t-1}.
\end{align}

\paragraph{Duration and at–risk indicators.}
Let the time already spent in the current regime at $t-1$ be
\begin{align}
  d_{i,t-1} = \#\{ \tau\le t-1 : s_{i,\tau}=s_{i,t-1}\text{ consecutively back from }t-1\},
\end{align}
At–risk outcomes are
\begin{align}
  y^{\text{drop}}_{it} &= \mathbb{I}\{s_{i,t-1}=1,\, s_{it}=0\},\qquad
  y^{\text{rise}}_{it}  = \mathbb{I}\{s_{i,t-1}=0,\, s_{it}=1\}. \label{eq:at-risk}
\end{align}

\paragraph{Discrete–time GLMMs (separate hazards).}
For players with $s_{i,t-1}=1$ (drop risk, $H\!\to\!L$) we estimate
\begin{align}
  \mathrm{logit}\,\Pr\!\big(y^{\text{drop}}_{it}=1 \mid s_{i,t-1}=1,\,X_{it}\big)
  &= \alpha_{v(i)}^{\text{D}} + \delta_{g(i)}^{\text{D}} + a_i^{\text{D}}
     + \beta^{\text{D}}_{\text{peer}}\,x_{i,t-1}
     + \beta^{\text{D}}_{r}\,r_t
     + f^{\text{D}}(d_{i,t-1})
     + \Gamma^{\text{D}^\top} Z_i, \label{eq:drop-glmm}
\end{align}
analogously for rises when $s_{i,t-1}=0$ ($L\!\to\!H$),
\begin{align}
  \mathrm{logit}\,\Pr\!\big(y^{\text{rise}}_{it}=1 \mid s_{i,t-1}=0,\,X_{it}\big)
  &= \alpha_{v(i)}^{\text{R}} + \delta_{g(i)}^{\text{R}} + a_i^{\text{R}}
     + \beta^{\text{R}}_{\text{peer}}\,x_{i,t-1}
     + \beta^{\text{R}}_{r}\,r_t
     + f^{\text{R}}(d_{i,t-1})
     + \Gamma^{\text{R}^\top} Z_i. \label{eq:rise-glmm}
\end{align}
where superscript D and R stand for drop and rise, respectively. Here $v(i)$ indexes village, $g(i)$ the (fixed) game group, $a_i^{(\cdot)}$ a player random intercept, $r_t$ the round index,
and $Z_i$ collects baseline covariates (gender, marital status, religion, ethnicity, religion, food insufficiencies, access routes, and $z$–scored
age, friends, adversaries, network density/size and education). Duration is flexibly captured by
\begin{align}
  f^{(\cdot)}(d) \;=\; \sum_{k=1}^{K}\delta^{(\cdot)}_k\,B_k(d), \qquad K=3, \label{eq:spline}
\end{align}
a natural cubic spline in $d_{i,t-1}$. Random intercepts $\alpha_{v(i)}^{(\cdot)}$, $\delta_{g(i)}^{(\cdot)}$, and $a_i^{(\cdot)}$
allow for village, group, and individual heterogeneity.

\paragraph{Integrated Markov logit.}
Pooling both directions,
\begin{align}
  \mathrm{logit}\,\Pr(s_{it}=1 \mid s_{i,t-1},X_{it})
  &= \alpha_{v(i)} + \delta_{g(i)} + a_i + \theta\,s_{i,t-1}
     + \big(\beta_{\text{peer}} + \tilde\beta_{\text{peer}}\,s_{i,t-1}\big)x_{i,t-1} \nonumber\\
  &\quad + \big(\Gamma + \tilde\Gamma\,s_{i,t-1}\big)^\top Z_i
     + \big[f_0(d_{i,t-1}) + s_{i,t-1}\,f_1(d_{i,t-1})\big], \label{eq:markov-logit}
\end{align}
with village, group, and player random intercepts $(\alpha_{v(i)},\delta_{g(i)},a_i)$.
Average (row–mean) transition probabilities used in Figure~\ref{fig:calibration_markov_hmm} are
\begin{align}
  \widehat P_{q\to 1} = \frac{1}{n_q}\sum_{(i,t): s_{i,t-1}=q} \widehat p_{it},\qquad
  \widehat P_{q\to 0} = 1-\widehat P_{q\to 1},\qquad q\in\{0(L),1(H)\}. \label{eq:rowmeans}
\end{align}
The stationary share is $\pi_H=\widehat P_{L\to H}/\big(\widehat P_{L\to H}+\widehat P_{H\to L}\big)$.

\paragraph{HMM cross–check.}
With latent state $S_{it}\in\{1,2\}$ and $z$–scored contribution $Y_{it}$,
\begin{align}
  \Pr(S_{it}=s \mid S_{i,t-1}=r) = T_{rs}, \qquad
  Y_{it}\mid S_{it}=s \sim \mathcal N(\mu_s,\sigma_s^2). \label{eq:hmm}
\end{align}
The Viterbi path $\widehat S_{it}$ yields
\begin{align}
  \widehat T_{rs} = \frac{\sum_{it}\mathbb{I}\{\widehat S_{i,t-1}=r,\widehat S_{it}=s\}}
                          {\sum_{it}\mathbb{I}\{\widehat S_{i,t-1}=r\}}. \label{eq:hmm-T}
\end{align}

\paragraph{Survival and evaluation.}
Time to first switch is modeled by a Cox model with village frailty
\begin{align}
  \lambda_i(t \mid x_i,u_{v(i)}) = \lambda_0(t)\,\exp(x_i^\top\beta + u_{v(i)}), \quad
  u_{v}\sim \text{Gamma}(\theta^{-1},\theta^{-1}). \label{eq:cox-frailty}
\end{align}
For predictive performance: $\text{Brier}=\frac{1}{n}\sum_i (y_i-\hat p_i)^2$; AUC $=\Pr(\hat p^>\!>\hat p^-)$; calibration regresses $\mathrm{logit}(y_i)$ on $\mathrm{logit}(\hat p_i)$, ideal $(\alpha_c,\beta_c)=(0,1)$.

\paragraph{Population–averaged robustness.}
As a GEE alternative, we report binomial GLMs with player–clustered (sandwich) variance
\begin{align}
  \widehat{\mathrm{Var}}_{\text{CR}}(\hat\beta)
    = (X^\top W X)^{-1} \Big( \sum_{g} U_g U_g^\top \Big) (X^\top W X)^{-1},\quad
    U_g = \sum_{i\in g} x_i\,w_i^{1/2}(y_i-\hat\mu_i). \label{eq:sandwich}
\end{align}

% ---------------------- Results ----------------------
\paragraph{Main results.}
The scaled peer signal $x_{i,t-1}$ (per‑endowment change) dominates (Table~\ref{tab:glmm_or}): it lowers drop risk (OR $=0.0027$, $p<10^{-180}$) and raises rise risk (OR $\approx226$, $p<10^{-140}$). Age decreases both hazards (drop: OR $=0.917$, $p=0.012$; rise: OR $=0.890$, $p<0.001$). Education increases switching in both directions (drop: OR $=1.15$, $p<0.001$; rise: OR $=1.13$, $p<0.001$). More educated participants are more exploratory: they are more likely to transition from low to high contributions (it is currently in the main text), AND they also exhibit slightly stronger responsiveness to their peers (SI Appendix~\ref{sec:switch_learn}). For rise $L\!\to\!H$, men are more likely to rise (OR $=1.44$, $p<10^{-9}$), and more adversaries slightly raise the rise hazard (OR $=1.06$, $p=0.028$). Relative to Catholics, those with no religion are more likely to rise (OR $=1.29$, $p=0.010$).

\paragraph{Model fit and calibration.}
Fixed–effects predictions give AUCs of 0.711 (drop) and 0.702 (rise), Brier 0.139 and 0.137, and calibration close to ideal (intercept/slope: drop $-0.025/0.977$; rise $-0.031/0.974$); see Figure~\ref{fig:calibration_markov_hmm}\,\subref{fig:calibration_markov_hmm:drop} and 
\subref{fig:calibration_markov_hmm:rise}.

\paragraph{Duration.}
Duration profiles show a mild early decline then shallow increase for drops
(Figure~\ref{fig:marginal+duration}\,\subref{fig:marginal+duration:drop_duration})
and a hump for rises peaking around 5–6 rounds
(Figure~\ref{fig:marginal+duration}\,\subref{fig:marginal+duration:rise_duration}).

\paragraph{Heterogeneity.}
The drop model exhibits a friends$\times$gender interaction
(Figure~\ref{fig:marginal+duration}\,\subref{fig:marginal+duration:friends_gender}).
The rise pattern by religion is shown in
Figure~\ref{fig:marginal+duration}\,\subref{fig:marginal+duration:religion_rise}.

\paragraph{Transition system.}
The integrated logit implies
\[
\widehat{\mathbf P}_{\text{logit}}
=\begin{bmatrix}
0.608 & 0.392\\
0.372 & 0.628
\end{bmatrix},
\]
with stationary share $\pi_H\simeq0.513$ (Figure~\ref{fig:calibration_markov_hmm}\,\subref{fig:calibration_markov_hmm:markov}). A 2–state Gaussian HMM yields
\[
\widehat{\mathbf P}_{\text{HMM}}
=\begin{bmatrix}
0.759 & 0.241\\
0.165 & 0.835
\end{bmatrix},
\]
stationary mass $\approx0.594$, and a gradual decline in the high–state share (Figure~\ref{fig:calibration_markov_hmm}\,\subref{fig:calibration_markov_hmm:hmm_share}).

\paragraph{Robustness.}
Cox models with village frailty show weak duration dependence; KM curves by friends terciles are overlapping (Appendix Fig.~A.1). Probit GLMMs reproduce patterns (Table~\ref{tab:glmm_or}, notes). Cluster–robust GLMs agree with the mixed–effects results (Appendix Table~A.1). Varying $c^\star\in\{\text{mean},\text{median},\text{mean}\pm0.5\,\text{sd}\}$ leaves conclusions unchanged (Appendix Table~A.2).

\paragraph{Placebo.}
Adding the lead peer signal $x_{i,t+1}$ to a model for $s_t$ yields a large positive coefficient ($\approx 3.52$, $p<10^{-30}$), consistent with serial correlation; peer coefficients are predictive/associational.

\begin{figure}[htbp]
  \centering

  % --- top row ---
  \begin{subfigure}{0.49\textwidth}
    \centering\includegraphics[width=\linewidth]{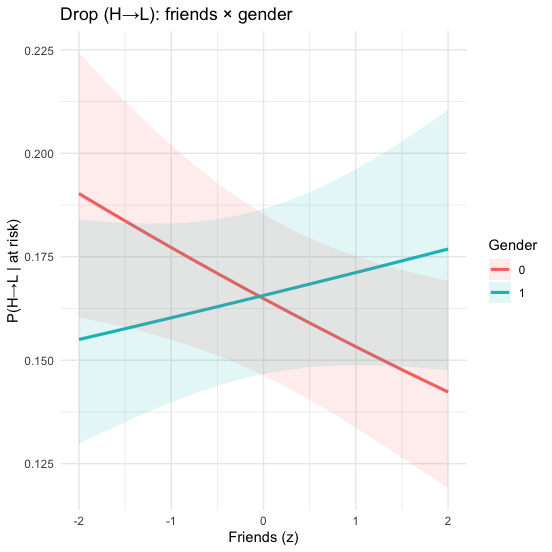}
    \subcaption{Drop hazard vs.\ friends ($z$) by gender.}
    \label{fig:marginal+duration:friends_gender}
  \end{subfigure}\hfill
  \begin{subfigure}{0.49\textwidth}
    \centering\includegraphics[width=\linewidth]{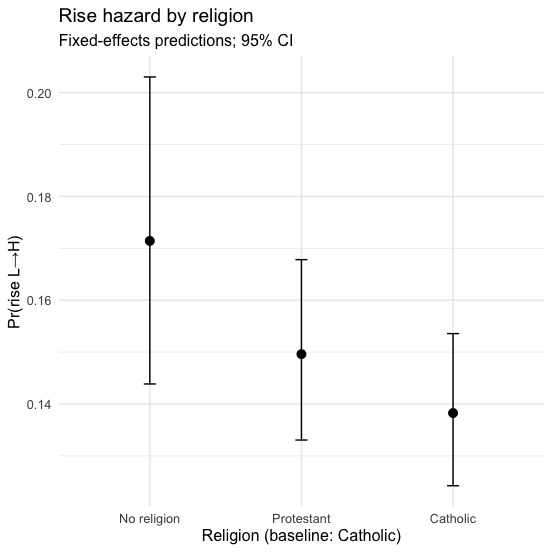}
    \subcaption{Rise hazard by religion (baseline: Catholic).}
    \label{fig:marginal+duration:religion_rise}
  \end{subfigure}

  \medskip

  % --- bottom row ---
  \begin{subfigure}{0.49\textwidth}
    \centering\includegraphics[width=\linewidth]{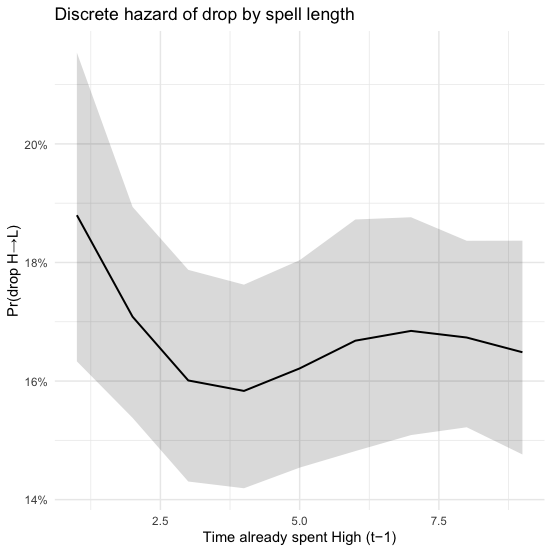}
    \subcaption{Drop: discrete hazard by duration.}
    \label{fig:marginal+duration:drop_duration}
  \end{subfigure}\hfill
  \begin{subfigure}{0.49\textwidth}
    \centering\includegraphics[width=\linewidth]{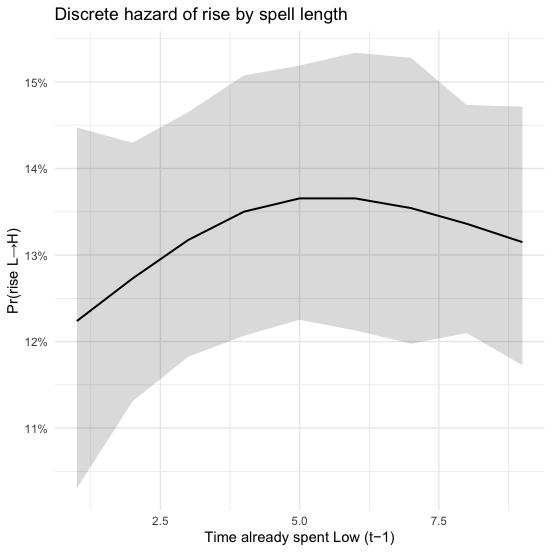}
    \subcaption{Rise: discrete hazard by duration.}
    \label{fig:marginal+duration:rise_duration}
  \end{subfigure}

  \caption{Marginal hazards from the discrete–time GLMMs. 
(a) Predicted hazard of leaving the High state (H$\to$L) as a function of the standardized number of friends ($z$), by gender. 
(b) Predicted hazard of entering the High state (L$\to$H) by religious affiliation (reference: Catholic). 
(c–d) Duration dependence of the H$\to$L and L$\to$H hazards, respectively. 
Curves display fixed–effects predictions with 95\% confidence intervals, holding other covariates at typical values; duration is modeled with a natural cubic spline.}

  \label{fig:marginal+duration}
\end{figure}

% Preamble needs:
% \usepackage{graphicx}
% \usepackage{subcaption}
% (optional) \usepackage{caption}  % for \caption* if you prefer an unnumbered note

\begin{figure}[htbp]
  \centering

  % ---- top row ----
  \begin{subfigure}{0.49\textwidth}
    \centering\includegraphics[width=\linewidth]{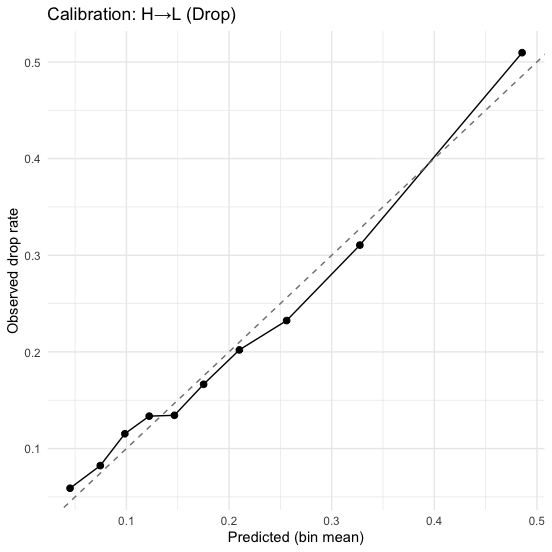}
    \subcaption{Drop (H$\to$L).}
    \label{fig:calibration_markov_hmm:drop}
  \end{subfigure}\hfill
  \begin{subfigure}{0.49\textwidth}
    \centering\includegraphics[width=\linewidth]{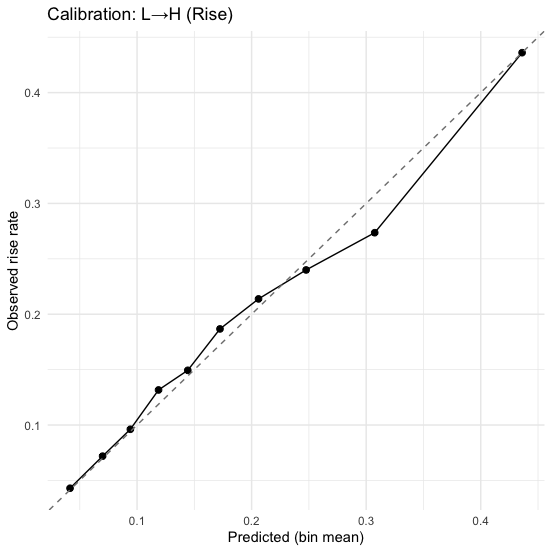}
    \subcaption{Rise (L$\to$H).}
    \label{fig:calibration_markov_hmm:rise}
  \end{subfigure}

  \medskip

  % ---- bottom row ----

  \begin{subfigure}{0.49\textwidth}
    \centering\includegraphics[width=\linewidth]{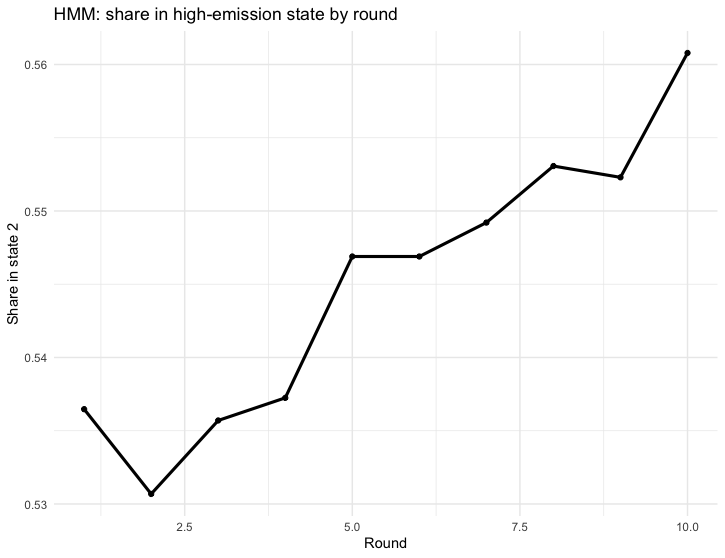}
    \subcaption{HMM: share in the high–emission state by round.}
    \label{fig:calibration_markov_hmm:hmm_share}
  \end{subfigure}

  \caption{Calibration and transition structure. 
(a–b) Decile calibration plots for the discrete–time hazards of leaving and entering the High state (H$\to$L, L$\to$H); points and connecting lines are compared to the $45^\circ$ benchmark for perfect calibration. 
(c) Share of players in the high–emission HMM state by round.}
  \label{fig:calibration_markov_hmm}
\end{figure}

To be clear, HMM transition probabilities are based on a latent state that 
filters measurement noise and misclassification, so diagonal (“stay”) probabilities can exceed the 
raw stay rates computed from dichotomized observed contributions (Table~\ref{tab:counts}). This table reports \emph{observed} transition hazards computed from dichotomized contributions (raw frequencies). By contrast, the integrated Markov–logit matrix in Figure~\ref{fig:calibration_markov_hmm} shows \emph{row–means of predicted transition probabilities} conditional on covariates and duration. Because one object is an empirical frequency and the other is an average model prediction, the diagonal (“stay”) rates need not coincide; the model‐based row means are typically lower when the specification allows switching given covariate variation.

% ---------------------- TABLES ---------------------------------------
\begin{table}[htbp]\centering\small
\caption{Observed transitions across all rounds.\label{tab:counts}}
\begin{tabular}{lcc}
\toprule
 & \multicolumn{2}{c}{$s_t$}\\
\cmidrule(lr){2-3}
$s_{t-1}$ & L & H\\
\midrule
L & 9{,}143 & 2{,}082 \\
H & 2{,}359 & 9{,}735 \\
\midrule
Hazard & $P(H\!\to\!L)=0.195$ & $P(L\!\to\!H)=0.185$ \\
\bottomrule
\end{tabular}
\end{table}

% --------------------------------------------------------------
\section{Switching \& Learning}
\label{sec:switch_learn}

% Switchers

\paragraph{Set–up and objects.}
Let $s_{it}\in\{0,1\}$ denote player $i$’s Low/High state in round $t$ (as in Eq.~\eqref{eq:state});
let $m_{i,t-1}$ be the leave–one–out peer mean and $d_{i,t-1}$ the completed duration in the
current regime (``spell'') at $t\!-\!1$ (Eq.~\eqref{eq:spline}). Throughout the \emph{hazards}
and \emph{Markov–logit} models, we \emph{scale} the leave–one–out peer mean by one endowment
unit ($E{=}12$~L) so a one–unit change in the peer regressor corresponds to a 12~L shift in peers’
average giving; reported odds ratios (ORs) are therefore per-endowment change (for a per–1~L
interpretation, raise the OR to the power $1/12$). By contrast, in the \emph{levels} (updating)
models below, we use peer means and own lags in Lempiras (0–12) to preserve unit
interpretability (L per 1~L). Potential transitions are observed for $t\in\{2,\dots,10\}$.
Define the \emph{switch indicator}
\[
y^{\text{sw}}_{it}=\mathbb{I}\{s_{it}\ne s_{i,t-1}\},
\quad
y^{\text{dir}}_{it}=\mathbb{I}\{s_{i,t-1}=1,\,s_{it}=0\}\in\{0(\mathrm{LH}),1(\mathrm{HL})\}
\]
among true switchers. Empirically, $\approx$19\% of opportunities are switches ($n=4{,}216$ out of $22{,}239$),
with unconditional hazards $P(L\!\to\!H)=0.185$ and $P(H\!\to\!L)=0.195$.

\paragraph{Volatility: probability of any switch.}
We estimate a mixed–effects logit with village, group, and player random intercepts:
\begin{align}
\mathrm{logit}\,\Pr\!\big(y^{\text{sw}}_{it}=1\mid X_{it}\big)
&= \alpha_{v(i)} + \delta_{g(i)} + a_i
 + \beta_{\text{peer}}\,m^{\text{(scaled)}}_{i,t-1}
 + \beta_r\,t
 + f(d_{i,t-1})
 + \Gamma^\top Z_i, \label{eq:anyswitch}
\end{align}
where $m^{\text{(scaled)}}_{i,t-1}$ is the peer mean scaled by the 12L endowment (ORs per endowment change),
$f(\cdot)$ is a natural cubic spline (3~df), and $Z_i$ collects baseline covariates used throughout.
Outcomes: AUC $=0.564$, Brier $=0.152$, calibration intercept/slope $=0.007/0.992$ (fixed–effects predictions).
Selected ORs from \eqref{eq:anyswitch} (per endowment change unless noted): peer mean $0.75$ (lower volatility when
peers are generous), male $1.20$, age $0.92$, education $1.16$, adversaries $1.05$; other covariates
are small or imprecise.

\paragraph{Direction conditional on switching.}
Among switchers we fit
\begin{align}
\mathrm{logit}\,\Pr\!\big(y^{\text{dir}}_{it}=1\mid X_{it},\,y^{\text{sw}}_{it}=1\big)
&= \alpha_{v(i)} + \delta_{g(i)} + a_i
 + \tilde\beta_{\text{peer}}\, m^{\text{(scaled)}}_{i,t-1}
 + \tilde\beta_r\,t
 + \tilde f(d_{i,t-1})
 + \tilde\Gamma^\top Z_i, \label{eq:dir}
\end{align}
with $y^{\text{dir}}_{it}{=}1$ meaning a drop (HL) and $0$ a rise (LH). Results: AUC $=0.624$, Brier $=0.237$.
The peer signal sharply tilts direction toward rises: $\tilde\beta_{\text{peer}}{<}0$ (e.g., $\approx{-}3.16$;
OR $\approx 0.043$ for HL vs.\ LH), so higher peer giving \emph{reduces} the odds of dropping rather than rising.
Duration shows a modest negative middle–spline component; other baselines are not reliably different from zero.
A heatmap over rounds shows that drops occur earlier than rises (Figure~\ref{fig:calibration_markov_hmm}).

\paragraph{Learning in levels and heterogeneity.}
We model round–to–round updating in contributions in \emph{Lempiras} (0–12):
\begin{align}
c_{it}
&= \alpha_i + \gamma_{v(i),t}
 + \beta_{\text{group}}\, m_{i,t-1}
 + \beta_{\text{own}}\, c_{i,t-1}
 + \varepsilon_{it}, \label{eq:levels-fe}
\end{align}
with player fixed effects $\alpha_i$ and village$\times$round fixed effects $\gamma_{v,t}$ (s.e.\ clustered at group).
Specification \eqref{eq:levels-fe} yields $\hat\beta_{\text{group}}=0.752\ (\!0.019)$ and
$\hat\beta_{\text{own}}=-0.074\ (\!0.020)$; within-$R^2=0.134$.
Allowing random slopes by player and adding village, group, and village$\times$round random intercepts,
\begin{align}
c_{it}
&= \alpha + (\beta_{\text{group}}+u_{i,\text{g}})\, m_{i,t-1}
          + (\beta_{\text{own}}  +u_{i,\text{o}})\, c_{i,t-1}
          + b_{v,t} + \delta_{g(i)} + \alpha_{v(i)} + a_i + e_{it}, \label{eq:levels-lmm}
\end{align}
gives $\hat\beta_{\text{group}}=0.696$ and $\hat\beta_{\text{own}}=0.221$ with substantial dispersion
(random‑slope SDs $\approx0.78$ and $0.35$). Cross–level interactions indicate: (i) peer responsiveness rises for
men ($\widehat{\beta}_{\text{peer}\times\text{male}}=0.084$), (ii) own–stickiness declines with schooling
($\widehat{\beta}_{\text{own}\times\text{educ}}=-0.063$), and (iii) peer responsiveness increases slightly with schooling
($\widehat{\beta}_{\text{peer}\times\text{educ}}=0.045$). Peer$\times$friends is slightly negative in a parsimonious
specification ($-0.044$), indicating weaker marginal imitation for more embedded players; peer$\times$adversaries is near zero.
State–dependence of updating is negligible: interactions of $m_{i,t-1}$ and $c_{i,t-1}$ with the last state $s_{i,t-1}$ are small
(Figure~\ref{fig:peer_2x2}\,\subref{fig:peer_2x2:state_dep}).\footnote{Group fixed effects are not separately identified in \eqref{eq:levels-fe}
because each player belongs to a single group and player fixed effects absorb time‑invariant group factors.}

\paragraph{Temporal profile, nonlinearity, and falsification.}
Interacting $m_{i,t-1}$ with round yields a flat trajectory (peer$\times$round $=-0.0034$, n.s.). Spline fits show smooth, monotone peer responses without strong curvature. A lead–lag falsification in levels,
\[
c_{it}=\dots+\beta_{\text{peer}}\,m_{i,t-1}+\gamma_{\text{lead}}\,m_{it}+ \dots,
\]
finds $\widehat\gamma_{\text{lead}}=0.176$ (s.e.\ $0.013$), reflecting serial correlation in group contributions and reinforcing a predictive/associational reading of peer effects.

\paragraph{Counterfactual propagation.}
Using the integrated Markov logit (Eq.~\eqref{eq:markov-logit}), we form row–mean transition matrices (Eq.~\eqref{eq:rowmeans}) and iterate from the observed round–1 share to round~10. Baseline predicted High share is $0.513$. Three counterfactuals yield:
\begin{align*}
\text{Friends}+1\text{SD (H at risk)}&:\;0.522;\\
\text{No–religion (L at risk)}&:\;0.533;\\
\text{Peer mean}+1\text{SD in rounds }1\text{--}5&:\;0.646. 
\end{align*}

Subgroup simulations show round–10 High shares of $0.573$ (males) versus $0.474$ (females), and $0.500/0.518/0.522$ for low/medium/high friend terciles.

All the empirical patterns align with prior evidence that network embeddedness fosters cooperation \citep{haan2006friendship,shirado2013quality}; that early informational environments generate path dependence \citep{fellner2021information}; that gender-linked heterogeneity is common in social dilemmas \citep{furtner2021gender,seguino1996gender}; that religious identity can shape prosocial choices \citep{benjamin2016religious}; and that socioeconomic status alone is not a reliable negative predictor of prosociality in the field \citep{andreoni2021higher,cardenas2000real}.

\begin{figure}[H]
  \centering
  \includegraphics[width=.55\textwidth]{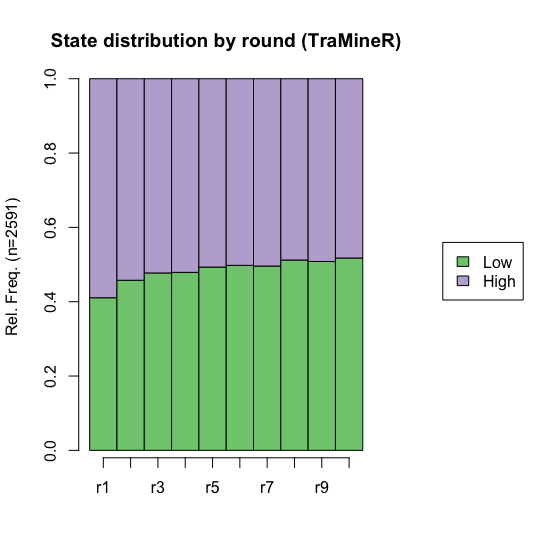}
  \caption{State distribution by round (\texttt{TraMineR} in R) \cite{gabadinho2011analyzing} The share in Low increases gradually; hazards are $P(L\!\to\!H)=0.185$ and $P(H\!\to\!L)=0.195$.}
  \label{fig:seqdist}
\end{figure}
% Preamble:
% \usepackage{graphicx}
% \usepackage{subcaption}
% \usepackage{float} % if you want [H]

\begin{figure}[H]
  \centering

  % ---- top row ----
  \begin{subfigure}{0.49\textwidth}
    \centering\includegraphics[width=\linewidth]{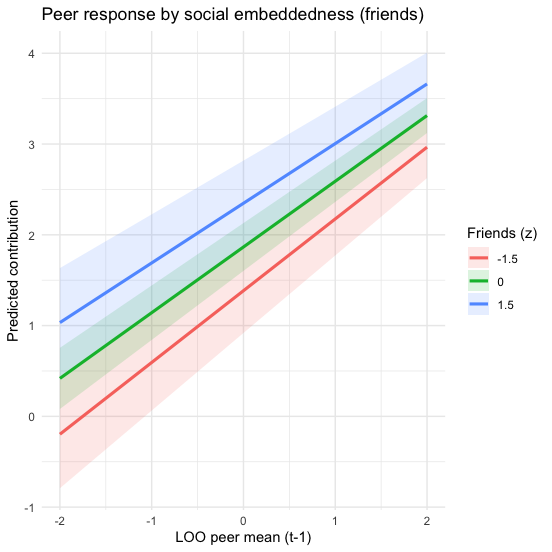}
    \subcaption{Peer response by friends ($z$).}
    \label{fig:peer_2x2:friends}
  \end{subfigure}\hfill
  \begin{subfigure}{0.49\textwidth}
    \centering\includegraphics[width=\linewidth]{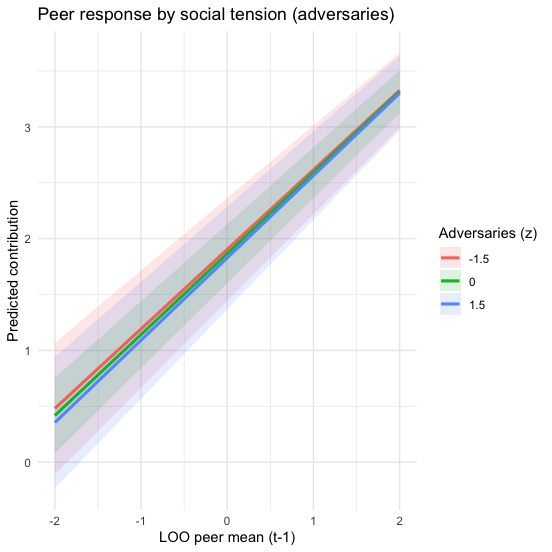}
    \subcaption{Peer response by adversaries ($z$).}
    \label{fig:peer_2x2:adversaries}
  \end{subfigure}

  \medskip

  % ---- bottom row ----
  \begin{subfigure}{0.49\textwidth}
    \centering\includegraphics[width=\linewidth]{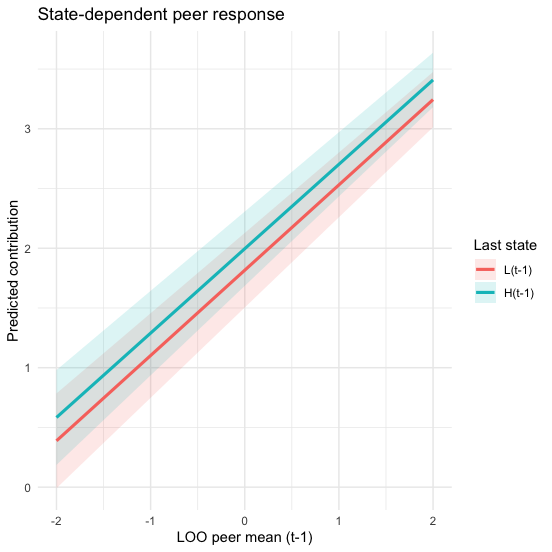}
    \subcaption{State–dependent peer response (last–round Low vs High nearly parallel).}
    \label{fig:peer_2x2:state_dep}
  \end{subfigure}\hfill
  \begin{subfigure}{0.49\textwidth}
    \centering\includegraphics[width=\linewidth]{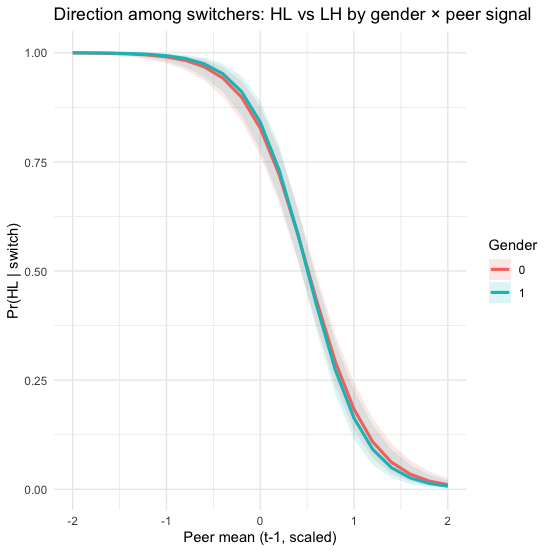}
    \subcaption{Direction among switchers: $\Pr(\text{HL}\mid\text{switch})$ over peer signal, by gender.}
    \label{fig:peer_2x2:dir_by_gender}
  \end{subfigure}

  \caption{Peer–response margins from the levels specification with interactions. 
Panels (a)–(b) plot predicted own–contribution responses to the peer signal (leave–one–out group mean) across social context: (a) by standardized number of friends ($z$); (b) by standardized number of adversaries ($z$). 
Panel (c) shows state dependence by estimating separate peer–response curves for players who were Low vs.\ High in the previous round; the nearly parallel lines indicate limited slope heterogeneity by prior state. 
Panel (d) conditions on switchers between rounds and reports $\Pr(H\!\to\!L \mid \text{switch})$ as a function of the peer signal, stratified by gender. 
Curves are model‐based (levels) predictions with 95\% confidence intervals, holding other covariates at typical values.}
  \label{fig:peer_2x2}
\end{figure}

\section{IV theory \& diagnostics}\label{app:iv_diagnostics}

Each round $t$, player $i$ in a group of size $N=5$ chooses $c_{it}\in[0,12]$.
Instantaneous utility combines material payoffs, altruism, and aversion to being at the norm:
\[
\pi_{it} = \frac{b}{N}\!\left(c_{it} + (N-1)\bar{c}_{-i,t}\right) - \kappa c_{it}
+ d_i c_{it}^{\alpha} - h_i \exp\!\left\{\knorm \big(c_{it}-\bar{c}_{-i,t-1}\big)^2\right\},
\]
with $\alpha\in(0,1)$, $d_i>0$, $h_i\in[0,1]$, $\knorm >0$, and $\bar{c}_{-i,t}$ the leave–one–out
peer mean. Under $b/N<\kappa<b$, best–reply dynamics admit two locally attracting fixed points,
generating the two–basin structure we document.

Empirically, we estimate
\[
\tilde{c}_{it} \;=\; \beta\,\tilde{\bar{c}}_{-i,t} + \tilde{X}'_{it}\gamma + \tilde{\varepsilon}_{it},
\]
after removing player and village$\times$round fixed effects by two–way demeaning.
To address simultaneity of $\tilde{\bar{c}}_{-i,t}$ and $\tilde{\varepsilon}_{it}$, we use:
(i) \emph{leave–one–out composition instruments} (demeaned group means of predetermined traits of
the other players) and (ii) a \emph{time–varying shift–share} IV for \emph{lagged} peers,
\[
Z^{SS}_{i,t-1} \;=\; \sum_l s_{g,-i}^{(l)}\,\mu_{t-1}^{(l)}, \qquad
\text{and its leave–one–village variant}\quad
Z^{LOV}_{i,t-1} \;=\; \sum_l s_{g,-i}^{(l)}\,\mu_{t-1}^{(-v),(l)},
\]
where $s_{g,-i}^{(l)}$ is the LOO group share for trait $l$ and $\mu_{t-1}^{(-v),(l)}$ is the
cross–village mean at $t-1$ excluding village $v$. Random assignment within villages and the
predetermined nature of traits support relevance and exclusion. Because avatars are anonymous,
composition can affect $c_{it}$ only via peers’ contributions.

We report instrument balance on pre–game traits within village$\times$round, cluster–robust KP first–stage
$F$, reduced–form $Z\!\to\! c_{it}$, a placebo at the earliest round with instrument variation
($t=2$ in our data), and Hansen/Sargan $J$ tests when over–identified. For the lagged
shift–share (LOV) IV, the cluster–robust KP $F$ is $\approx6.1$, so we complement the 2SLS point estimate
with weak–ID–robust Anderson–Rubin CIs (95\% [$-0.14, 0.90$]; 90\% [$0.30, 0.88$]).

% ----- Figure (schematic) -----
\begin{figure}[H]
\centering
\begin{tikzpicture}[
  >=Latex, scale=0.75, every node/.style={transform shape}, node distance=18mm,
  var/.style={draw, rounded corners, inner sep=6pt, align=center},
  latent/.style={draw, dashed, rounded corners, inner sep=6pt, align=center}, note/.style={inner sep=2pt, font=\small}
]
\node[var] (shares) {Groupmate composition \\ $\{s_{igl}\}$ (random within village)};
\node[var, below=of shares] (shifts) {Outside-village shifters by round \\ $\{\mu^{(-v)}_{t-1,l}\}$};
\node[var, right=25mm of shares] (Z) {Instrument \\ $Z^{LOV}_{i,t-1}$};
\node[var, right=25mm of Z] (peer) {Lagged peer mean \\ $\bar c_{-i,t-1}$};
\node[var, right=25mm of peer] (own) {Own contribution \\ $c_{it}$};
\node[latent, below=of Z] (fes) {Controls \& FE: $\alpha_i$, $\gamma_{v,t}$};
\node[latent, above=of peer] (anon) {Anonymity: peers unobserved};
\draw[->] (shares) -- (Z); \draw[->] (shifts) -- (Z); \draw[->] (Z) -- node[above, note]{First stage} (peer);
\draw[->] (peer) -- node[above, note]{$\beta$} (own); \draw[->] (fes) -- (peer); \draw[->] (fes) -- (own);
\draw[->, dashed] (Z) to[bend left=20] node[above, note]{\emph{excluded by anonymity}} (own);
\draw[->, dashed] (anon) to[bend right=20] (Z);
\node[note, below=30mm of peer, align=left] (checks) {%
\textbf{Checks:}\\
-- Balance of $Z$ on $X_i$ within $\gamma_{v,t}$\\
-- KP first-stage $F$ (cluster)\\
-- Reduced form $Z\!\to\! c_{it}$\\
-- Placebo at earliest round with var in $Z$ ($t=2$)\\
-- Hansen/Sargan $J$ (if over-ID)
};
\end{tikzpicture}
\caption{IV/2SLS design: random group composition $\times$ outside–village round shifters (LOV).}
\end{figure}
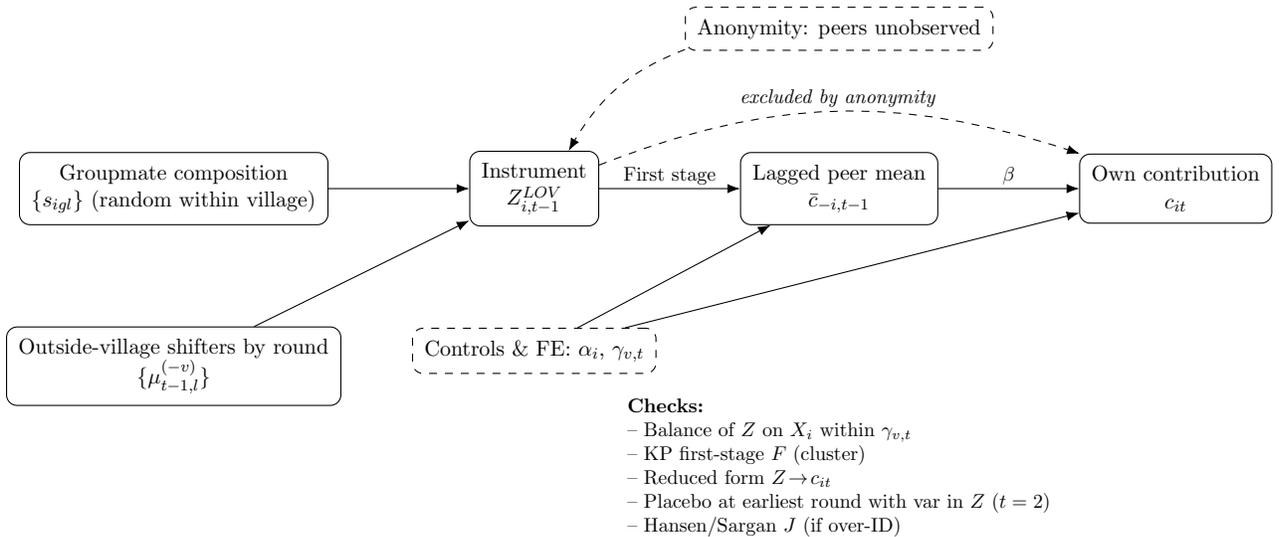

\subsection*{Drop–one IV sweep (LOO composition instruments)}
\begin{table}[H]
\centering
\caption{Three–IV subsets from the LOO composition set.}
\label{tab:iv_sweep}
\begin{tabular}{lccc}
\toprule
IV subset & Weak–IV $F$ & Wu–Hausman $p$ & Sargan $p$ \\
\midrule
Gender, Protestant, Non-Religious & $17.18$ & $0.987$ & $0.015$ \\
Gender, Protestant, Indigenous & $35.06$ & $0.360$ & $0.020$ \\
\textbf{Gender, Non-Religious, Indigenous} & \textbf{34.99} & \textbf{0.458} & \textbf{0.360} \\
Protestant, Non-Religious, Indigenous & $17.68$ & $0.082$ & $0.045$ \\
\bottomrule
\end{tabular}
\end{table}

% --------------------------------------------------------------
\section{Structural back‑out}
\label{sec:struct}

%code NEW_WORK.R

This subsection summarizes a structural back–out that maps observed updates into two interpretable
primitives, an altruism weight and an effective norm–pull, using only within–session data. Full
derivations, estimation details, and sensitivity checks are provided in Appendix~D; here we report the
identification logic, units, and main patterns.

We summarize behavior using two primitives. First, an \emph{altruism weight} $d_i$ governing the
intrinsic marginal benefit from giving (warm–glow). Second, an effective norm--aversion $\phi_i$ that penalizes being at the group's lagged mean. In the first–order condition, only the
product of a curvature parameter and a conformity scale is identified, so we collect these in
\[
\phi_i \equiv 2\knorm \,h_i,
\]
and estimate $(d_i,\phi_i)$ by minimum distance on each player’s within–session choices using the
leave–one–out lagged peer mean (units: Lempiras).\footnote{All peer means exclude the focal player.
Variables are in \emph{lempiras} (0–12). Identification is \emph{partial}: only the product
$\phi_i=2\knorm h_i$ is point–identified from the quadratic deviation cost; the mapping to $(k,h_i)$ is
conventional and used for scaling only. We fix $\alpha=\tfrac12$ in the warm–glow term and show
sensitivity to $\alpha\in\{0.3,0.7\}$ in SI. Data handling follows the main text: rounds
$2$–$10$ for lags; a small ridge ($10^{-3}$) when $c_{it}=0$; and 99th–percentile winsorization for
the back–out residual.}

Three facts emerge. (a) The distribution of $d_i$ is broad (median $\approx2.8$; IQR
$1.9$–$3.4$); plugging these into the individual best–reply delivers an implied high fixed point at
the 12~L cap for essentially all players, i.e., altruism is \emph{in principle} sufficient to
support high contributions. (b) The inferred norm–pull $\phi_i$ is near zero for the vast
majority (about 96\% below $0.1$), indicating weak aversion to matching the most recent group norm. (c) The \emph{combination}, high $d_i$ but weak $\phi_i$, implies a
fragile high state: without an early cooperative majority to anchor behavior, contributions tend to
relax downward. This micro pattern mirrors the reduced–form markers: an early tipping threshold
near 5.75~L, a village critical mass around 60\%, and asymmetric persistence with rarer $L\!\to\!H$
escapes than $H\!\to\!L$ drops.\\

Interpretation is descriptive, not causal. The back–out treats lagged peer means as a shifting norm
parameter and maps choices into primitives given that norm; it does not separately identify reputation,
sanctions, or higher–order beliefs. Sensitivity checks (alternative $\alpha = \{0.3, 0.7\}$, excluding boundary
rounds, no winsorization) leave the qualitative ranking intact and is presented in \ref{sec:robustness_SI}.

\begin{figure}[H]
  \centering
  \includegraphics[width=\linewidth]{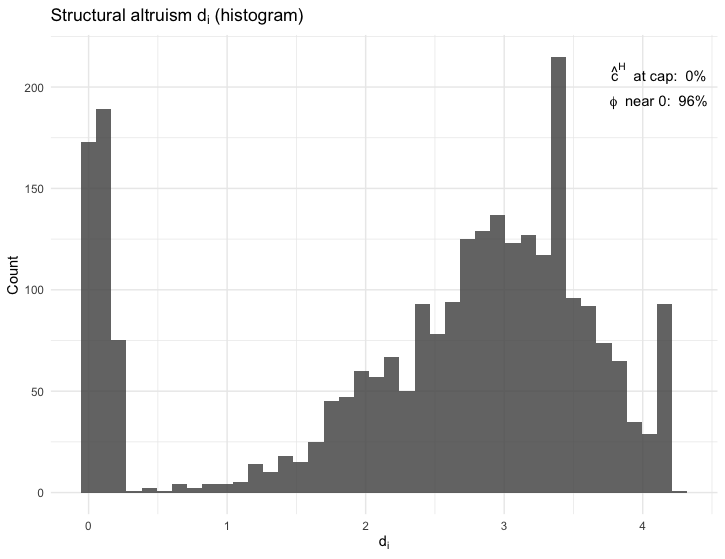}
  \caption{Histogram of recovered altruism $d_i$ (Lempiras units). Bars show the distribution across
  players (rounds 2–10 used; ridge added at $c=0$; 99th–percentile winsorization of the fit residual).
  For nearly all individuals the implied high fixed point $\hat c^{H}$ sits at the 12‑L cap, while the
  effective norm‑pull $\phi_i$ (not shown) is near zero for most players, consistent with weak conformity
  forces.}
  \label{fig:di}
\end{figure}

\noindent\emph{Link to reduced form.} The weak inferred norm‑pull aligns with the IV evidence: little
same‑round peer effect after fixed effects, but a positive \emph{lagged} response toward the recent group
norm (Section~\ref{sec:iv_results}).

\subsection*{Placebo at baseline}
Using the earliest round with non–zero variation in $Z$ ($t=2$), regressing own contribution
on $Z$ within village FE yields a near–zero coefficient (clustered SE), consistent with exclusion.

\subsection*{First–stage visualization}
Figure~\ref{fig:first_stage_plot} plots the residualized endogenous regressor against the
residualized instrument for the strongest LOO IV, with fitted line and 95\% band.
\begin{figure}[H]
 \centering
 \includegraphics[width=0.7\textwidth]{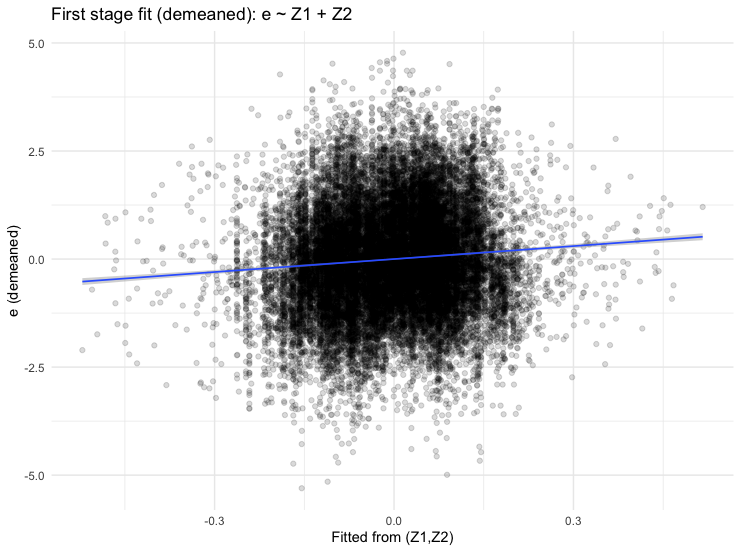}
 \caption{Partialled first stage: residualized peer mean vs. residualized instrument.}
 \label{fig:first_stage_plot}
 \end{figure}

\section{Robustness \& sensitivity}
\label{sec:robustness_SI}

\subsection{Multiple flips within the ten‑round session}

\begin{table}[H]
\centering
\caption{How often players cross the High/Low threshold (6 lempiras).}
\label{tab:flip_stats}
\begin{tabular}{lrr}
\toprule
Statistic & Count & Share \\
\midrule
Total players & 2\,591 & 1.000 \\
Exactly one flip & 266 & 0.103 \\
Multiple flips (\(\geq 2\)) & 1\,110 & 0.428 \\
\bottomrule
\end{tabular}
\end{table}

\subsection{Endpoint based on the last two rounds \((r_9\lor r_{10})\)}

\begin{table}[H]
\centering
\caption{Confusion matrix: original (\(r_{10}\)) vs.\ two‑round endpoint.}
\label{tab:conf_two}
\begin{tabular}{lcccc}
\toprule
 & \multicolumn{4}{c}{Two‑round endpoint (\(r_9\lor r_{10}\))} \\
\cmidrule(l){2-5}
Original (\(r_{10}\)) & HH & HL & LH & LL \\
\midrule
HH & 1\,055 &   0 &   0 &   0 \\
HL &   169 & 303 &   0 &   0 \\
LH &     0 &   0 & 195 &   0 \\
LL &     0 &   0 &  86 & 783 \\
\bottomrule
\end{tabular}
\end{table}

\vspace{0.5em}
\noindent The share of coefficients whose signs coincide with the
baseline specification is \(0.85\), confirming qualitative stability.

\subsection{Multinomial logit with two‑round endpoint}
\begin{table}[H]
\centering
\caption{Multinomial logit (LL baseline); endpoint = High in \(r_9\) \textit{or} \(r_{10}\). Coefficients with robust standard errors in parentheses.}
\label{tab:multi_two}
\scriptsize
\begin{tabular}{lccc}
\toprule
Predictor & \multicolumn{3}{c}{Outcome relative to LL} \\
\cmidrule(l){2-4}
 & \(\;\)LH\(\;\) & \(\;\)HL\(\;\) & \(\;\)HH\(\;\) \\
\midrule
Intercept                 & $-1.178^{***}$ \;(0.153) & $-1.077^{***}$ \;(0.151) &  $0.186^{*}$   \;(0.102) \\
Age                       &  $0.084$ \;(0.085)       & $-0.069$ \;(0.086)       & $-0.155^{***}$ \;(0.058) \\
Friends                   &  $0.027$ \;(0.080)       & $-0.036$ \;(0.080)       &  $0.109^{**}$  \;(0.053) \\
Adversaries               & $-0.028$ \;(0.079)       &  $0.078$ \;(0.071)       &  $0.050$       \;(0.050) \\
Friend density            &  $0.106$ \;(0.142)       & $-0.245$ \;(0.153)       &  $0.011$       \;(0.098) \\
Adversary density         &  $0.111$ \;(0.116)       &  $0.259^{**}$ \;(0.118)  &  $0.014$       \;(0.082) \\
Network size              &  $0.122$ \;(0.107)       &  $0.030$ \;(0.106)       &  $0.048$       \;(0.074) \\
Education                 &  $0.075$ \;(0.084)       &  $0.156^{**}$ \;(0.079)  &  $0.025$       \;(0.056) \\
Male                      &  $0.284^{*}$ \;(0.150)   &  $0.155$ \;(0.148)       &  $0.413^{***}$ \;(0.100) \\
Married                   & $-0.150$ \;(0.161)       &  $0.009$ \;(0.158)       &  $0.003$       \;(0.108) \\
Religion: none            &  $0.676^{***}$ \;(0.248) &  $0.206$ \;(0.264)       &  $0.368^{**}$  \;(0.180) \\
Religion: Protestant      &  $0.195$ \;(0.159)       &  $0.097$ \;(0.154)       &  $0.160$       \;(0.104) \\
\midrule
Observations              & \multicolumn{3}{c}{2\,591} \\
Residual Dev. / AIC       & \multicolumn{3}{c}{5\,968/6\,040} \\
\bottomrule
\end{tabular}

\medskip
\raggedright
\footnotesize
\({}^{***}p<0.01\), \({}^{**}p<0.05\), \({}^{*}p<0.10\).
Only predictors featuring in the main text are displayed; full output
available on request.
\end{table}

Across all checks, key qualitative findings, friendship stabilises
cooperation, adversarial pressure erodes it, men exhibit wider
variance, and the religious minority enjoys higher upward mobility, are
unchanged.  The behavioural split into persistent High and persistent
Low contributors is therefore robust to reasonable changes in threshold
choice, endpoint definition, and model specification.

\subsection{Robustness of Alpha}
\label{sec:alpha}

In the baseline structural back–out we set $\alpha=\tfrac{1}{2}$ for the warm‑glow term $d_i c_i^{\alpha}$. To assess robustness to curvature, we re‑estimate the primitives at $\alpha\in\{0.3,0.7\}$. As expected, changing $\alpha$ mostly rescales the recovered $d_i$ levels while leaving qualitative features intact. In particular, the effective norm‑pull $\phi_i$ concentrates at zero for virtually all players and the share whose implied high fixed point reaches the endowment cap remains about 5\% under both curvatures. The sample size is unchanged ($N=2{,}591$). These facts are summarized below and in the machine‑readable tables included in the replication package.

\paragraph{Key invariants across $\alpha\in\{0.3,0.7\}$.}
(i) \textit{Norm‑pull}: $\phi_i$ is estimated at $0$ at all reported quantiles (Min, Q1, Median, Q3, Max; Mean $=0$) for both $\alpha=0.3$ and $\alpha=0.7$. (ii) \textit{Share at cap}: the share with $c_i^{H}$ at the contribution cap is $0.05$ for both curvatures. (iii) \textit{Near‑zero norm‑pull mass}: the share with $\phi_i\le 0.1$ is $1.00$ in both cases.\footnote{See Tables~\ref{tab:sa1_alpha30} and \ref{tab:sa1_alpha70} for the distributional summaries.}

Because a lower curvature ($\alpha=0.3$) implies faster diminishing marginal returns in $c^\alpha$, the same observed behavior is rationalized by larger $d_i$ values; the converse holds for $\alpha=0.7$. Median (and other quantiles) shift accordingly, but this is a level effect, not a reversal of qualitative conclusions.

\begin{table}[H]
\centering
\caption{Structural primitives at $\boldsymbol{\alpha=0.3}$ (player level, $N{=}2{,}591$).}
\label{tab:sa1_alpha30}
\begin{tabular}{lrrrrrrrr}
\toprule
 & Min & P10 & Q1 & Median & Q3 & P90 & Max & Mean \\
\midrule
$d_i$   & 0.000 & 2.860 & 5.278 & 7.281 & 8.823 & 10.256 & 11.388 & 6.866 \\
\midrule
Share at cap $c^H$ & \multicolumn{8}{c}{0.05} \\
Share with $\phi_i\le 0.1$ & \multicolumn{8}{c}{1.00} \\
\bottomrule
\end{tabular}

\end{table}

\begin{table}[H]
\centering
\caption{Structural primitives at $\boldsymbol{\alpha=0.7}$ (player level, $N{=}2{,}591$).}
\label{tab:sa1_alpha70}
\begin{tabular}{lrrrrrrrr}
\toprule
 & Min & P10 & Q1 & Median & Q3 & P90 & Max & Mean \\
\midrule
$d_i$   & 0.014 & 0.999 & 1.299 & 1.491 & 1.619 & 1.727 & 1.806 & 1.386 \\
\midrule
Share at cap $c^H$ & \multicolumn{8}{c}{0.05} \\
Share with $\phi_i\le 0.1$ & \multicolumn{8}{c}{1.00} \\
\bottomrule
\end{tabular}

\end{table}

Across $\alpha\in\{0.3,0.7\}$, the data continue to imply (i) broad heterogeneity in altruism intensities $d_i$ with the expected level rescaling as $\alpha$ changes; (ii) essentially no measurable norm‑pull in the one‑period deviation penalty ($\phi_i\simeq 0$ for almost everyone); and (iii) a small minority of players (about 5\%) whose implied high steady state lies at the contribution cap. These features are precisely those that underlie the reduced‑form markers in the main text, an empirical tipping threshold, asymmetric state persistence, and strong sensitivity to the recent group norm, so all qualitative conclusions are invariant to $\alpha$ within this range.

\section{Illustrative village examples}

\subsection*{Example of Games with Gender, Friendship Status, Religion and Financial Autonomy differences}

\begin{figure}[H]
  \centering
  \includegraphics[width=0.9\textwidth]{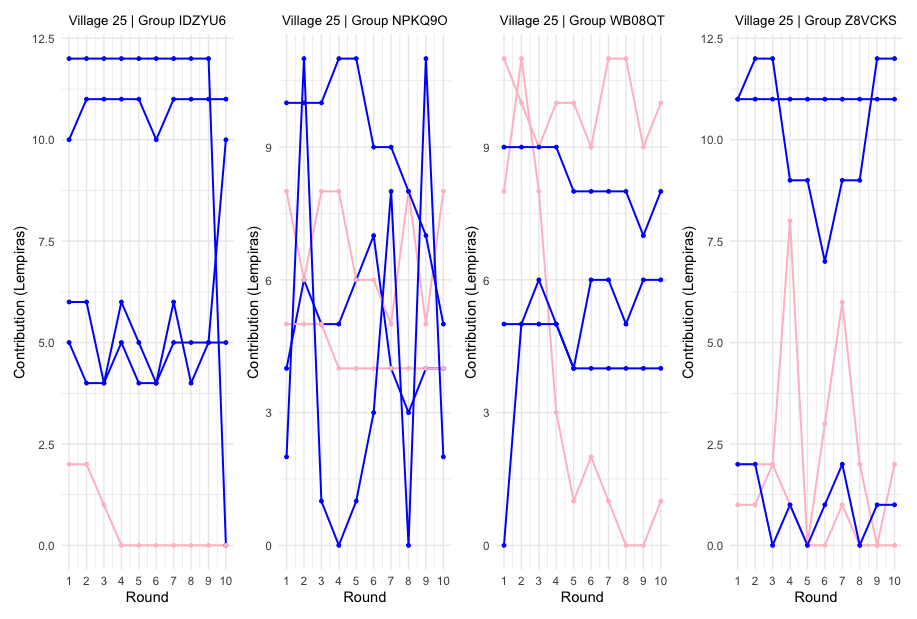}
  \caption{Village 25 were most men (blue) are selected from women (women)}
  \label{fig:VillageMen}
\end{figure}

\begin{table}[H]
\centering
\begin{tabular}{ccccc}
\hline
Group & Village & ID     & Men Avg & Women Avg \\
\hline
5     & 25      & IDZYU6 & 0.50     & 7.95       \\
6     & 25      & NPKQ9O & 5.55     & 5.90       \\
7     & 25      & WB08QT & 6.75     & 5.87       \\
8     & 25      & Z8VCKS & 1.65     & 7.40       \\
\hline
\end{tabular}
\caption{Average contribution by gender for each group in Village~25}
\label{tab:men}
\end{table}

In Figure~\ref{fig:VillageMen}, three out of the four groups, men contribute more on average than women. The only exception is Group 5, where women contributions are significantly higher, shown in table~\ref{tab:men}.

\begin{table}[H]
\centering
\begin{tabular}{ccrr}
\hline
Village & Group   & Avg.\ friends (Above 6) & Avg.\ friends ($\leq$ 6) \\
\hline
25      & IDZYU6  & 16.00                  & 7.67                \\
25      & NPKQ9O  &  8.50                  & 6.33                \\
25      & WB08QT  &  7.00                  & 5.00                \\
25      & Z8VCKS  &  7.00                  & 7.33                \\
\hline
\end{tabular}
\caption{Average friendship degree by contribution category in Village 25}
\label{tab:village25_friends}
\end{table}

In Figure~\ref{fig:VillageMen} in three of the four groups (IDZYU6, NPKQ9O and WB08QT), players whose mean contribution exceeds 6 Lempiras enjoy substantially larger friendship networks than those contributing at most 6. Only in group Z8VCKS do lower contributors have a marginally higher average number of friends (7.33 vs.\ 7.00), reversing the overall pattern, shown in table \ref{tab:village25_friends}.

\begin{figure}[H]
  \centering
  \includegraphics[width=0.9\textwidth]{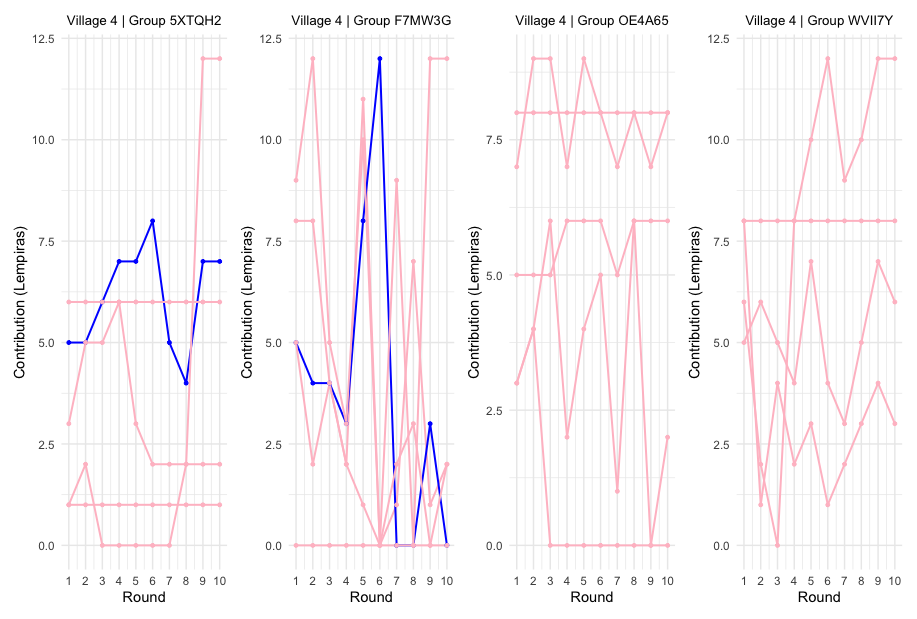}
  \caption{Village 4, where most women (pink) are selected from men (blue)}
  \label{fig:VillageWomen}
\end{figure}

\begin{table}[H]
\centering
\begin{tabular}{ccrr}
\hline
Village & Group   & Women Avg & Men Avg \\
\hline
4       & 5XTQH2  & 3.28      & 6.10    \\
4       & F7MW3G  & 3.42      & 3.90    \\
4       & OE4A65  & 5.10      & ---     \\
4       & WVII7Y  & 6.10      & ---     \\
\hline
\end{tabular}
\caption{Average contribution by gender for each group in Village~4}
\label{tab:women}
\end{table}

In Figure~\ref{fig:VillageWomen}, two out of the four groups (the other two groups do not contain men), men contribute more on average than women, shown in table~\ref{tab:women}.

\begin{figure}[H]
  \centering
  \includegraphics[width=0.9\textwidth]{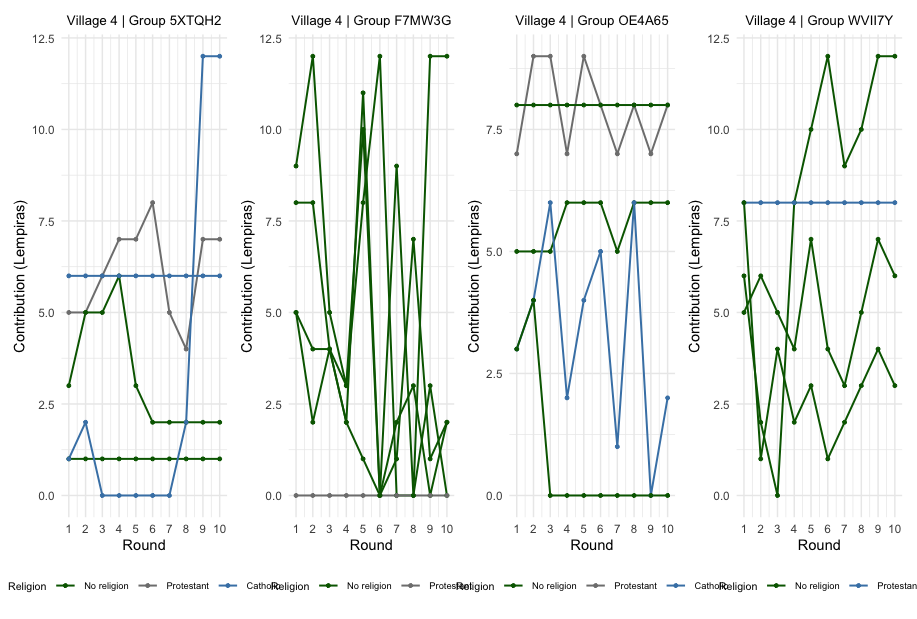}
  \caption{Village 25 were most men (blue) are selected from women (women)}
  \label{fig:Religion}
\end{figure}

In Figure~\ref{fig:Religion} in all four groups, Catholics contributed less than Protestants and individuals with no religion \ref{tab:v4_religion}.

\begin{table}[H]
\centering
\begin{tabular}{lrr}
\hline
Group   & Avg.\ Catholic & Avg.\ Non‑Catholic \\
\hline
5XTQH2  & 2.10           & 5.00               \\
F7MW3G  & 4.40           & 0.00               \\
OE4A65  & 4.77           & 5.60               \\
WVII7Y  & 5.47           & 8.00               \\
\hline
\end{tabular}
\caption{Average contribution by religious affiliation for each group in Village~4}
\label{tab:v4_religion}
\end{table}

In three of the four groups (5XTQH2, OE4A65, and WVII7Y), Catholic players contributed less on average than non‑Catholics.

\begin{figure}[H]
  \centering
  \includegraphics[width=0.9\textwidth]{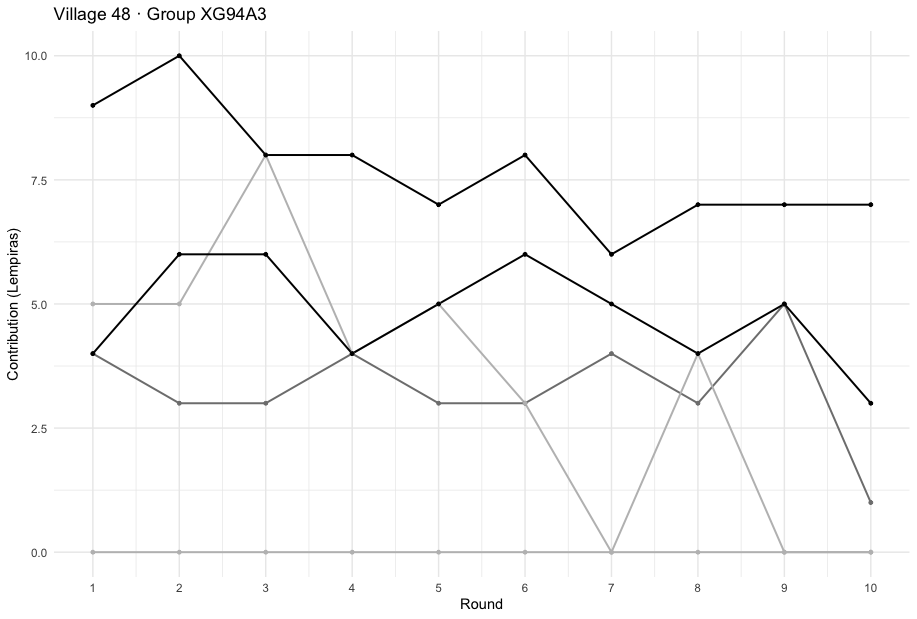}
  \caption{Village 4 were most women (pink) are selected from men (blue)}
  \label{fig:FinAutonomy}
\end{figure}

\begin{table}[H]
\centering
\begin{tabular}{lr}
\hline
Status              & Avg.\ Contribution \\
\hline
Not Autonomous      & 1.70               \\
Autonomous          & 6.25               \\
\hline
\end{tabular}
\caption{Average contribution by financial autonomy status}
\label{tab:autonomy_contrib}
\end{table}

Financially autonomous individuals in group XG94A3 in village 48 contributed more than non financially autonomous.

\section{Type classification: Free Riders, Conditional cooperators, Triangles and Other}

As a micro–foundation, Table~\ref{tab:type_share} shows that nearly half of players are
\emph{conditional cooperators} (48.4\%), about one in twelve are \emph{free riders} (8.5\%), and a small share
are \emph{triangle} types (3.1\%), with the remainder classified as \emph{other} (40.0\%).
Free riders are concentrated in the Low-path cluster (Table~\ref{tab:type_by_cluster}), and
conditional/triangle types display strong positive updating to peers’ lagged contributions (Table~\ref{tab:lmm_fixed}),
consistent with the two-path dynamics we document.\\

\begin{table}[H]
\centering
\begin{threeparttable}
\caption{Revealed–slope types (shares, player level)}
\label{tab:type_share}
\begin{tabular}{lcc}
\toprule
Type & $n$ & Share (\%) \\
\midrule
Conditional cooperator & 1{,}254 & 48.4 \\
Triangle & 81 & 3.1 \\
Free rider & 220 & 8.5 \\
Other & 1{,}036 & 40.0 \\
\bottomrule
\end{tabular}
\end{threeparttable}
\end{table}

\begin{table}[H]
\centering
\begin{threeparttable}
\caption{Type composition within Ward--D2 clusters}
\label{tab:type_by_cluster}
\begin{tabular}{lcccc}
\toprule
 & \multicolumn{2}{c}{C1} & \multicolumn{2}{c}{C2} \\
\cmidrule(lr){2-3}\cmidrule(lr){4-5}
Type & $n$ & \% of C1 & $n$ & \% of C2 \\
\midrule
Conditional cooperator & 878 & 50.4 & 376 & 44.3 \\
Triangle               &  62 &  3.6 &  19 &  2.2 \\
Free rider             &   0 &  0.0 & 220 & 25.9 \\
Other                  & 803 & 46.1 & 233 & 27.5 \\
\bottomrule
\end{tabular}
\end{threeparttable}
\end{table}

Using player–level empirical-Bayes slopes from a random-slope learning model and an evidence-based curvature
screen, we classify updating rules into \emph{Conditional cooperator}, \emph{Free rider}, \emph{Triangle}, and
\emph{Other}. The aggregate shares are: 48.4\% conditional cooperators, 8.5\% free riders, 3.1\% triangle, and
40.0\% other (Table~\ref{tab:type_share}). Unsupervised Ward–D2 clustering on the full ten-round trajectories
confirms a two-regime structure (average silhouette $=0.360$ at two clusters ($K$=2); Table~\ref{tab:hc_sil_by_k}). The confusion
matrix shows that cluster C2 is predominantly Low (LL $=641/848$), whereas cluster C1 contains most High finishers
and movers (HH $=1{,}038$; HL $=306$; LH $=171$) (Table~\ref{tab:hc_confusion}). Free riders concentrate in the
Low-path cluster (25.9\% of C2; essentially absent in C1), producing quasi-separation in Low-path logits
(Table~\ref{tab:type_by_cluster}, Table~\ref{tab:lowpath_or}). In a mixed-effects specification with type
interactions, conditional and triangle types have significantly positive responses to peers’ lag
($\beta_{\bar c_{-i}^{t-1}\times\text{Cond}}=0.347,\ p<10^{-12}$;
$\beta_{\bar c_{-i}^{t-1}\times\text{Tri}}=0.355,\ p<10^{-4}$), while the \emph{Other} group responds negatively
($\beta=-0.359,\ p<10^{-13}$); the free-rider baseline shows no meaningful peer response (Table~\ref{tab:lmm_fixed}).
Types are fairly stable over time (early vs.\ late agreement $=75.9\%$; Table~\ref{tab:stability}), and triangle players’
response peaks at about $0.38$ within-player SD above their typical peer environment (IQR $\approx 0.18$–$0.55$;
Table~\ref{tab:triangle_vertex}). At the village level, a higher \emph{free-rider share} predicts a larger Low-path
fraction (OLS coefficient $0.643\pm0.066$, $p<0.001$; Table~\ref{tab:group_ols}), aligning the micro composition
with macro persistence.

\begin{table}[H]
\centering
\begin{threeparttable}
\caption{Low-path (C2) membership. Logistic regression (odds ratios with 95\% CI)\tnote{a}}
\label{tab:lowpath_or}
\begin{tabular}{lccc}
\toprule
Predictor & OR & 95\% CI (lower) & 95\% CI (upper) \\
\midrule
Intercept                              & $4.88\times 10^{7}$ & 0.000 & $1.25\times 10^{268}$ \\
Conditional cooperator (vs Free rider) & 0.000               & 0.000 & $2.61\times 10^{252}$ \\
Triangle (vs Free rider)               & 0.000               & 0.000 & $2.39\times 10^{252}$ \\
Other (vs Free rider)                  & 0.000               & 0.000 & $1.97\times 10^{252}$ \\
Friends (z)                            & 0.916               & 0.822 & 1.022 \\
IH (z)                                 & 0.986               & 0.887 & 1.097 \\
Male (vs Female)                       & 0.657               & 0.526 & 0.821 \\
Financial autonomy $=1$ (vs $0$)       & 0.297               & 0.104 & 0.844 \\
\bottomrule
\end{tabular}
\begin{tablenotes}[flushleft]
\footnotesize
\item[a] Reference type is \emph{Free rider}. Near-zero/huge ORs reflect quasi/complete separation (free riders concentrate in C2).
\end{tablenotes}
\end{threeparttable}
\end{table}

\begin{table}[H]
\centering
\begin{threeparttable}
\caption{Learning model with type interactions (LMM fixed effects)\tnote{a}}
\label{tab:lmm_fixed}
\begin{tabular}{lrrrr}
\toprule
Term & Estimate & SE & $t$ & $p$ \\
\midrule
Intercept                                   & 0.738 & 0.295 &  2.50 & 0.012 \\
Own lag                                     & 0.0296 & 0.0068 &  4.36 & $1.28\times 10^{-5}$ \\
Peer mean (lag)                              & -0.011 & 0.0432 & -0.26 & 0.797 \\
Type: Conditional cooperator                 & 3.346 & 0.323 & 10.36 & $<2\times 10^{-16}$ \\
Type: Triangle                               & 3.596 & 0.586 &  6.13 & $9.55\times 10^{-10}$ \\
Type: Other                                  & 8.657 & 0.329 & 26.34 & $<2\times 10^{-16}$ \\
Peer lag $\times$ Conditional cooperator     & 0.347 & 0.0470 &  7.38 & $1.60\times 10^{-13}$ \\
Peer lag $\times$ Triangle                   & 0.355 & 0.0836 &  4.24 & $2.21\times 10^{-5}$ \\
Peer lag $\times$ Other                      & -0.359 & 0.0477 & -7.53 & $5.28\times 10^{-14}$ \\
\bottomrule
\end{tabular}
\begin{tablenotes}[flushleft]
\footnotesize
\item[a] Model: $c_{i,t} \sim c_{i}^{t-1} + \bar c_{-i,t-1} \times \text{type} + (1+\bar c_{-i,t-1}\,|\,\text{player}) + (1\,|\,\text{group})$; reference type = \emph{Free rider}. “Boundary (singular) fit” indicates a near-perfect correlation between random intercept and slope; fixed-effects inferences are unaffected.
\end{tablenotes}
\end{threeparttable}
\end{table}

\begin{table}[H]
\centering
\begin{threeparttable}
\caption{Type stability: rounds 1--5 vs.\ 6--10}
\label{tab:stability}
\begin{tabular}{lcc}
\toprule
Agreement & $n$ & Share (\%) \\
\midrule
TRUE  & 1{,}966 & 75.9 \\
FALSE &   625 & 24.1 \\
\bottomrule
\end{tabular}
\end{threeparttable}
\end{table}

\begin{table}[H]
\centering
\begin{threeparttable}
\caption{Triangle types: vertex of the quadratic response in standardized $x$}
\label{tab:triangle_vertex}
\begin{tabular}{lccc}
\toprule
Mean & Median & 25th pctl & 75th pctl \\
\midrule
0.382 & 0.331 & 0.181 & 0.547 \\
\bottomrule
\end{tabular}
\end{threeparttable}
\end{table}

\begin{table}[H]
\centering
\begin{threeparttable}
\caption{Group share Low‐path vs.\ type composition (OLS)}
\label{tab:group_ols}
\begin{tabular}{lrrr}
\toprule
Predictor & Coef. & SE & $p$ \\
\midrule
Intercept                    & 0.293 & 0.0099 & $<0.001$ \\
Free rider (share)           & 0.643 & 0.0664 & $<0.001$ \\
Conditional cooperator (share) & 0.036 & 0.0230 & 0.114 \\
Triangle (share)             & -0.077 & 0.1059 & 0.467 \\
\midrule
\multicolumn{4}{l}{Residual SE $=0.204$;\quad $R^2=0.0866$ (adj.\ $R^2=0.0839$);\quad $F(3,1015)=32.08$, $p<2.2\times 10^{-16}$.}
\\[-1ex]
\bottomrule
\end{tabular}
\end{threeparttable}
\end{table}

\begin{table}[H]
\centering
\begin{threeparttable}
\caption{Multinomial logit of revealed type (baseline: Free rider)}
\label{tab:multinom_types}
\begin{tabular}{lrrr}
\toprule
\multicolumn{4}{l}{\textbf{Outcome: Conditional cooperator (vs Free rider)}}\\
\midrule
Predictor & Coef. & SE & $z$ \\
\midrule
Intercept                &  1.595 & 0.101 & 15.863 \\
Adversaries          & -0.089 & 0.083 & -1.078 \\
Age                  & -0.270 & 0.091 & -2.971 \\
Financial autonomy     &  1.327 & 1.033 &  1.284 \\
Friends            &  0.233 & 0.098 &  2.394 \\
Gender          &  0.531 & 0.195 &  2.723 \\
Adversarial Network Density &  0.118 & 0.141 &  0.840 \\
Friendship Network Density & -0.166 & 0.136 & -1.214 \\
\midrule
\multicolumn{4}{l}{\textbf{Outcome: Triangle (vs Free rider)}}\\
\midrule
Predictor & Coef. & SE & $z$ \\
\midrule
Intercept                & -1.251 & 0.194 & -6.451 \\
Adversaries          & -0.031 & 0.149 & -0.207 \\
Age                  & -0.240 & 0.162 & -1.476 \\
Financial autonomy      &  1.834 & 1.242 &  1.476 \\
Friends            &  0.149 & 0.169 &  0.883 \\
Gender            &  0.721 & 0.325 &  2.218 \\
Adversarial Network Density &  0.293 & 0.263 &  1.116 \\
Friendship Network Density & -0.496 & 0.272 & -1.822 \\
\midrule
\multicolumn{4}{l}{\textbf{Outcome: Other (vs Free rider)}}\\
\midrule
Predictor & Coef. & SE & $z$ \\
\midrule
Intercept                &  1.314 & 0.103 & 12.737 \\
Adversaries           & -0.034 & 0.083 & -0.411 \\
Age                   & -0.158 & 0.092 & -1.723 \\
Financial autonomy      &  1.442 & 1.035 &  1.393 \\
Friends               &  0.183 & 0.099 &  1.847 \\
Gender           &  0.720 & 0.197 &  3.651 \\
Adversarial Network Density &  0.139 & 0.142 &  0.977 \\
Friendship Network Density  & -0.144 & 0.139 & -1.041 \\
\bottomrule
\end{tabular}
\begin{tablenotes}[flushleft]
\footnotesize
\item \emph{Note:} Coefficients are log-odds relative to \emph{Free rider}. $z$ and SE from \texttt{nnet::multinom}. See text for interpretation.
\end{tablenotes}
\end{threeparttable}
\end{table}

\begin{figure}[H]
  \centering
    \centering
    \includegraphics[width=\textwidth]{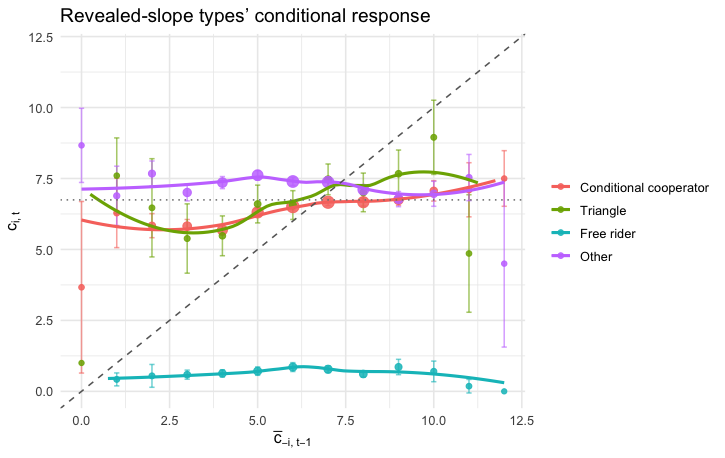}
  \caption{Revealed–slope types’ conditional response (balanced assignment). Own contribution in $t$ ($c_{it}$) vs.\ leave‑one‑out peer mean in $t\!-\!1$ ($\bar c_{-i,t-1}$).
  Colored loess lines show the average conditional‑response by revealed type: \emph{conditional cooperator}
  (upward‑sloping), \emph{free rider} (flat near zero), \emph{triangle} (hump‑shaped), and \emph{other}.
  Points are binned means with 95\% CIs; marker size reflects cell counts. The dashed 45° line indicates
  perfect matching; the dotted horizontal line marks the round‑1 mean. Shares: 48.4\% conditional cooperators,
  8.5\% free riders, 3.1\% triangle, 40.0\% other.}
  \label{fig:condresp_types}
\end{figure}

Figure~\ref{fig:condresp_types} plots $c_{it}$ against $\bar c_{-i,t-1}$ with binned means (95\%~CIs) and
type-specific loess lines. The distribution of types is: Conditional cooperators $=48.4\%$, Free riders $=8.5\%$,
Triangle $=3.1\%$, Other $=40.0\%$ (Table~\ref{tab:type_share}). Conditional and triangle types exhibit robust
positive conditional responses, whereas free riders’ slope is near zero and the residual “other’’ group responds
weakly or negatively on average (Table~\ref{tab:lmm_fixed}). Using full-trajectory Ward–D2 clustering, the two
regimes are well separated (average silhouette $0.360$ at two clusters ($K$=2); Table~\ref{tab:hc_sil_by_k}): C1 concentrates High
finishers and movers (HH $=1{,}038$, HL $=306$, LH $=171$) while C2 is predominantly Low (LL $=641$)
(Table~\ref{tab:hc_confusion}). Free riders are concentrated in C2 (25.9\% of C2; essentially none in C1)
(Table~\ref{tab:type_by_cluster}). Types are reasonably stable over time (agreement between rounds~1–5 and 6–10
$=75.9\%$; Table~\ref{tab:stability}). For triangle players the estimated peak occurs at $\hat x^\star \approx 0.38$
within-player SD above their typical peer environment (IQR $\approx 0.18$–$0.55$; Table~\ref{tab:triangle_vertex}).
At the group level, a larger free-rider share predicts more Low-path membership (OLS coefficient $0.643\pm0.066$,
$p<0.001$; Table~\ref{tab:group_ols}).\\

\paragraph{Estimation note on the Low-path logit.}
Because free riders are concentrated in the Low-path cluster (and absent in the High-path cluster), the Low-path
logit exhibits quasi-complete separation when using \emph{Free rider} as the reference type, inflating the intercept
and shrinking the other type coefficients toward $-\infty$ (Table~\ref{tab:lowpath_or}). We therefore interpret the
cross-tab and, as a robustness check, report penalized logit estimates (Firth) which yield finite odds ratios but the
same qualitative conclusion: relative to free riders, conditional and triangle types have much lower odds of Low-path
membership, conditional on the covariates.\footnote{Standard options include Firth (e.g., \texttt{logistf}) or
bias-reduced GLMs (e.g., \texttt{brglm2}); Bayesian logit with weakly informative priors produces similar finite
estimates.}

% ===== Results text (defense) =====
\section{Calibration and model fit}
Figure~\ref{fig:calibration} visualizes the calibration of the Fermi--Moran imitation model to the field data. 
Panel~(a) shows the loss surface \(L(d,k)=\sum_{s,s'\in\{H,L\}}\!\big(T_{ss'}(d,k)-\widehat{T}_{ss'}\big)^{2}\) over the fitness tilt \(d\) (favoring High contributors when \(d>0\)) and the imitation intensity \(k\). The basin around the optimum indicates identifiability rather than a flat ridge, and the optimizer settles near \((\hat d,\hat k)\approx(-0.51,\,0.50)\).
Panel~(b) compares the empirical Round~1\(\to\)10 transition matrix \(\widehat{T}\) with the simulated \(T(\hat d,\hat k)\); the match captures the moderate persistence of High types while revealing a systematic under‐prediction of Low\(\to\)Low stickiness. 
Panel~(c) plots the fraction of High contributors by round in the data against the calibrated simulator, with a Monte-Carlo band; trajectories align on the decline in cooperation and the conditional stability of the High state. 
Together, these diagnostics support the evolutionary-branching interpretation: social imitation with a slight tilt against High contributions reproduces the observed bifurcation, while the residual gap for Low persistence motivates heterogeneity (e.g., agent thresholds or network-weighted copying) developed in the discussion.

% ===== Methods text (concise) =====

We binarize contributions into High/Low using the Round~1 mean as the threshold and compute the empirical transition matrix \(\widehat{T}\) from Round~1 to Round~10. 
For each \((d,k)\), we simulate groups of size \(N=5\) for ten rounds under a pairwise Fermi update with adoption probability
\[
p_{i\to j}=\frac{1}{1+\exp\!\big(-k\,[w_i-w_j]\big)},\quad 
w_\ell=1+d\,\mathbb{1}\{s_\ell=\text{High}\},
\]
to obtain \(T(d,k)\). 
Parameters that minimize \(L(d,k)\) are obtained via a coarse grid search to initialize the limited-memory quasi-Newton optimization algorithm (L-BFGS-B), subject to \(k\ge 0\).
We also report round-wise High shares from the calibrated simulator with 10–90\% envelopes. 

% ===== Figure =====
\begin{figure}[H]
  \centering
  \includegraphics[width=\linewidth]{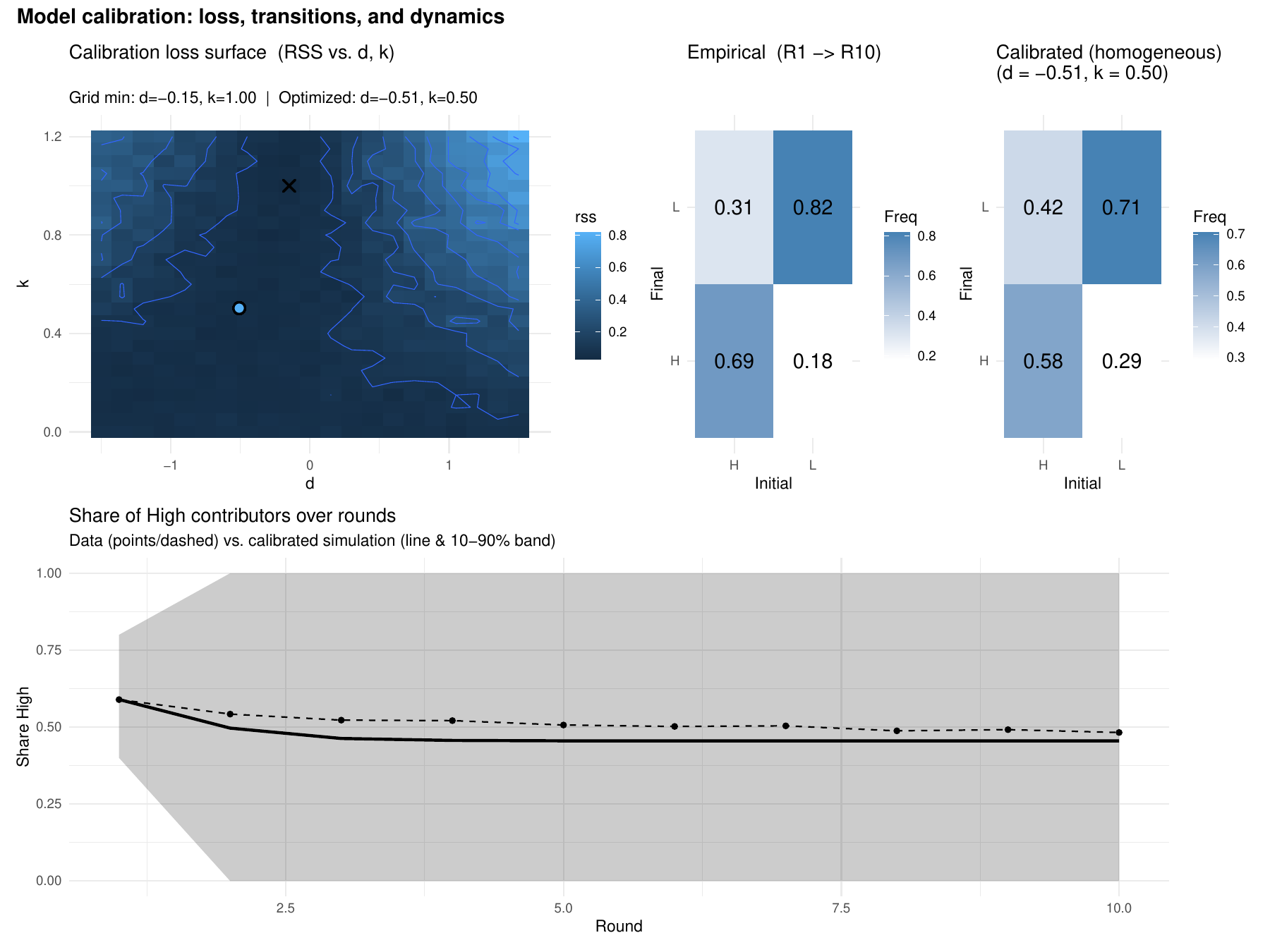}
  \caption{\textbf{Calibration of the imitation model.}
  \textbf{(a)} Loss surface \(L(d,k)\) with grid minimum and optimizer marked. 
  \textbf{(b)} Empirical Round~1\(\to\)10 transition matrix \(\widehat{T}\) vs.\ calibrated \(T(\hat d,\hat k)\) (here \(\hat d\approx-0.51,\ \hat k\approx0.50\)). 
  \textbf{(c)} Fraction High by round: data (points/dashed) vs.\ calibrated simulation (solid) with 10–90\% band.
  The model reproduces the conditional stability of High contributors and overall decline, while under-predicting Low–state persistence, motivating heterogeneity/extensions discussed in the text.}
  \label{fig:calibration}
\end{figure}

\subsection{Link to the conditional‑cooperator slope}

The reduced‑form “conditional‑cooperation” slope from \(c_{igt}\) on \(\peer_{-i,g,t-1}\) (per‑player OLS)
is essentially uncorrelated with \(\hat\phi_i\) from the structural model (Pearson \(r\approx 0.014\); OLS \(p=0.52\)).
This cautions against treating the copying slope as a direct measure of the structural norm‑pull parameter.

\section{Large Tables}

\newpage

% Requires: \usepackage{booktabs,threeparttable}

%-----------------------------%
% M1. Model fit and random effects
%-----------------------------%
% Requires in preamble:
% \usepackage{booktabs,threeparttable}

% =========================
% Table A: Random-slope learning (full covariates set #1)
% =========================

\begin{table}[H]
\centering
\begin{threeparttable}
\caption{Random-slope learning model (full covariates, without Financial Autonomy)}
\label{tab:rs_learn_full_A}
\footnotesize

\textbf{Linear mixed model fit by REML.} \\
REML criterion at convergence: \textbf{98278.4} \\
Number of obs: \textbf{22{,}239}; groups: player = \textbf{2{,}471}, village\_code.round\_n = \textbf{1{,}206}, group = \textbf{522}, village\_code = \textbf{134}.

\vspace{0.75em}
\textbf{Random effects}
\begin{tabular}{llrr}
\toprule
Groups                & Name          & Variance & Std.~Dev. \\
\midrule
player                & c\_group\_lag & 0.009588 & 0.09792 \\
player.1              & c\_own\_lag   & 0.064528 & 0.25402 \\
player.2              & round\_n      & 0.017530 & 0.13240 \\
player.3              & (Intercept)   & 4.457348 & 2.11124 \\
village\_code.round\_n& (Intercept)   & 0.025182 & 0.15869 \\
group                 & (Intercept)   & 0.000000 & 0.00000 \\
village\_code         & (Intercept)   & 0.333841 & 0.57779 \\
Residual              &               & 3.150782 & 1.77504 \\
\bottomrule
\end{tabular}

\vspace{0.75em}
\textbf{Fixed effects}
\begin{tabular}{lrrrrl}
\toprule
Term                  & Estimate     & Std.~Error   & df           & $t$ value & $p$ \\
\midrule
(Intercept)           & 5.025e{+}00  & 4.799e{-}01  & 1.995e{+}02  & 10.471    & $<2\mathrm{e}{-}16^{***}$ \\
Own Contribution Lag           & 1.255e{-}01  & 9.143e{-}03  & 4.313e{+}03  & 13.729    & $<2\mathrm{e}{-}16^{***}$ \\
Peer Contribution Lag          & 4.652e{-}03  & 1.285e{-}02  & 5.354e{+}03  &  0.362    & 0.7174 \\
Exp Round             &-4.065e{-}02  & 5.787e{-}03  & 1.215e{+}03  & -7.024    & 3.58 $\mathrm{e}{-}12^{***}$ \\
Age                   &-4.332e{-}03  & 4.532e{-}03  & 1.944e{+}03  & -0.956    & 0.3392 \\
Gender                & 5.685e{-}01  & 1.198e{-}01  & 1.967e{+}03  &  4.744    & 2.25$\mathrm{e}{-}06^{***}$ \\
Friends               & 2.979e{-}02  & 1.444e{-}02  & 1.976e{+}03  &  2.062    & 0.0393$^{*}$ \\
Adversaries           & 5.676e{-}02  & 4.657e{-}02  & 1.929e{+}03  &  1.219    & 0.2231 \\
Food Insecurities     &-2.432e{-}02  & 1.199e{-}01  & 1.974e{+}03  & -0.203    & 0.8393 \\
Marital Status        & 1.104e{-}01  & 1.299e{-}01  & 1.936e{+}03  &  0.850    & 0.3955 \\
Network Density Friendship  &-8.619e{+}00  & 1.492e{+}01  & 1.206e{+}02  & -0.578    & 0.5645 \\
Network Density Adversarial & 3.771e{+}01  & 6.782e{+}01  & 1.258e{+}02  &  0.556    & 0.5792 \\
Network Size          &-1.826e{-}04  & 9.864e{-}04  & 1.261e{+}02  & -0.185    & 0.8534 \\
Education             & 5.386e{-}02  & 2.601e{-}02  & 1.996e{+}03  &  2.071    & 0.0385$^{*}$ \\
Non Religion          & 3.677e{-}01  & 2.100e{-}01  & 2.000e{+}03  &  1.751    & 0.0801$^{.}$ \\
Protestant            & 1.954e{-}01  & 1.352e{-}01  & 1.312e{+}03  &  1.445    & 0.1487 \\
Indigenous            &-1.032e{-}01  & 1.886e{-}01  & 1.087e{+}03  & -0.547    & 0.5844 \\
Access Routes         &-3.133e{-}02  & 1.092e{-}01  & 1.247e{+}02  & -0.287    & 0.7745 \\
\bottomrule
\end{tabular}

\begin{tablenotes}
\footnotesize
\item Signif.~codes: $^{***}p<0.001$, $^{**}p<0.01$, $^{*}p<0.05$, $^{.}p<0.1$.
\end{tablenotes}
\end{threeparttable}
\end{table}

% =========================
% Table B: Random-slope learning (full covariates set #2, incl. financial autonomy)
% =========================
\begin{table}[H]
\centering
\begin{threeparttable}
\caption{Random-slope learning model (full covariates, with financial autonomy)}
\label{tab:rs_learn_full_B}
\footnotesize

\textbf{Linear mixed model fit by REML.} \\
REML criterion at convergence: \textbf{78472.9} \\
Number of obs: \textbf{17{,}847}; groups: player = \textbf{1{,}983}, village\_code.round\_n = \textbf{1{,}206}, group = \textbf{522}, village\_code = \textbf{134}.

\vspace{0.75em}
\textbf{Random effects}
\begin{tabular}{llrr}
\toprule
Groups                & Name          & Variance & Std.~Dev. \\
\midrule
player                & c\_group\_lag & 0.01210  & 0.1100 \\
player.1              & c\_own\_lag   & 0.06576  & 0.2564 \\
player.2              & round\_n      & 0.01895  & 0.1377 \\
player.3              & (Intercept)   & 4.44420  & 2.1081 \\
village\_code.round\_n& (Intercept)   & 0.02420  & 0.1556 \\
group                 & (Intercept)   & 0.00000  & 0.0000 \\
village\_code         & (Intercept)   & 0.32613  & 0.5711 \\
Residual              &               & 3.05073  & 1.7466 \\
\bottomrule
\end{tabular}

\vspace{0.75em}
\textbf{Fixed effects}
\begin{tabular}{lrrrrl}
\toprule
Term                  & Estimate     & Std.~Error   & df           & $t$ value & $p$ \\
\midrule
(Intercept)           & 4.797e{+}00  & 5.255e{-}01  & 2.205e{+}02  &  9.128    & $<2\mathrm{e}{-}16^{***}$ \\
Own Contribution Lag           & 1.150e{-}01  & 1.029e{-}02  & 3.504e{+}03  & 11.179    & $<2\mathrm{e}{-}16^{***}$ \\
Peer Contribution Lag         & 1.902e{-}02  & 1.422e{-}02  & 4.166e{+}03  &  1.338    & 0.18110 \\
Exp Round             &-4.337e{-}02  & 6.388e{-}03  & 1.160e{+}03  & -6.789    & 1.80$\mathrm{e}{-}11^{***}$ \\
Age                   &-4.765e{-}03  & 5.413e{-}03  & 1.565e{+}03  & -0.880    & 0.37882 \\
Gender                & 6.525e{-}01  & 1.383e{-}01  & 1.591e{+}03  &  4.717    & 2.60$\mathrm{e}{-}06^{***}$ \\
Friends               & 3.845e{-}02  & 1.610e{-}02  & 1.605e{+}03  &  2.388    & 0.01703$^{*}$ \\
Adversaries           & 4.111e{-}02  & 5.101e{-}02  & 1.553e{+}03  &  0.806    & 0.42041 \\
Food Insecurities     & 8.764e{-}03  & 1.348e{-}01  & 1.594e{+}03  &  0.065    & 0.94818 \\
Marital Status        & 1.140e{-}01  & 1.734e{-}01  & 1.550e{+}03  &  0.658    & 0.51094 \\
Network Density Friendship  &-5.730e{+}00  & 1.608e{+}01  & 1.273e{+}02  & -0.356    & 0.72220 \\
Network Density Adversarial & 3.325e{+}01  & 7.207e{+}01  & 1.243e{+}02  &  0.461    & 0.64533 \\
Network Size          &-3.588e{-}05  & 1.062e{-}03  & 1.317e{+}02  & -0.034    & 0.97310 \\
Education             & 6.577e{-}02  & 2.998e{-}02  & 1.608e{+}03  &  2.194    & 0.02838$^{*}$ \\
Non Religion          & 5.465e{-}01  & 2.403e{-}01  & 1.621e{+}03  &  2.274    & 0.02309$^{*}$ \\
Protestant            & 2.427e{-}01  & 1.508e{-}01  & 1.081e{+}03  &  1.609    & 0.10786 \\
Indigenous            &-4.570e{-}02  & 2.093e{-}01  & 8.919e{+}02  & -0.218    & 0.82724 \\
Access Routes         &-5.528e{-}02  & 1.181e{-}01  & 1.315e{+}02  & -0.468    & 0.64049 \\
Financial Autonomy    & 1.342e{+}00  & 4.833e{-}01  & 1.754e{+}03  &  2.776    & 0.00556$^{**}$ \\
\bottomrule
\end{tabular}

\begin{tablenotes}
\footnotesize
\item Signif.~codes: $^{***}p<0.001$, $^{**}p<0.01$, $^{*}p<0.05$, $^{.}p<0.1$.
\end{tablenotes}
\end{threeparttable}
\end{table}

% Requires: \usepackage{booktabs,threeparttable,multirow}

%-----------------------------%
% A1. End-state counts
%-----------------------------%
\begin{table}[H]
\centering
\caption{Counts of end-state transitions}
\label{tab:app-trans-counts}
\begin{threeparttable}
\footnotesize
\begin{tabular}{lrrrr}
\toprule
Transition & HH & HL & LH & LL \\
\midrule
Count      & 1055 & 472 & 195 & 869 \\
\bottomrule
\end{tabular}
\end{threeparttable}
\end{table}
\begin{table}[H]
\centering
\caption{Dynamic state logit: odds ratios for High state ($s_{igt}{=}1$)}
\label{tab:m1_state_or}
\begin{tabular}{lccc}
\toprule
Term & OR & 95\% CI & $p$ \\
\midrule
$s_{ig,t-1}$ (High)        & 1.22 & [1.10, 1.36] & 0.0002 \\
Peer mean$_{t-1}$/12       & 5.97 & [5.44, 6.56] & $<0.0001$ \\
Round $t$                  & 0.98 & [0.96, 0.99] & 0.0116 \\
Initial High ($s_{ig1}$)   & 87.66 & [67.53, 113.78] & $<0.0001$ \\
Avg.\ peer mean$_{t-1}$    & 0.82 & [0.74, 0.90] & $<0.0001$ \\
Male                       & 1.33 & [1.03, 1.71] & 0.0294 \\
Age                        & 1.06 & [0.91, 1.23] & 0.4711 \\
Friends                    & 1.12 & [0.98, 1.28] & 0.0887 \\
Adversaries                & 1.01 & [0.89, 1.15] & 0.9037 \\
Friendship network density & 1.02 & [0.80, 1.32] & 0.8557 \\
Adversaries network density& 1.00 & [0.81, 1.22] & 0.9702 \\
Network size               & 0.98 & [0.81, 1.18] & 0.8176 \\
Education                  & 0.92 & [0.80, 1.06] & 0.2336 \\
Food insecurity            & 0.88 & [0.69, 1.14] & 0.3420 \\
Married                    & 1.02 & [0.77, 1.34] & 0.9036 \\
Ethnicity (Indigenous)     & 0.75 & [0.52, 1.08] & 0.1276 \\
No religion                & 1.42 & [0.92, 2.19] & 0.1184 \\
Protestant                 & 1.18 & [0.90, 1.54] & 0.2299 \\
Access routes              & 1.02 & [0.86, 1.22] & 0.8006 \\
\midrule
Player RE variance         & \multicolumn{3}{c}{6.82 (SD 2.61)} \\
Group RE variance          & \multicolumn{3}{c}{$\approx 0.00$ (boundary)} \\
Village RE variance        & \multicolumn{3}{c}{0.00 (boundary)} \\
$N_{\text{obs}}$ / players / villages & \multicolumn{3}{c}{22{,}239 / 2{,}471 / 134} \\
\bottomrule
\end{tabular}
\begin{flushleft}\footnotesize
Notes: Binomial GLMM with logit link; odds ratios (ORs) are shown for fixed effects.
Continuous covariates are standardized (1~SD). Peer mean terms use the leave-one-out
group mean; the contemporaneous peer mean is lagged by one round and scaled so that
1 equals a full endowment. Village and group random-intercept variances are at (or
very near) the boundary (singular), indicating negligible between-village and 
between-group heterogeneity.
\end{flushleft}
\end{table}

\begin{table}[H]
\centering
\caption{LCMM High–class membership (logit): odds ratios}
\label{tab:mem_lcmm_or}
\begin{tabular}{lccc}
\toprule
Term & OR & 95\% CI & $p$ \\
\midrule
Male & 1.43 & [1.20, 1.69] & $<0.001$ \\
Friends & 1.10 & [1.01, 1.21] & 0.037 \\
Age  & 1.00 & [0.90, 1.11] & 0.958 \\
Adversaries & 1.01 & [0.92, 1.10] & 0.908 \\
Friendship Network density  & 0.97 & [0.79, 1.19] & 0.749 \\
Adversarial Network density & 1.01 & [0.86, 1.20] & 0.861 \\
Network size & 0.96 & [0.83, 1.12] & 0.642 \\
Education  & 0.96 & [0.87, 1.06] & 0.423 \\
Food insecurity & 0.97 & [0.89, 1.06] & 0.529 \\
Married & 1.01 & [0.84, 1.22] & 0.905 \\
No religion & 1.36 & [1.01, 1.84] & 0.043 \\
Protestant & 1.18 & [0.97, 1.42] & 0.095 \\
Ethnicity  & 0.70 & [0.54, 0.92] & 0.010 \\
Access routes & 1.00 & [0.87, 1.15] & 0.984 \\
\midrule
Village RE variance & \multicolumn{3}{c}{0.104 (SD 0.322)} \\
$N$ (players / villages) & \multicolumn{3}{c}{2{,}471 / 134} \\
\bottomrule
\end{tabular}
\begin{flushleft}\footnotesize
Notes: GLMM with logit link and village random intercept; response is LCMM class membership (High vs.\ Low). Continuous covariates are standardized (1~SD). 95\% CIs are computed as $\exp(\hat\beta \pm 1.96\,\widehat{\mathrm{SE}})$ from the reported estimates and standard errors.
\end{flushleft}
\end{table}

\begin{table}[htbp]\centering\small
\caption{Subgroup counterfactuals (round–10 High share).}
\label{tab:subgroup_cf}
\begin{tabular}{lcc}
\toprule
Subgroup & Baseline & Notes\\
\midrule
Male   & 0.573 & row–mean transition from male subsample\\
Female & 0.474 & row–mean transition from female subsample\\
Friends terciles (low/med/high) & 0.500 / 0.518 / 0.522 & baselined at round 1\\
\bottomrule
\end{tabular}
\end{table}

\begin{table}[htbp]\centering\small
\caption{Discrete‐time GLMMs: odds ratios (95\% CI) for key covariates.\label{tab:glmm_or}}
\begin{tabular}{lcc}
\toprule
 & Drop $H\!\to\!L$ ($N\!=\!11{,}529$) & Rise $L\!\to\!H$ ($N\!=\!10{,}710$)\\
\midrule
Peer mean$_{t-1}$ (scaled) & $0.00273^{***}$ [0.00183, 0.00408] & $226^{***}$ [150, 342]\\
Age ($z$)                   & $0.917^{**}$ [0.857, 0.981]        & $0.890^{***}$ [0.833, 0.952]\\
Friends ($z$)               & 0.973 [0.919, 1.030]               & 1.06 [1.00, 1.13]\\
Adversaries ($z$)           & 1.01  [0.961, 1.07]                & $1.06^{*}$ [1.01, 1.12]\\
Education $b0100$ ($z$)     & $1.15^{***}$ [1.09, 1.22]          & $1.13^{***}$ [1.06, 1.20]\\
Male (=1)                   & 1.01  [0.905, 1.12]                & $1.44^{***}$ [1.29, 1.61]\\
No religion (vs.\ Cath.)    & 0.912 [0.761, 1.09]                & $1.29^{**}$ [1.06, 1.56]\\
Protestant (vs.\ Cath.)     & 0.896 [0.792, 1.01]                & 1.10 [0.966, 1.24]\\
Indigenous                  & $1.35^{***}$ [1.14, 1.59]          & 1.14 [0.958, 1.35]\\
Village RE SD               & 0.316                               & 0.268\\
AIC/BIC                     & 10212.8 / 10359.8                   & 9338.1 / 9483.7\\
\bottomrule
\end{tabular}

\medskip
\footnotesize Notes: Mixed‐effects logit with village random intercepts; 95\% CIs from Wald SEs. $^{***}p<0.001$, $^{**}p<0.01$, $^{*}p<0.05$. Full tables including duration spline terms are available upon request.
\end{table}

% Requires: \usepackage{booktabs,threeparttable}

%-----------------------------%
% L1. Mixed-effects learning model — fit & random effects
%-----------------------------%
\begin{table}[H]
\centering
\caption{Mixed-effects learning model: fit statistics and random effects}
\label{tab:learn_fit_re}
\begin{threeparttable}
\footnotesize
\begin{tabular}{lrrrrr}
\toprule
 & AIC & BIC & logLik & Deviance & df.resid \\
\midrule
Model & 98175.7 & 98240.2 & -49079.9 & 98159.7 & 23311 \\
\bottomrule
\end{tabular}

\vspace{0.6em}
\textbf{Scaled residuals:} Min = $-5.7907$ \quad 1Q = $-0.4351$ \quad Median = $-0.0081$ \quad 3Q = $0.4548$ \quad Max = $6.3393$.

\vspace{0.6em}
\textbf{Number of obs:} 23{,}319; \textbf{Groups:} player = 2{,}591; village\_code = 134.

\vspace{0.6em}
\textbf{Random effects}
\begin{tabular}{llll}
\toprule
Groups        & Name        & Variance     & Std.Dev. \\
\midrule
player        & (Intercept) & 4.064e{+}00  & 2.016e{+}00 \\
player        & Own Lag     & 6.301e{-}02  & 2.510e{-}01 \\
player        & Peer Lag   & 1.206e{-}01  & 3.473e{-}01 \\
village\_code & (Intercept) & 8.903e{-}14  & 2.984e{-}07 \\
Residual      &             & 2.427e{+}00  & 1.558e{+}00 \\
\bottomrule
\end{tabular}

\end{threeparttable}
\end{table}

%-----------------------------%
% L2. Mixed-effects learning model — fixed effects
%-----------------------------%
\begin{table}[H]
\centering
\caption{Mixed-effects learning model: fixed effects}
\label{tab:learn_fixef}
\begin{threeparttable}
\scriptsize
\begin{tabular}{lrrrrl}
\toprule
Term        & Estimate & Std. Error & df         & $t$ value & $p$ \\
\midrule
(Intercept) & 5.660e{-}01 & 8.811e{-}02 & 3.778e{+}03 &  6.424 & $1.49\times10^{-10}$ *** \\
Own Lag    & 1.368e{-}01 & 8.416e{-}03 & 3.570e{+}03 & 16.256 & $< 2\times10^{-16}$ *** \\
Peer Lag  & 7.439e{-}01 & 1.307e{-}02 & 5.330e{+}03 & 56.936 & $< 2\times10^{-16}$ *** \\
\bottomrule
\end{tabular}
\begin{tablenotes}[flushleft]
\footnotesize
\item Signif. codes: $^{*}p<0.05$, $^{**}p<0.01$, $^{***}p<0.001$.
\end{tablenotes}
\end{threeparttable}
\end{table}

Table \ref{tab:learn_fixef} reports a negative but economically tiny coefficient on ``Average Contribution Others---Previous Round.'' That table is a levels mixed-effects model aimed at heterogeneity (no individual fixed effects $\alpha_i$; no village$\times$round fixed effects $\gamma_{v,t}$), not a learning equation. Without $\alpha_i$ and $\gamma_{v,t}$, the regressor partly loads generic within-village time drift and regression-to-the-mean; it also omits own lag $c_{ig,t-1}$. \\

Once we estimate the canonical updating equation with individual FE, village$\times$round FE, and both lags (Table \ref{tab:learn_fixef}), the expected positive response to lagged peers emerges strongly. We therefore treat Table \ref{tab:learn_fixef} as the interpretative peer-learning result and relegate the Section's \ref{sec:hetero} coefficient to descriptive background.

\section{Supplementary Diagnostics}\label{app:si_diag}

\subsection{Variance across village, group, and player}

We decompose the variance of contributions \(c_{igt}\) using a random‑intercept model,
\[
c_{igt} \;=\; \mu \;+\; u_v \;+\; u_{gv(v)} \;+\; u_{igv} \;+\; \varepsilon_{igt},
\]
with village (\(u_v\)), nested group within village (\(u_{gv(v)}\)), and player (\(u_{igv}\)) effects. Player heterogeneity accounts for roughly two‑thirds of the variance, village differences are small, and the group‑ID level adds essentially nothing once player effects are included.

\renewcommand{\arraystretch}{1.15}
\begin{table}[H]\centering\small
\caption{AD‑1. Variance components for \(c_{igt}\). Shares are \(\text{Var}(\cdot)/\sum \text{Var}(\cdot)\).}
\label{tab:ad1_var}
\begin{tabular}{@{}lrrrr@{}}
\toprule
Level  & Variance & Std.\ Dev. & Share & Share (\%) \\
\midrule
Player \(u_{igv}\)   & 7.426 & 2.725 & 0.623 & 62.3 \\
Group \(u_{gv(v)}\)  & 0.000 & 0.000 & 0.000 & 0.0 \\
Village \(u_{v}\)    & 0.396 & 0.629 & 0.033 & 3.3 \\
Residual \(\varepsilon_{igt}\) & 4.099 & 2.025 & 0.344 & 34.4 \\
\midrule
Total                & 11.921 & 3.453 & 1.000 & 100.0 \\
\bottomrule
\end{tabular}
\end{table}

\noindent The dominance of \(u_{igv}\) motivates modeling individual differences explicitly; there is little additional structure at the group label once player effects are present.

\subsection{Welfare accounting}

Material payoffs in round \(t\) are
\[
\pi^{\text{mat}}_{igt} \;=\; \frac{b}{N}\Big(c_{igt} + \sum_{j\neq i}c_{jgt}\Big) \;-\; \kappa\,c_{igt},
\]
with \(b=2\), \(\kappa=1\), \(N=5\). The observed average payoff per player‑round is \(\approx 6.29\). Under full cooperation \(c_{igt}=12\) for all players, payoffs rise to \(12\). A per‑unit subsidy \(m\) reduces the effective cost to \(\kappa'=\max\{\kappa-m,0\}\); at \(m=0.5\) the implied payoff reaches \(18\), illustrating the leverage of moderate cost relief.

\renewcommand{\arraystretch}{1.15}
\begin{table}[H]\centering\small
\caption{Per‑player welfare across scenarios.}
\label{tab:ad3_welfare}
\begin{tabular}{@{}lrr@{}}
\toprule
Scenario & Payoff per player & Policy parameter \(m\) \\
\midrule
Observed data & 6.29 & 0 \\
Full cooperation & 12.00 & 0 \\
Subsidy policy & 18.00 & 0.5 \\
\bottomrule
\end{tabular}
\end{table}

\subsection{Norm‑fit comparison}

We compare random‑intercept models for \(c_{igt}\) using the lagged peer norm \(\tilde\peer_{-i,g,t-1}\) and degree moderators
(friends, adversaries; standardized within the panel). All specifications include a random intercept for the player;
\(R^2_m\) (marginal) is the variance explained by fixed effects; \(R^2_c\) (conditional) includes random effects.\\

Let $c_{igt}$ denote the contribution of player $i$ in group $g$, village $v$, round $t$,
and let $\tilde\peer_{-i,g,t-1}$ denote the demeaned leave-one-out peer mean in round $t-1$.
Let $F_i$ and $A_i$ be the reported number of friends and adversaries for player $i$,
standardized as $f_i=\mathrm{z}(F_i)$ and $a_i=\mathrm{z}(A_i)$.
All models include a random intercept $u_i\sim \mathcal{N}(0,\sigma_u^2)$
for each player and i.i.d.\ residual $\varepsilon_{igt}\sim \mathcal{N}(0,\sigma^2)$.

\begin{align*}
\text{(mG)}\quad
c_{igt} &= \beta_0
+ \beta_G\,\tilde\peer_{-i,g,t-1}
+ u_i + \varepsilon_{igt}, \\[6pt]
\text{(mGF)}\quad
c_{igt} &= \beta_0
+ \beta_G\,\tilde\peer_{-i,g,t-1}
+ \beta_F f_i
+ \beta_{GF}\,(\tilde\peer_{-i,g,t-1}\cdot f_i)
+ u_i + \varepsilon_{igt}, \\[6pt]
\text{(mGA)}\quad
c_{igt} &= \beta_0
+ \beta_G\,\tilde\peer_{-i,g,t-1}
+ \beta_A a_i
+ \beta_{GA}\,(\tilde\peer_{-i,g,t-1}\cdot a_i)
+ u_i + \varepsilon_{igt}, \\[6pt]
\text{(mGFA)}\quad
c_{igt} &= \beta_0
+ \beta_G\,\tilde\peer_{-i,g,t-1}
+ \beta_F f_i + \beta_A a_i \\
&\quad+ \beta_{GF}\,(\tilde\peer_{-i,g,t-1}\cdot f_i)
+ \beta_{GA}\,(\tilde\peer_{-i,g,t-1}\cdot a_i)
+ u_i + \varepsilon_{igt}.
\end{align*}

\noindent
Table~\ref{tab:ad4_normfit} reports AIC, BIC, marginal $R^2_m$,
conditional $R^2_c$, and RMSE for each specification. All models fit almost
identically ($R^2_c\approx 0.67$); the peer norm $\tilde\peer_{-i,g,t-1}$ carries
the explanatory signal, and degree moderators $f_i,a_i$ add negligible fit once it is included.

\renewcommand{\arraystretch}{1.15}
\begin{table}[H]\centering\small
\caption{Model selection: peer norm with degree moderation.}
\label{tab:ad4_normfit}
\begin{tabular}{@{}lrrrrr@{}}
\toprule
Model & AIC & BIC & \(R^2_m\) & \(R^2_c\) & RMSE \\
\midrule
mGF  (peer\(\times\)friends)        & 105932 & 105980 & 0.00106 & 0.671 & 4.26 \\
mGA  (peer\(\times\)adversaries)    & 105934 & 105982 & 0.00044 & 0.671 & 4.26 \\
mGFA (peer\(\times\)both)           & 105935 & 105999 & 0.00123 & 0.671 & 4.26 \\
mG   (peer only)                    & 105965 & 105997 & 0.00039 & 0.671 & 1.88 \\
\bottomrule
\end{tabular}
\end{table}

\noindent The conditional \(R^2\approx0.67\) is almost entirely due to persistent player heterogeneity; degree
moderation adds negligible explanatory power once \(\tilde\peer_{-i,g,t-1}\) is in the model.

\subsection{Finite‑horizon unravel check}

Grouping sessions by the round‑1 share above the empirical threshold \(\hat c\) (quartiles), we track group means
over \(t=1,\dots,10\). Early high‑share groups remain high on average; early low‑share groups do not converge upward
by the final round, consistent with strategic complementarities rather than full backward‑induction unraveling.

\begin{figure}[H]
  \centering
  \includegraphics[width=\textwidth]{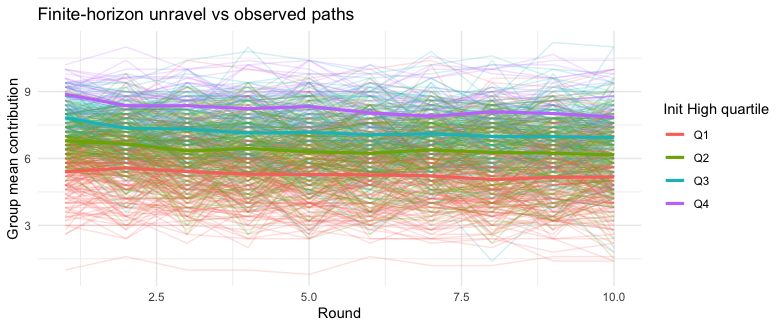}
  \caption{Group mean contributions by round, colored by quartiles of the initial
  high‑share (above \(\hat c\)). Thick lines are within‑quartile averages; faint lines are group‑level paths.}
  \label{fig:ad5_unravel}
\end{figure}

\subsection{Bootstrap of the bifurcation share}

We resample villages with replacement, re‑cluster the full 10‑round trajectories with Ward–D2 at two clusters $K$=2),
and record the share assigned to the Low‑path cluster. The distribution centers slightly above one‑half
with substantial mass between \(\approx0.3\) and \(\approx0.65\), confirming that the Low basin is not a small‑sample artifact.

\begin{figure}[H]
  \centering
  \includegraphics[width=\textwidth]{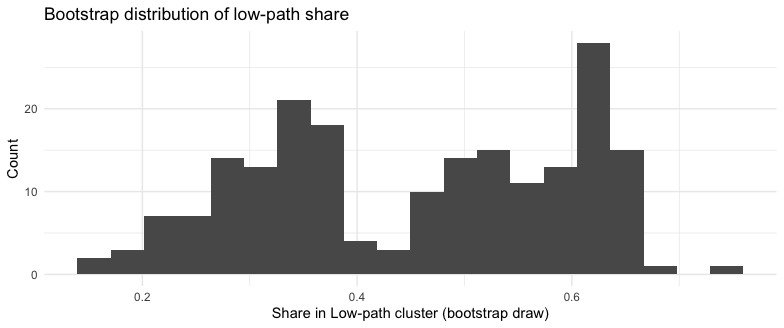}
  \caption{Each draw resamples villages with replacement, reclusters full 10‑round
  trajectories using Ward–D2 on round‑wise $z$–scores (at two clusters $K$=2), and tags the
  cluster with the lower mean as the Low path. The distribution is centered a
  bit above one-half with substantial mass in the $0.30$–$0.65$ range,
  indicating that a sizeable Low basin is not a small‑sample artifact.}
  \label{fig:ad6_boot}
\end{figure}

\subsection{Threshold sensitivity}

We vary the High/Low cutoff \(c^{\mathrm{thr}}\in\{5,6,7\}\) (Lempiras) when predicting High end‑states at \(t=10\).
Intercepts become more negative as the cutoff increases; the autonomy coefficient rises slightly, the friends
coefficient declines modestly, and gender effects taper with higher thresholds. Qualitative conclusions are stable.

\begin{figure}[H]
  \centering
  \includegraphics[width=\textwidth]{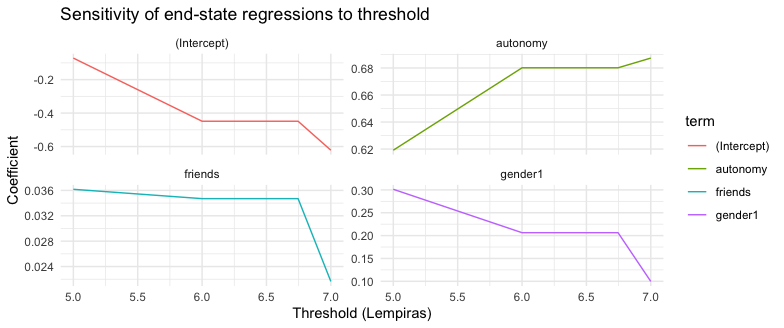}
  \caption{\textbf{Sensitivity of end-state predictors to the High/Low cutoff.}
  Coefficients from logits of finishing High ($c_{i,10}>\text{threshold}$) on
  round‑1 covariates evaluated at thresholds $5,6,7$ Lempiras.
  Intercepts become more negative at higher cutoffs; the autonomy effect rises
  modestly; the friends effect declines slightly; and the male indicator
  attenuates. Qualitative conclusions on ``who ends High'' are robust to the
  cutoff.}
  \label{fig:ad7_thresh}
\end{figure}

\subsection{Conceptual flow}
A simple flow connects observed traits (friends, education, gender, religion) to behavioral parameters \((d_i,\phi_i,\omega_i)\), to partial-adjustment dynamics, to basin membership, to village outcomes, with policy levers feeding back through incentives. This clarifies where measurement and interventions bite in the pipeline.

\begin{figure}[H]
  \centering
  \includegraphics[width=\textwidth]{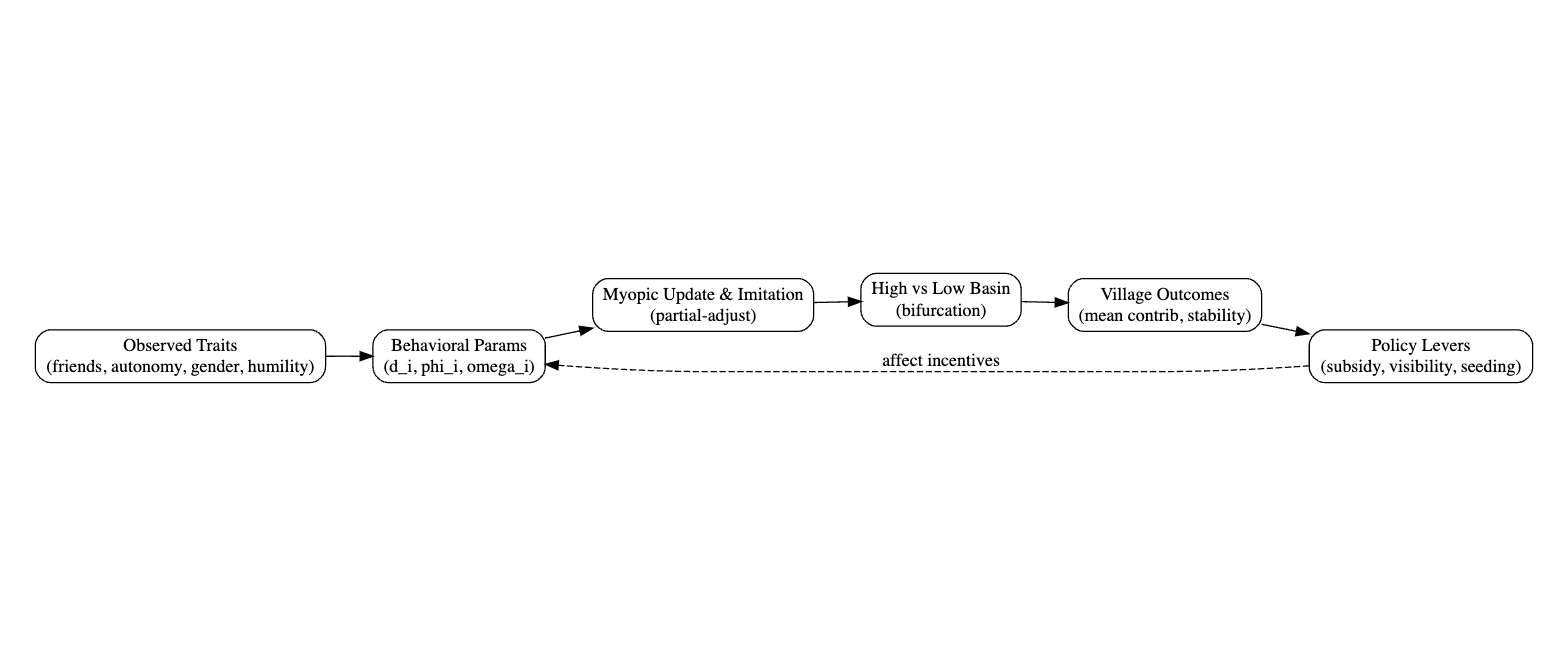}
  \caption{Observed traits $\to$ behavioral parameters $(d_i,\phi_i,\omega_i)\to$ myopic
  update \& imitation $\to$ basin membership $\to$ village outcomes, with policy
  levers feeding back via incentives.}
  \label{fig:ad12_flow}
\end{figure}

\section{Background on Honduras Social Network RCT  Cohort and Field Infrastructure}

\subsection{Study Region}
The study region covered more than 200 square miles of mountainous terrain near the Guatemala border, with an estimated population of roughly 92,000 residents. Out of 238 identified towns and villages, 176 were selected on the basis of population size, accessibility, and safety. The adult population in these study villages was estimated at 32,800 individuals.

All persons aged 12 years or older, living or working in the study villages, were eligible at baseline. Individuals unable to provide consent due to cognitive impairment were excluded.

\subsection{Recruitment}
Prior to launching fieldwork, we mapped settlements in the region to capture terrain, rainfall patterns, and distances to health facilities. A photographic census was then conducted using tablet-based software. Photos, GPS coordinates, and demographic details (age, sex, marital status) were recorded for all residents. Out of an estimated 32,800 eligible individuals, 93\% (30,422) consented to be included in the census. Village-level participation ranged from 55 to 627 respondents, with households averaging 2.8 participants.

Key descriptive characteristics were as follows: mean village size was 173 individuals (range 55–627), mean household size was 2.8 (range 1–13), and mean age was 33 (range 12–94). Fifty-four percent were women and 58\% reported being married or living as married.

\subsection{Study Design}

The study region consisted of an area of over 200 square miles of rugged mountainous terrain near the Guatemala border with an estimated total population of 92,000 people. We selected 176 villages from these 238 small towns and villages in this area. Factors like population size, accessibility, and safety were considered when selecting the final list of villages. The included 176 villages had an estimated adult population of 32,800, according to local Honduras government census data.

Owing to high adolescent birth rates in this population, all individuals over the age of 12 who lived or worked in the study villages were eligible to enroll at baseline in 2015. Individuals who were cognitively impaired and unable to provide consent at baseline were excluded.

Recruitment preparation for this project included comprehensive geographical mapping of the 238 small towns and villages (cities were excluded), which allowed for a more precise understanding of terrain, rainfall, and access to health facilities. A photographic census was then conducted across the 176 study villages, capturing photographs, GPS coordinates, and basic demographics (age, gender, marital status). Of the approximately 32,800 eligible individuals, 93\% (N=30,422) consented to be censused. Respondent counts per village ranged from 55 to 620, and the average household size was 2.8.

\begin{table}[H]
\centering
\caption{Baseline Census Demographics, $N=30{,}422$}
\label{tab:census_demo}
\begin{tabular}{lcc}
\toprule
Characteristic & Mean & Range \\
\midrule
Village size & 173 & [55 -- 627] \\
Household size & 2.8 & [1 -- 13] \\
Age (years) & 33 & [12 -- 94] \\
Women & 54\% & --- \\
Married / living as married & 58\% & --- \\
\bottomrule
\end{tabular}
\end{table}

\begin{table}[H]
\centering
\caption{Baseline Cohort Characteristics, $N=24{,}812$}
\label{tab:cohort_stats}
\begin{tabular}{lc}
\toprule
Characteristic & Value  \\
\midrule
Age (years) & 33  [12 -- 93] \\
Women & 58\%  \\
Married / living as married & 58\% \\
Less than primary education & 70\% \\
Religion: Catholic & 51\% \\
Religion: Protestant & 32\%  \\
No religion & 16\%  \\
Indigenous & 12\%  \\
Self-reported general health fair/poor & 44\%  \\
Self-reported mental health fair/poor & 40\%  \\
\bottomrule
\end{tabular}
\end{table}

\subsection{Study Design}
After the photographic census, we mapped complete sociocentric networks in all 176 villages, attempting to enumerate all reported ties within each community. Villages were then assigned in a $2 \times 8$ factorial design: (i) eight saturation levels indicating the fraction of households targeted in each village (0, 0.05, 0.10, 0.20, 0.30, 0.50, 0.75, 1.00), crossed with (ii) two targeting strategies (random household selection versus a ``friend-of-a-random'' nomination procedure).

To improve balance, villages were grouped into 11 blocks of 16, minimizing within-block variation in household counts and average subjects per household. Within each block, one village was allocated to each of the 16 factorial cells.

In the random targeting arm, the number of households implied by the dosage was selected randomly without replacement. For example, in a 37-household village assigned a 20\% dosage, seven households were chosen at random. In the friend nomination arm, the number of seed households implied by the dosage was sampled first. For each seed, one member was randomly chosen, and a randomly selected social contact from a different household was recruited. If no eligible contacts were available, or if a nominated household had already been treated, reselection was performed until the target sample was reached.

\subsection{Timeline of Data Collection Waves}

The timeline is illustrated in \ref{fig:rct_study}. We describe each wave:

\begin{itemize}
    \item Wave 0 (Jun 2015–Dec 2015): Photographic census of 176 villages; 93\% coverage of residents aged 12+ ($N = 30{,}422$).
    \item Wave 1 (Oct 2015–Jun 2016): Baseline survey with sociocentric network mapping, health attitudes, and behaviors ($N \approx 24{,}700$).
    \item Wave 2 (Jan 2018–Aug 2018): Interim survey 12 months into intervention; updated locations and statuses ($N \approx 21{,}500$).
    \item Wave 3 (Jan 2019–Dec 2019): New census of adults 15+, followed by endline survey with networks, health, and behaviors ($N \approx 22{,}500$).
    \item Wave 4 (Aug 2022–Jul 2023): New census in a purposive subset of 82 villages and follow-up survey, adding domains on flourishing, ego-network change, and mobility in older adults ($N = 10{,}941$).
\end{itemize}

\begin{figure}[H]
  \centering
  \includegraphics[scale=0.9]{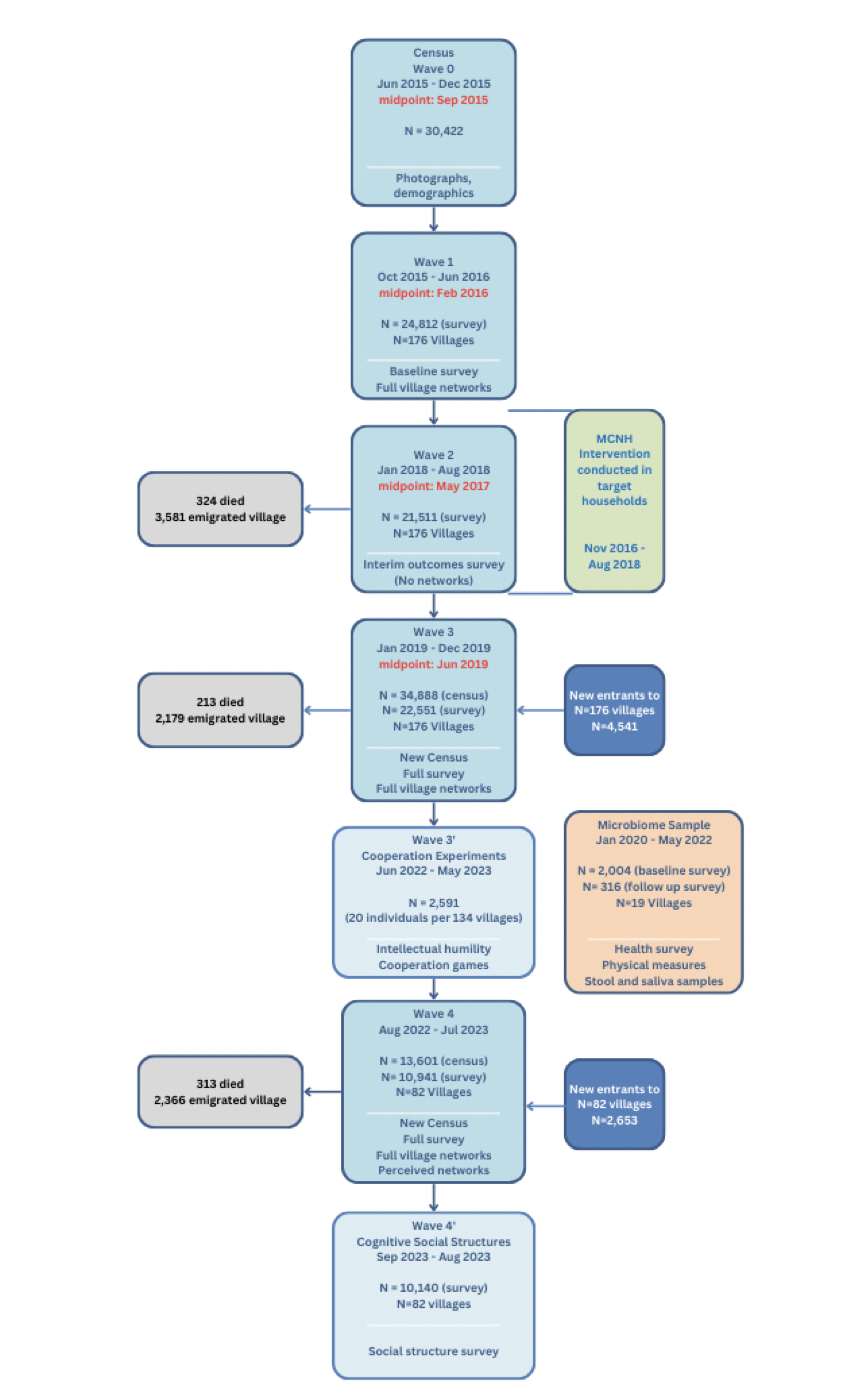}
  \caption{Waves of Data Collection}
  \label{fig:rct_study}
\end{figure}

\subsection{Field Operations}
More than 100 local enumerators were trained to (i) build infrastructure and conduct preliminary data collection, (ii) recruit participants and conduct census enumeration, and (iii) administer surveys and interviews. Four field offices were established in strategic locations to minimize travel time. Each office was equipped with tablets, netbooks, printers, high-speed internet, and local servers for secure data transfer and synchronization. Enumerators received intensive training on instruments and software by U.S.-based project managers and were supervised daily by Honduran coordinators.

\subsection{Community Engagement}
Partnerships were developed with municipal and health authorities, and meetings were held with health staff and community health workers. The Ministry of Health reviewed and approved all study protocols, consent processes, and instruments. Village leaders and indigenous councils were engaged in advance to present study objectives before recruitment and data collection.

\subsection{Missing Data and Attrition}
Among the $\sim$32,800 eligible individuals, 93\% consented to participate, reducing baseline non-response bias. We attempted to recontact and track participants who moved within the region during follow-ups. Of 5,633 randomized households, 4,861 (86\%) had at least one respondent complete the endline survey. Attrition did not differ by treatment status ($\chi^2(1, N=5633) = 0.80$, $p = 0.37$).

\section{Inclusion Criteria}
Respondents were included in analyses if they met the following criteria: (a) they completed all applicable survey forms, (b) they were enrolled in the census and participated in both Wave 1 and Wave 3 surveys, and (c) their household and village identifiers remained consistent across these waves. The same inclusion rules were applied for mediation and robustness checks. Specifically, individuals included in mediation analyses at Wave 2 must have completed both census and survey instruments and be present in Waves 1–3. For robustness analyses at Wave 4, individuals must have completed census and survey instruments and be consistently present in Waves 1–4.

\subsection{Cooperation Experiments}
Between June 2022 and June 2023, we conducted simultaneous studies in 134 villages, involving 2,591 participants.

\subsubsection{Design}
The study was carried out to evaluate village-level social capital and collective efficacy through cooperation experiments.

\subsubsection{Data Collection}
Forty residents were randomly sampled in each village. Of these, 15–20 participated in the cooperation game ($N=2{,}591$). Public goods games were conducted in person: participants were convened at accessible village locations, seated anonymously in groups of five, and provided with Android tablets. They were trained to use the interface before beginning. Participants saw only avatars on their screens (including their own) and could not identify or interact with others.

Each participant received 12 Lempiras per round and decided how much to keep versus contribute to a shared pool. Contributions were doubled and distributed equally across group members. Ten rounds were conducted, though the total number was not disclosed in advance. Afterward, participants kept their earnings plus a 50 Lempira participation fee (roughly a day wage).

\end{document}